\documentclass[12pt]{article}
\usepackage{makeidx}
\usepackage{natbib,amsfonts,amsmath,verbatim}
\usepackage{amssymb}
\usepackage{graphicx, caption, subfig}
\usepackage{multirow,rotating}
\usepackage{theorem}
\usepackage{fullpage}
\usepackage{xr}
\usepackage [displaymath,mathlines]{lineno}

\usepackage{float}
\usepackage{rotfloat}

\newtheorem{assumption}{Assumption}
\newtheorem{sscheme}{Sampling Scheme}
\newtheorem{algorithm}{Algorithm}

\newcommand{\bs}{\boldsymbol}

\def\wh{\widehat }
\def\wt{\widetilde }
\def\({\left (}
\def\){\right )}

\input{tcilatex}

\def\transp{\tiny T}
\newcommand*\patchAmsMathEnvironmentForLineno[1]{%
  \expandafter\let\csname old#1\expandafter\endcsname\csname #1\endcsname
  \expandafter\let\csname oldend#1\expandafter\endcsname\csname end#1\endcsname
  \renewenvironment{#1}%
     {\linenomath\csname old#1\endcsname}%
     {\csname oldend#1\endcsname\endlinenomath}}%
\newcommand*\patchBothAmsMathEnvironmentsForLineno[1]{%
  \patchAmsMathEnvironmentForLineno{#1}%
  \patchAmsMathEnvironmentForLineno{#1*}}%
\AtBeginDocument{%
\patchBothAmsMathEnvironmentsForLineno{equation}%
\patchBothAmsMathEnvironmentsForLineno{align}%
\patchBothAmsMathEnvironmentsForLineno{flalign}%
\patchBothAmsMathEnvironmentsForLineno{alignat}%
\patchBothAmsMathEnvironmentsForLineno{gather}%
\patchBothAmsMathEnvironmentsForLineno{multline}%
}

\begin{document}
\title{A flexible Particle Markov chain Monte Carlo method}
\author{Eduardo F. Mendes \\
School of Applied Mathematics \\
Funda\c{c}\~{a}o Getulio Vargas \and \and Christopher K. Carter \\
School of Economics \\
University of New South Wales \and David Gunawan \\ School of Economics \\ University of New South Wales \and Robert Kohn \\
School of Economics \\
University of New South Wales }
\date{}
\maketitle

\begin{abstract}
Particle Markov Chain Monte Carlo methods are used to carry out inference in non-linear and non-Gaussian state space models,
where the posterior density of the states is approximated using particles.
Current approaches usually perform Bayesian inference using either a particle Marginal Metropolis-Hastings (PMMH) algorithm or a particle Gibbs (PG) sampler.
This paper shows how the two ways of generating variables mentioned above can be combined in a flexible manner to give sampling schemes that converge to a desired target distribution.
The advantage of our approach is that the sampling scheme can be tailored to obtain good results for different applications.
For example, when some parameters and the states are highly correlated, such parameters can be generated using PMMH, while all other parameters are generated using PG
because it is easier to obtain good proposals for the parameters within the PG framework.
We derive some convergence properties of our sampling scheme and also investigate its  performance empirically  by applying it to univariate and multivariate stochastic volatility
models and comparing it to other PMCMC methods proposed in the literature.
\end{abstract}


\section{Introduction\label{S: Introd}}

Our article deals with statistical inference for both the unobserved states
and the parameters in a class of state space models. Its main goal is to
give a flexible approach to constructing sampling schemes that converge to
the posterior distribution of the states and the parameters. The sampling
schemes generate particles as auxiliary variables. This work extends the
methods proposed by \cite{andrieuetal2010}, \cite{olssonryden2011}, \cite{lindstenschon2012a}, \cite%
{lindstenetal2014}, \citet{Fearnhead2016}, and \citet{Deligiannidis2018}.

\cite{andrieuetal2010} introduce two particle Markov chain Monte Carlo (MCMC)  methods for state space models.
The first is particle marginal Metropolis-Hastings (PMMH),
where the parameters are
generated with the states integrated out. The second is particle Gibbs (PG),
which generates the parameters given the states. They show that the
augmented density targeted by this algorithm has the joint posterior density
of the parameters and states as a marginal density. \cite{andrieuetal2010}
and \cite{andrieuroberts2009} show that the law of the marginal sequence of
parameters and states, sampled using either PG or PMMH, converges to the
true posterior as the number of iterations increase. Both particle MCMC
methods are the focus of recent research. \cite{olssonryden2011} and \cite%
{lindstenschon2012a} use \textit{backward simulation} \citep{godsilletal2004}
for sampling the state vector, instead of \textit{ancestral tracing} %
\citep{kitagawa1996}.
\cite{lindstenschon2012a} extend the PG sampler to a particle Metropolis within Gibbs (PMwG)
sampler to deal with the case where the parameters cannot be generated exactly conditional on the states.
\citet{Fearnhead2016} proposed an augmented particle MCMC methods. They show that their method can improve the mixing of the particle Gibbs when the parameters are highly correlated with the states.
Recently, \citet{Deligiannidis2018} proposed the correlated pseudo marginal Metropolis-Hastings method that significantly reduce the number of particles used by the standard pseudo marginal method. 
Unless stated otherwise, we write PG to denote both the PG and PMwG samplers that generate the parameters conditional on the states.

We note that there are no formal results in the literature to guide the user on whether to use PMMH or PG for any given problem.
Our work extends the particle MCMC framework to
situations where using just PMMH or just PG is inefficient.
It is well-known from the literature on Gaussian and conditionally Gaussian state space models that confining MCMC for state space models to Gibbs sampling or Metropolis-Hastings sampling can result in inefficient or even degenerate sampling. See, for example, \cite{kimetal1998} who show for a stochastic volatility model that generating the states conditional on the parameters and the parameters conditional on the states can result in a highly inefficient sampler. See also \cite{carterkohn1996} and \cite{gerlachetal2000} who demonstrate using a signal plus noise model that a Gibbs sampler for the states and indicator variables for the structural breaks produces a degenerate sampler. A natural solution is to combine Gibbs and Metropolis-Hastings samplers. Motivated by that, we
derive a particle sampler on the same augmented space as the PMMH and
PG samplers, in which some parameters are sampled conditionally on the
states and the remaining parameters are sampled with the states integrated
out. We call this a PMMH+PG sampler.
We show that the PMMH+PG sampler targets the same augmented density as
the PMMH or PG samplers. We provide supplementary material showing that the
Markov chain generated by the algorithm is uniformly ergodic, given
regularity conditions. It implies that the marginal law of the Markov chain
generated by $n^{th}$ iteration of the algorithm converges to the posterior
density function geometrically fast, uniformly on its starting value, as $%
n\rightarrow \infty $.

We use \textit{ancestral tracing} in the particle Gibbs step to
make the presentation accessible. The
online supplementary material shows how to modify the
methods proposed in the paper to incorporate auxiliary particle filters and \textit{backward simulation} in the particle Gibbs step.
The same convergence results for the latter  methods are obtained by modifying
the arguments in \cite{olssonryden2011}.

We apply our PMMH+PG sampler to several univariate and multivariate examples using simulated and real datasets.
As a main application we propose a general algorithm for Bayesian inference on a multivariate factor stochastic volatility (SV) model.
This model is used to jointly model many co-varying financial time series,
as it is able to capture the common features using only a small number of latent factors (see, e.g. \cite{Chib2006} and \cite{Kastner:2017}).
We consider a factor SV model in which the volatilities of the factors follow a traditional SV model (as in \cite{Chib2006} and \cite{Kastner:2017})
and the log-volatilities of the idiosyncratic errors follow either a continuous time Ornstein-Uhlenbeck (OU) process \citep{Stein1991} or a GARCH diffusion process \citep{Chib2004,  Kleppe2010}.
The OU process admits a closed form transition density whereas the GARCH process does not.
Similar factor models can also be applied to spatial temporal data with a large number of spatial measurements at each time point.

We use these examples to compare the performance of our sampling schemes to the standard PMMH and PG samplers of \cite{andrieuetal2010},
the particle Gibbs with data augmentation sampler of \citet{Fearnhead2016}, and the correlated PMMH of \citet{Deligiannidis2018}.
For the standard and correlated PMMH, we consider adaptive random walk proposals and the refined proposals by \citet{dahlin2015} and \citet{Nemeth2016}.
We show that the PMMH + PG sampler outperforms these methods in the situation where we have both a large number of parameters and a large number of latent states.
In general, there are likely to be a number of different sampling schemes that can solve the same problems addressed in our article,
and which sampler is best depends on a number of factors such as the model, the data set and  the number of observations.
We also note that our PMMH + PG approach can be further refined by using the data augmented PMMH and PG sampling schemes proposed
by \citet{Fearnhead2016} and the refined proposals for the PMMH sampling scheme by \citet{dahlin2015} and \citet{Nemeth2016}.

The rest of the paper is organized as follows. Section \ref{s:prelim} introduces the
basic concepts and notation used throughout the paper as well as the PMMH+PG sampler for estimating
a single state space model and its associated parameters.
Sections \ref{SSS: cts time OU process} and  \ref{PMMH+PG SV factor} compare the performance of the PMMH+PG sampler to other competing PMCMC methods for estimating univariate and multivariate stochastic volatility models,
respectively. The paper has an online supplement which contains some further empirical and technical results.


\section{The PMMH+PG sampling scheme for state space
models\label{s:prelim}}
This section introduces a sampling scheme that combines PMMH and PG steps for the
Bayesian estimation of a state space model. The first three sections
give preliminary results and Section~\ref{s:pmwg} presents the sampling scheme.
The methods and models introduced in
this section are used in the univariate models in
Section~\ref{SSS: cts time OU process} and the multivariate models in Section~\ref{PMMH+PG SV factor}.

\subsection{State space model\label{s:SSM}}
Define $\mathbb{N}$ as the set of positive integers and let $\{X_{t}\}_{t\in
\mathbb{N}}$ and $\{Y_{t}\}_{t\in \mathbb{N}}$ denote $\mathcal{X}$-valued
and $\mathcal{Y}$-valued stochastic processes, where $\{X_{t}\}_{t\in
\mathbb{N}}$ is a latent Markov process with initial density $f_1^{\theta
}(x)$ and transition density $f_t^{\theta }(x^{\prime }|x)$, i.e.,
\begin{equation*}
X_{1}\sim f_1^{\theta }(\cdot )\quad \mbox{and}\quad X_{t}|(X_{t-1}=x)\sim
f_t^{\theta }(\cdot |x)\quad (t=2,3,\dots ).
\end{equation*}%

The latent process $\{X_{t}\}_{t\in \mathbb{N}}$ is observed only through
$\{Y_{t}\}_{t\in \mathbb{N}}$, whose value at time $t$ depends on the value
of the hidden state at time $t$, and is distributed according to $g_t^{\theta
}(y|x)$
\begin{equation*}
Y_{t}|(X_{t}=x)\sim g_t^{\theta }(\cdot |x)\quad (t=1,2,\dots ).
\end{equation*}%
The densities $f_t^{\theta }$ and $g_t^{\theta }$ are indexed
by a parameter vector $\theta \in \Theta $, where $\Theta $ is an open
subset of $\mathbb{R}^{d_{\theta }}$, and all densities are with respect to
suitable dominating measures, denoted as $dx$ and $dy$.
The dominating measures are frequently taken to be the Lebesgue measure if $%
\mathcal{X}\in \mathcal{B}(\mathbb{R}^{d_{x}})$ and $\mathcal{Y}\in \mathcal{%
B}(\mathbb{R}^{d_{y}})$, where $\mathcal{B}(A)$ is the Borel $\sigma $%
-algebra generated by the set $A$. Usually $\mathcal{X} = \mathbb{R}^{d_x}$ and $\mathcal{Y} = \mathbb{R}^{d_y}$.

We use the colon notation for collections of random variables, i.e., 
$a_{t}^{1:N}=\left( a_{t}^{1},\dots ,a_{t}^{N}\right) $ and for $t\leq u$, $
a_{t:u}^{1:N}=\left( a_{t}^{1:N},\dots ,a_{u}^{1:N}\right) $. 
The joint probability density function of $\left( x_{1:T},y_{1:T}\right) $ is
\begin{equation*}
p\left( x_{1:T},y_{1:T}|\theta \right) =f_1^{\theta }(x_{1})g_1^{\theta
}(y_{1}|x_{1})\,\prod_{t=2}^{T}f_t^{\theta }(x_{t}|x_{t-1})\,g_t^{\theta
}(y_{t}|x_{t}).
\end{equation*}%
We define $Z_{1}(\theta ):=p(y_{1}|\theta )$ and $Z_{t}(\theta
):=p(y_{t}|y_{1:t-1},\theta )$ for $t\geq 2$, so the likelihood is $%
Z_{1:T}\left( \theta \right) =Z_{1}(\theta )\times Z_{2}(\theta )\ldots
Z_{T}(\theta )$. The joint filtering density of $X_{1:t}$ is
\begin{equation*}
p\left( x_{1:t}|y_{1:t},\theta \right) =\frac{p\left( x_{1:t},y_{1:t}|\theta
\right) }{Z_{1:t}\left( \theta \right) }.
\end{equation*}%
The posterior density of $\theta $ and $X_{1:T}$ can also be factorized as
\begin{equation*}
p(x_{1:T},\theta |y_{1:T})=\frac{p(x_{1:T},y_{1:T}|\theta )p(\theta )}{%
\overline{Z}_{1:T}},
\end{equation*}%
where the marginal likelihood $\overline{Z}_{T}=\int_{\Theta
}Z_{1:T}\left( \theta \right) \,p(\theta )\,\mathrm{d}\theta =p(y_{1:T})$.
This factorization is used in the particle Markov chain Monte Carlo
algorithms.

\subsection{Target distribution for state space models\label{s:targetdist}}

We first approximate the joint filtering densities $\{p(x_{t}|y_{1:t},\theta):\,t=1,2,\dots \}$ sequentially, using particles, i.e., weighted samples, $(x_{t}^{1:N},\bar{w}_{t}^{1:N})$,
drawn from auxiliary distributions $m_{t}^{\theta }$.
This requires specifying \textit{importance densities} $m_{1}^{\theta }(x_{1}):=m_{1}(x_{1}|Y_{1}=y_{1},\theta )$ and $m_{t}^{\theta}(x_{t}|x_{t-1}):=m_{t}(x_{t}|X_{t-1}=x_{t-1},Y_{1:t}=y_{1:t},\theta )$,
and a resampling scheme $\mathcal{M}(a_{t-1}^{1:N}|\bar{w}_{t-1}^{1:N})$, where each $a_{t-1}^{i}=k$ indexes a particle in $(x_{t-1}^{1:N},\bar{w}_{t-1}^{1:N})$, and is sampled with probability $\bar{w}_{t-1}^{k}$.
We refer to \cite{doucetetal2000}, \cite{merveetal2001}, and \cite{guoetal2005} for the choice of importance densities and \cite{doucetal2005} for a comparison between resampling schemes.
Unless stated otherwise, upper case letters indicate random variables and lower case letters indicate the corresponding values of these random variables, e.g., $A_{t}^{j}$ and $a_{t}^{j}$, $X_{t}$ and $x_{t}$.
We denote the vector of particles by
\begin{equation}\label{eq:defineU}
U_{1:T}:=\left( X_{1}^{1:N},\ldots ,X_{T}^{1:N},A_{1}^{1:N},\ldots
,A_{T-1}^{1:N}\right)
\end{equation}
where $a_{t}^{j}$ is the value of the random variable $A_{t}^{j}$ and its sample space by $\mathcal{U}:=\mathcal{X}^{TN}\times \mathbb{N}^{(T-1)N}$.

The Sequential Monte Carlo (SMC) algorithm used here is the same one as in Section 4.1 of \cite{andrieuetal2010}, and is defined in Section~\ref{s:algorithms} and Algorithm~\ref{alg:smc} in the supplementary material.
The algorithm provides an unbiased estimate
\begin{align*}\widehat{Z}_{T}\left(\theta\right)=Z(u_{1:T},\theta ):=\prod_{t=1}^{T}\left(N^{-1}\sum_{i=1}^{N}w_{t}^{i}\right),
 \end{align*}
of the likelihood, where
\begin{align*}
w_1^i & = \frac{f_1^\theta(x_1^i)g_1^\theta(y_1|x_1^i) }{m_1^\theta(x_1^i)} , w_t^i = \frac{g_t^\theta(y_t|x_t^i)f_t^\theta(x_t^i|x_{t-1}^{a_{t-1}^i})} {m_t^\theta(x_t^i|x_{t-1}^{a_{t-1}^i})}
\,\,\,\text{for} \,\,\, t =2,\dots, ,T, \,\,\, \text{and} \,\,\, {\overline w}_t^i = \frac{w_t^i}{\sum_{j=1}^N w_t^j}.
\end{align*}

The joint distribution of the particles given the parameters is
\begin{eqnarray}
{\psi \left( u_{1:T}|\theta \right) :=\prod_{i=1}^{N}m_{1}^{\theta} \left( x_{1}^{i}\right)
\prod_{t=2}^{T}\left\{ \mathcal{M}(a_{t-1}^{1:N}|\bar{w}_{t-1}^{1:N})%
\prod_{i=1}^{N}m_{t}^{\theta }\left( x_{t}^{i}|x_{t-1}^{a_{t-1}^{i}}\right)
\right\}}.  \label{eq:pfdens}
\end{eqnarray}

The key idea of particle MCMC methods is to construct a target distribution on an augmented space that includes the particles $U_{1:T}$ and has a marginal distribution equal to $p(x_{1:T},\theta |y_{1:T})$.
This section describes the target distribution from \cite{andrieuetal2010}. Later sections describe particle MCMC methods to sample from this distribution and hence sample from $p(x_{1:T},\theta |y_{1:T})$.
Section~\ref{appendixbsi} of the supplementary material describes other choices of target distribution and how it is straightforward to modify our results to apply to them.

The simplest way of sampling from the particle approximation of $p(x_{1:T}|y_{1:T},\theta)$ is called \textit{ancestral tracing}. It was introduced in \cite{kitagawa1996} and used in \cite{andrieuetal2010}
and consists of sampling one particle from the final particle filter.
The method is equivalent to sampling an index $J=j$ with probability $\bar w_T^{j}$, tracing back its ancestral lineage $b_{1:T}^j$ ($b_T^j = j$ and $b_{t-1}^j=a_{t-1}^{b_t^j}$)
and choosing the particle $x_{1:T}^j = (x_1^{b_1^j},\dots,x_T^{b_T^j})$.

With some abuse of notation, for a vector $a_{t}$, denote $a_{t}^{(-k)}=\left( a_{t}^{1},\dots ,a_{t}^{k-1},a_{t}^{k+1},\dots,a_{t}^{N}\right) $, with obvious changes for $k\in \{1,N\}$, and denote
\begin{equation*}
u_{1:T}^{(-j)}=\left\{ x_{1}^{(-b_{1}^{j})},\ldots ,x_{T-1}^{(-b_{T-1}^{j})},x_{T}^{(-j)},a_{1}^{(-b_{1}^{1})},\ldots,a_{T-1}^{(-b_{T-1}^{j})}\right\} .
\end{equation*}%
It simplifies the notation to sometimes use the following one-to-one transformation
\begin{equation*}
\left( u_{1:T},j\right) \leftrightarrow \left\{x_{1:T}^{j},b_{1:T-1}^{j},j,u_{1:T}^{(-j)}\right\} ,
\end{equation*}%
and switch between the two representations and use whichever is more convenient. Note that the right hand expression will sometimes be written as $\left\{ x_{1:T},b_{1:T-1},j,u_{1:T}^{(-j)}\right\} $ without ambiguity.

We now assume Assumptions \ref{assu:propstatespace} and \ref{assu:resampling}, given in Section~\ref{s:algorithms} of the online supplement. The target distribution from \cite{andrieuetal2010} is
\begin{align}
\tilde{\pi}^{N}\left( x_{1:T},b_{1:T-1},j,u_{1:T}^{(-j)},\theta
\right) \mathrel{:=}
\frac{p(x_{1:T},\theta |y_{1:T})}{N^{T}}\frac{\psi \left( u_{1:T}|\theta
\right) }{m_{1}^{\theta }\left( x_{1}^{b_{1}}\right) \text{ }\prod_{t=2}^{T}%
\bar{w}_{t-1}^{a_{t-1}^{b_{t}}}m_{t}^{\theta }\left(
x_{t}^{b_{t}}|x_{t-1}^{a_{t-1}^{b_{t}}}\right) },  \label{eq:targetdist}
\end{align}%
where $u_{1:T}$ is given in Eq. \eqref{eq:defineU}. Assumption \ref{assu:propstatespace} ensures that $\tilde{\pi}^{N}\left(u_{1:T}|\theta \right) $ is absolutely continuous with respect to $\psi\left( u_{1:T}|\theta \right) $,
so that $\psi \left( u_{1:T}|\theta \right) $ can be used as a Metropolis-Hastings proposal
density for generating from $\tilde{\pi}^{N}\left( u_{1:T}|\theta \right)$.

From Assumption \ref{assu:resampling}, Eq. (\ref{eq:targetdist}) has the following marginal distribution
\begin{equation}
\tilde{\pi}^{N}\left( x_{1:T},b_{1:T-1},j,\theta \right) =\frac{p(x_{1:T},\theta |y_{1:T})}{N^{T}}, \label{eq:margdist}
\end{equation}%
and hence $\tilde{\pi}^{N}\left( x_{1:T},\theta \right) =p(x_{1:T},\theta |y_{1:T})$. The online supplement gives further details.

\subsection{Conditional sequential Monte Carlo (CSMC)\label{s:CSMC}}

The particle Gibbs algorithm in \cite{andrieuetal2010} uses exact conditional distributions to construct a Gibbs sampler.
If we use the \textit{ancestral tracing }augmented distribution given in (\ref{eq:targetdist}),
then this includes the conditional distribution given by $\tilde{\pi}^{N}\left( u_{1:T}^{(-j)}|x_{1:T}^{j},b_{1:T-1}^{j},j,\theta \right)$,
which involves constructing the particle approximation conditional on a pre-specified path.
The \textit{conditional sequential Monte Carlo} algorithm, introduced in \cite{andrieuetal2010},
is a sequential Monte Carlo algorithm in which a particle $X_{1:T}^{J}=(X_{1}^{B_{1}^{J}},\dots,X_{T}^{B_{T}^{J}})$, and the associated sequence of ancestral indices $B_{1:T-1}^{J}$ are kept unchanged.
In other words, the conditional sequential Monte Carlo algorithm is a procedure that resamples all the particles and indices except for
$U_{1:T}^{J}=(X_{1:T}^{J},A_{1:T-1}^{J})=(X_{1}^{B_{1}^{J}},\dots,X_{T}^{B_{T}^{J}},B_{1}^{J},\dots ,B_{T-1}^{J})$.
Algorithm~\ref{alg:condsmc} of the supplementary material describes the conditional sequential Monte Carlo algorithm (as in \cite{andrieuetal2010}), consistent with $(x_{1:T}^{j},a_{1:T-1}^{j},j)$.

\subsection{Flexible sampling scheme for state space models\label{s:pmwg}}

This section introduces a sampling scheme that is suitable for the state space form given in Section~\ref{s:SSM},
where some of the parameters can be generated exactly conditional on the state vectors using PG step,
but other parameters must be generated using PMMH step. For simplicity,
let $\theta :=(\theta _{1},\theta _{2})$ be a partition of the parameter vector into $2$ components where each component may be a vector.
Let $\Theta =\Theta_{1}\times\Theta _{2}$ be the corresponding partition of the parameter space.
The following sampling scheme generates the vector of parameter $\theta _{1}$ using PMMH step and the vector of parameter $\theta _{2}$ using PG step.
We call this a PMMH+PG sampler.
It is important to note that the components in the parameter vector $\theta _{1}$ can be sampled separately in multiple PMMH steps and the components in the parameter vector $\theta _{2}$
can be sampled separately in multiple Gibbs steps.
Details are given in Section~\ref{s:theory} in the online supplement.


\begin{sscheme}[PMMH+PG Sampler]
\label{ssch:pmwg}
Given initial values for $U_{1:T}$, $J$ and $\theta$, one iteration of the MCMC involves the following steps.

\begin{enumerate}
\item (PMMH sampling) 

\begin{enumerate}
\item Sample $\theta _{1}^{\ast }\sim q_{1,1}(\cdot |U_{1:T}, J, \theta_{2},\theta _{1}).$
\item Sample $U_{1:T}^{\ast }\sim \psi(\cdot |\theta_{2},\theta _{1}^{\ast }).$
\item Sample $J^{\ast } \sim \tilde{\pi}^N (\cdot |U_{1:T}^{\ast}, \theta _{2},\theta _{1}^{\ast }).$
\item Set $(\theta _{1}, U_{1:T}, J ) \leftarrow (\theta _{1}^{\ast }, U_{1:T}^{\ast},J^{\ast })$ with probability
\begin{align}
\alpha _{1} & \left( U_{1:T}, J, \theta _{1};U_{1:T}^{\ast },J^{\ast},\theta _{1}^{\ast }|\theta _{2}\right)  =  1\wedge \nonumber \\
& \frac{\tilde{\pi}^{N}\left( U_{1:T}^{\ast } , \theta _{1}^{\ast}|\theta _{2}\right) }{\tilde{\pi}^{N}\left( U_{1:T}, \theta _{1}|\theta_{2}\right) }\,
\frac{q_{1}(U_{1:T}, \theta _{1}|U_{1:T}^{\ast }, J^{\ast}, \theta _{2},\theta _{1}^{\ast })}{q_{1}(U_{1:T}^{\ast },  \theta _{1}^{\ast}|U_{1:T}, J,\theta _{2},\theta _{1})} ,  \label{eq:PMwGaccprob}
\end{align}%
where%
\begin{eqnarray*}
q_{1}( U_{1:T}^{\ast },\theta _{1}^{\ast } | U_{1:T}, J, \theta
_{2},\theta _{1}) & = & q_{1,1}(\theta _{1}^{\ast }|U_{1:T}, J, \theta _{2},\theta
_{1})
 \psi(U_{1:T}^\ast|\theta _{2},\theta _{1}^{\ast}).
\end{eqnarray*}

\end{enumerate}

\item (PG sampling) 

\begin{enumerate}
\item Sample $\theta _{2}^{\ast }\sim q_{2}(\cdot
|X_{1:T}^{J},B_{1:T-1}^{J},J,\theta _{2},\theta _{1}).$

\item Set $\theta_2 \leftarrow \theta_{2}^{\ast }$ with probability
\begin{eqnarray}
\lefteqn{\alpha_{2}\left(\theta_{2};\theta_{2}^{\ast}|X_{1:T}^{J},B_{1:T-1}^{J},J,\theta _{1}\right) =}  \notag \\
&&1\wedge \frac{\tilde{\pi}^{N}\left( \theta_{2}^{\ast}|X_{1:T}^{J},B_{1:T-1}^{J},J,\theta _{1}\right) }{\tilde{\pi}^{N}\left(\theta_{2}|X_{1:T}^{J},B_{1:T-1}^{J},J,\theta_{1}\right) }
\times \frac{q_{2}(\theta _{2}|X_{1:T}^{J},B_{1:T-1}^{J},J,\theta_{1},\theta_{2}^{\ast })}{q_{2}(\theta_{2}^{\ast }|X_{1:T}^{J},B_{1:T-1}^{J},J,\theta_{1},\theta_{2}) } .
\label{eq:PMwGaccproba}
\end{eqnarray}
\end{enumerate}

\item Sample $U_{1:T}^{(-J)}\sim \tilde{\pi}^{N}(\cdot|X_{1:T}^{J},B_{1:T-1}^{J},J,\theta )$ using the conditional sequential Monte Carlo algorithm (CSMC) discussed in Section ~\ref{s:CSMC}.

\item Sample $J\sim \tilde{\pi}^{N}\left( \cdot |U_{1:T},\theta \right)$.
\end{enumerate}
\end{sscheme}

The generalization of the sampling scheme to the case where the components in the parameter vector $\theta _{1}$ are sampled separately in multiple PMMH steps and the components in the parameter vector $\theta _{2}$
are sampled separately in multiple Gibbs steps is straighforward and involves repeated steps of the same form as given in Part 1 and Part 2 respectively.

Note that Parts 2 to 4 are the same as the particle Gibbs sampler described in \cite{andrieuetal2010} or the particle Metropolis within Gibbs sampler described in \cite{lindstenschon2012}.
Part 1 differs from the particle Marginal Metropolis-Hastings approach discussed in \cite{andrieuetal2010} by generating the variable $J$ which selects the trajectory.
This is necessary since $J$ is used in Part 2.

A major computational cost of the algorithm is generating the particles $p^{\ast}$ times in Part 1, where $p^{\ast}$ is the number of PMMH steps, as well as running the CSMC algorithm in Part 3.
Hence there is a computational cost in using the PMMH+PG sampler compared to a particle Gibbs sampler. Similar comments apply to a blocked PMMH sampler.

Section~\ref{s:theory} of the supplementary material discusses the convergence of Sampling Scheme~\ref{ssch:pmwg}
to its target distribution.

\begin{remark}
 \cite{andrieuetal2010} show that%
\begin{equation}
\frac{\tilde{\pi}^{N}\left( U_{1:T},\theta _{1}|\theta _{2}\right) }{\psi\left( U_{1:T}|\theta _{2},\theta _{1}\right) }=\frac{Z(U_{1:T},\theta)p(\theta _{1}|\theta _{2})}{p\left( y_{1:T}|\theta _{2}\right) },
\label{eq:accprobsimplify1}
\end{equation}%
and hence the Metropolis-Hastings acceptance probability in Eq. (\ref{eq:PMwGaccprob}) simplifies to%
\begin{equation}
1\wedge \frac{Z(\theta _{1}^{\ast },\theta _{2},U_{1:T}^{\ast })}{Z(\theta_{1},\theta _{2},U_{1:T})}\,\frac{q_{1,1}(\theta _{1}|U_{1:T}^{\ast},J^{\ast},\theta _{2},\theta _{1}^{\ast })p(\theta _{1}^{\ast }|\theta
_{2})}{q_{1,1}(\theta _{1}^{\ast }|U_{1:T},J,\theta _{2},\theta_{1})p(\theta _{1}|\theta _{2})}.
\label{eq:accprobsimplify1a}
\end{equation}
Equation~\eqref{eq:accprobsimplify1a} shows the PMMH steps can be viewed as involving a particle approximation to an \textit{ideal} sampler which we  use to estimate the likelihood of the model.
This version of the PMMH algorithm can also be viewed as a Metropolis-Hastings algorithm using an unbiased estimate of the likelihood.
\end{remark}

\begin{remark}
Part 1 of the sampling scheme is a good choice for parameter vector $\theta _{1}$
which is highly correlated with the state vector $X_{1:T}$.
Part 2 of the sampling scheme is a good choice if the parameter vector $\theta_{2}$ is not highly correlated with the states
and it is possible to sample exactly from the distribution $\tilde{\pi}^{N}\left( \theta _{2}|X_{1:T}^{J},B_{1:T-1}^{J},J,\theta _{1}\right)$%
or a good approximation is available as a Metropolis-Hastings proposal. Using Eq. \eqref{eq:margdist}, the Metropolis-Hastings acceptance probability in Eq. \eqref{eq:PMwGaccproba} simplifies to
\begin{equation}
\frac{p\left(y_{1:T}|X_{1:T}^{J},\theta_{2}^{*},\theta_{1}\right)p\left(X_{1:T}^{J}|\theta_{2}^{*},\theta_{1}\right)p\left(\theta_{2}^{*}|\theta_{1}\right)}{p\left(y_{1:T}|X_{1:T}^{J}\theta_{2},\theta_{1}\right)p\left(X_{1:T}^{J}|\theta_{2},\theta_{1}\right)p\left(\theta_{2}|\theta_{1}\right)}\times\frac{q_{2}\left(\theta_{2}|X_{1:T}^{J},B_{1:T-1}^{J},J,\theta_{1},\theta_{2}^{*}\right)}{q_{2}\left(\theta_{2}^{*}|X_{1:T}^{J},B_{1:T-1}^{J},J,\theta_{1},\theta_{2}\right)}.
\end{equation} 
See \cite{lindstenschon2012} for more discussion about the particle
Metropolis-Hastings within Gibbs proposals in Part 2. \end{remark}



\section{Univariate Example: The univariate continuous time Ornstein-Uhlenbeck process\label{SSS: cts time OU process}}
This section applies the PMMH + PG sampler defined in Section \ref{s:pmwg} to the univariate continuous time Ornstein-Uhlenbeck SV model with covariates in the mean.

\subsection{Definition of inefficiency} \label{SS: preliminaries of examples}
To define our measure of the inefficiency of a sampler that takes computing time into account, we first define the integrated autocorrelation time (IACT) for a univariate parameter $\theta$,
\begin{equation}
\textrm{IACT}_\theta:= 1+2\sum_{j=1}^{\infty}\rho_{j,\theta}
\end{equation}
where $\rho_{j,\theta}$ is the correlation of the iterates of $\theta$ in the MCMC after the chain has converged.
A large value of IACT for one or more of the parameters indicates that the chain does not mix well.

We estimate $\textrm{IACT}_\theta$ based on $M$ iterates $\theta^{\left[1\right]},...,\theta^{\left[M\right]}$ (after convergence) as
\begin{align*}
{\wh {\rm IACT}}_{\theta,M} &=1+2\sum_{j=1}^{L_{M}}\wh {\rho}_{j,\theta},
\end{align*}
where $\widehat{\rho}_{j,\theta}$ is the estimate of $\rho_{j,\theta}$, $L_{M}=\min(1000,L)$ and $L=\min_{j\leq M}|\widehat{\rho}_{j,\theta}|<2/\sqrt{M}$ because $1/\sqrt M$
is approximately the standard error of the autocorrelation estimates when the series is white noise.
Let $\widehat{\textrm{IACT}}_{\textrm{MAX}}$ and $\widehat{\textrm{IACT}}_{\textrm{MEAN}}$ be the maximum and mean of the estimated IACT values over all the parameters in the model, respectively.
Our measure of the inefficiency of a sampler based on $\wh {\rm IACT}_{\rm MAX}$ is the time normalized variance (TNV),
\begin{equation}
\textrm{TNV}_{\textrm{MAX}}=\widehat{\textrm{IACT}}_{\textrm{MAX}}\times {\rm CT},
\end{equation}
where $\rm CT$ is the computing time in seconds per iteration; we define the inefficiency of a sampler based on $\wh {\rm IACT}_{\rm MEAN}$  similarly.
The relative time normalized variance (RTNV) shows the TNV relative to our method.

\subsection{The univariate continuous time Ornstein-Uhlenbeck process\label{SSS: subsection cts time OU process}}
 We consider the model
\begin{align}\label{eq:obseqnOU}
y_{t}=z_{t}^{'}\beta+\textrm{exp}\left(h_{t}/2\right)\varepsilon_{t},\qquad\textrm{where}\quad\varepsilon_{t}\sim N\left(0,1\right),
\end{align}
with the log-volatility $h_t$ generated by the continuous time Ornstein-Uhlenbeck (OU) process $\{h_{t}\}_{t\ge 1}$,
introduced by \cite{Stein1991}. This process satisfies,
\begin{equation}
    d h_{t}=\alpha\left(\mu-h_{t}\right) d t+\tau d W_{t},\label{eq:transitiondensityuniv}
\end{equation}
where $W_t$ is a Wiener process.
The transition densities for $h_{t}$ have the closed form \citep[][p. 7]{Lunde2015}
\begin{align}
h_{t}|h_{t-1} & \sim N\left(\mu+\exp\left(-\alpha\right)\left(h_{t-1}-\mu\right),\frac{1-\exp\left(-2\alpha\right)}{2\alpha}\tau^{2}\right),
\label{eq:exact transition}
\end{align}
with $h_{1} \sim N\( \mu , \frac{\tau^2}{2 \alpha}\)$.
This is a state space model of the form given in Section~\ref{s:SSM}  with $x_{1:T}=h_{1:T}$ and whose parameters are $\alpha>0$, $\mu$, $\tau^{2}>0$, and $\left(m_{\beta}\times1\right)$
vector $\beta$. This is a general time series model that allows for
a scalar dependent variable $y_{t}$ with possible dependence on covariates
in the mean as well as stochastic variance terms. Thus, $E\left(y_{t}|z_{t},h_{t},\theta\right)=z_{t}^{'}\beta$,
where $z_{t}$ can consist of lags of $y_{t}$; $\textrm{Var}\left(y_{t}|z_{t},h_{t},\theta\right)=\exp\left(h_{t}\right)$.
The model can be applied to many time series and has been extensively
used in the financial econometrics literature. It is straightforward
to generalise this model in a number of ways: for example, by allowing
for covariates in the conditional variance and including conditional
variance term in the mean. See \citet[pp. 216-221]{durbinkoopman2012},
who discuss the basic stochastic volatility model and some extensions.

Many stochastic volatility diffusion models do not have a closed form transition density,
e.g., the continuous time  GARCH diffusion process \cite{Chib2004,  Kleppe2010} discussed in Section  \ref{S: factor SV model explanation},  and it is then
 necessary to estimate such state space models using  an approximation such as the Euler discretization.
It is therefore  informative to study the relative performance of the PG+PMMH sampler for the OU process using  both the closed form transition equation in
Eq.~\eqref{eq:exact transition} as well as the OU with the Euler approximation in Eq.~\eqref{eq: OU with euler univ}, to see the relative loss due to the approximation.

The Euler scheme approximates the evolution of the log-volatilities ${h}_{t}$ in equation \eqref{eq:transitiondensityuniv} by placing $M-1$ evenly spaced points between times $t$ and $t+1$.
We denote the intermediate volatility components by $h_{t,1},...,h_{t,M-1}$, and it is convenient to set $h_{t,0}=h_{t}$ and $h_{t,M}=h_{t+1}$.
The equation for the Euler evolution, starting at $h_{t,0}$ is (see, for example, \cite{Stramer2011}, pg. 234)

\begin{align}\label{eq: OU with euler univ}
h_{t,j}|h_{t,j-1}\sim N\( h_{t,j-1}+\alpha\left(\mu-h_{t,j-1}\right)\delta, \tau^2\delta \),
\end{align}
for $j=1,...,M$, where $\delta = 1/M$.



\subsection{Empirical results\label{SS: Empirical results for OU}}
We use the following notation to describe the algorithm used in this
example. The basic samplers, as used in Sampling Scheme 1, are $\textrm{PMMH}\left(\cdotp\right)$
and $\textrm{PG}\left(\cdotp\right)$. These samplers can be used
alone or in combination. For example, $\textrm{PMMH}\left(\theta\right)$
means using a PMMH step to sample the parameter vector $\theta$;
$\textrm{PMMH}\left(\theta_{1}\right)+\textrm{PG}\left(\theta_{2}\right)$
means sampling $\theta_{1}$ in the PMMH step and $\theta_{2}$ in
the PG step; and $\textrm{PG}\left(\theta\right)$ means sampling
$\theta$ using the PG sampler. Our general procedure to determine an efficient
sampling scheme is to first run a PG algorithm to identify
which parameters have large IACT, or, in some cases, require a large amount of computational time to generate in the PG step. We then generate these parameters
in the PMMH step.

\subsubsection*{Univariate OU model with exact transition density and no covariate}
In this section, we consider the univariate OU model with exact transition density and
no covariate $\left(m_{\beta}=0\right)$. We compare the performance
of the following samplers: (I) $\textrm{PMMH}\left(\alpha,\tau^{2}\right)+\textrm{PG}\left(\mu\right)$,
(II) the particle Gibbs with ancestral tracing approach of \citet{andrieuetal2010}
$\left(\textrm{PGAT}\left(\mu,\tau^{2},\alpha\right)\right)$, (III)
the particle Gibbs with backward simulation approach of \citet{Lindsten2013} $\left(\textrm{PGBS}\left(\mu,\tau^{2},\alpha\right)\right)$,
(IV) PMMH with an adaptive random walk as the proposal density for
the parameters $\left(\textrm{PMMH-RW}\left(\mu,\tau^{2},\alpha\right)\right)$,
(V) PMMH with the Metropolis adjusted Langevin algorithm (MALA) of
\citet{Nemeth2016} for the proposal for the parameters $\left(\textrm{PMMH-MALA}\left(\mu,\tau^{2},\alpha\right)\right)$,
(VI) the correlated PMMH approach of \citet{Deligiannidis2018} with an adaptive
random walk as the proposal density for the parameters $\left(\textrm{Corr. PMMH-RW}\left(\mu,\tau^{2},\alpha\right)\right)$,
(VII) the correlated PMMH approach of \citet{Deligiannidis2018} with the Metropolis
adjusted Langevin algorithm of \citet{Nemeth2016} as the proposal
for the parameters $\left(\textrm{Corr. PMMH-MALA}\left(\mu,\tau^{2},\alpha\right)\right)$, and
(VIII) the particle Gibbs with data augmentation approach of \citet{Fearnhead2016}
$\left(\textrm{PGDA}\left(\mu,\tau^{2},\alpha\right)\right)$. The
score vector required for the MALA algorithm is estimated efficiently
using methods described in \citet{Nemeth2016a}.
The tuning parameters of the PGDA sampler
are set optimally according to the approach described in \citet{Fearnhead2016}.
The correlated
PMMH proposed by \citet{Deligiannidis2018} correlates the random
vectors $\boldsymbol{u}$ and $\boldsymbol{u}^{'}$ used to construct
the estimators of the likelihood at the current and proposed values
of the parameters ($\theta$ and $\theta^{'}$ respectively).
This is done to reduce the variance
of the difference between $\log\left(Z_{1:T}\left(\theta^{'},\boldsymbol{u}^{'}\right)\right)-\log\left(Z_{1:T}\left(\theta,\boldsymbol{u}\right)\right)$
which appears in the PMMH acceptance ratio.
The correlated PMMH significantly
reduces the number of particles required by the standard pseudo marginal
method proposed by \citet{andrieuetal2010}. We use $N=500$ particles
for the PMMH+PG, PGAT, PGBS, PMMH and PGDA samplers, and $N=50$ for
the correlated PMMH sampler. In this example, we
use the bootstrap particle filter to sample the particles for all samplers and the adaptive
random walk in \citet{robertsrosenthal2009} for the PMMH step in the PMMH+PG sampler as the proposal density
for the parameters. The particle filter and the parameter samplers
are implemented in Matlab.

We apply the methods to a sample of daily US steel industry stock
returns data obtained from the Kenneth French website\footnote{http://mba.tuck.dartmouth.edu/pages/faculty/ken.french/datalibrary.html},
using a sample from January 3rd, 2001 to the 24th of December, 2003,
a total of 1,000 observations. The priors for the OU parameters are
$\alpha\sim IG\left(\frac{v_{0}}{2},\frac{s_{0}}{2}\right)$, $\tau^{2}\sim IG\left(\frac{v_{0}}{2},\frac{s_{0}}{2}\right)$,
where $v_{0}=10$ and $s_{0}=1$, $p\left(\mu\right)\propto1$, and
$p\left(\beta\right)\propto1$. These prior densities cover most possible
values in practice. We ran all the sampling schemes for 11,000 iterations
and discarded the initial 1,000 iterations as warmup for all the methods.

Table \ref{tab:Univariate-OU-model with no covariates} shows the
IACT, TNV, and RTNV values for the parameters in the univariate OU
model with an exact transition density and no covariate estimated using the 8 different samplers described above.
The table shows the following points.
(1) Both the PGAT and PGBS samplers have large IACT values for both parameters
$\alpha$ and $\tau^{2}$,
and we show that putting those two parameters
in the PMMH step improves the mixing significantly.
We show later in this section and in Section \ref{s:simulations} that it is also beneficial to use a PMMH step for at least the $\alpha$
and $\tau^{2}$ parameters for the stochastic volatility diffusion
models that use an approximation such as the Euler discretization.
(2) In terms of $\textrm{TNV}_{\textrm{MEAN}}$, the PMMH+PG sampler
is 3.18, 3.12, 1.08, and 1.51 times better than the PGAT, PGBS, $\textrm{Corr. PMMH-MALA}$,
and PGDA samplers respectively, and the PMMH-RW, PMMH-MALA, and correlated
PMMH-RW methods are 1.33, 2.56, and 1.88 times better than the PMMH+PG sampler, respectively.
Similar conclusions can be made based on $\textrm{TNV}_{\textrm{MAX}}$.
(3) The best sampler for this example is the correlated PMMH-RW. (4)
The PMMH-MALA sampler has lower IACT values for all the parameters
compared to the PMMH-RW sampler, but the correlated PMMH-RW sampler
is better than the correlated PMMH-MALA sampler. This shows that there is no advantage of using particle
MALA over the random walk proposal.
It is therefore important to note that although the correlated PMMH can significantly reduce the number of particles required compared to standard PMMH, the variance of the estimate
of the gradient of the log-posterior is not sufficiently small with the choice of $N=50$ particles used by the correlated PMMH sampler.
This confirms the observation made by \cite{Nemeth2016} who write ``Our results show that the behaviour of particle MALA depends on how accurately we can estimate the gradient of the log-posterior.
If the error in the estimate of the gradient is not controlled sufficiently well as we increase dimension,
then asymptotically there will be no advantage in using particle MALA over a particle MCMC algorithm using a random-walk proposal''.
(5) The PGDA sampler has lower IACT values for both $\alpha$ and $\tau^{2}$ parameters compared to the PGBS and PGAT samplers, but it has higher IACT value for $\mu$.
This shows that the PGDA sampler is useful to improve the mixing of the parameters that are highly correlated with the states.
 

\begin{table}[H]
\caption{Inefficiency factors of $\alpha$, $\tau^{2}$, and $\mu$ for the Univariate OU model with an exact transition density and without covariates for the US steel industry stock
returns data with $T=1000$. Sampler I: $\textrm{PMMH}\left(\alpha,\tau^{2}\right)+\textrm{PG}\left(\mu\right)$,
Sampler II: $\textrm{PGAT}\left(\mu,\tau^{2},\alpha\right)$, Sampler
III: $\textrm{PGBS}\left(\mu,\tau^{2},\alpha\right)$, Sampler IV:
$\textrm{PMMH-RW}\left(\mu,\tau^{2},\alpha\right)$, Sampler V: $\textrm{PMMH-MALA}\left(\mu,\tau^{2},\alpha\right)$,
Sampler VI: $\textrm{Correlated PMMH-RW}\left(\mu,\tau^{2},\alpha\right)$,
Sampler VII: $\textrm{Correlated PMMH-MALA}\left(\mu,\tau^{2},\alpha\right)$,
and Sampler VIII: $\textrm{PGDA}\left(\mu,\tau^{2},\alpha\right)$.
\label{tab:Univariate-OU-model with no covariates}}

\centering{}%
\begin{tabular}{ccccccccc}
\hline 
Param & I & II & III & IV & V & VI & VII & VIII\tabularnewline
\hline 
$\alpha$ & 12.01 & 50.21 & 40.12 & 15.02 & 4.62 & 13.00 & 12.38 & 18.06\tabularnewline
$\mu$ & 1.56 & 1.65 & 1.48 & 12.81 & 4.59 & 14.17 & 28.77 & 9.16\tabularnewline
$\tau^{2}$ & 13.49 & 85.46 & 70.98 & 12.64 & 4.74 & 11.18 & 17.20 & 19.42\tabularnewline
\hline 
$\widehat{\textrm{IACT}}_{\textrm{MAX}}$ & 13.49 & 85.46 & 70.98 & 15.02 & 4.74 & 14.17 & 28.77 & 19.42\tabularnewline
$\widehat{\textrm{TNV}}_{\textrm{MAX}}$ & 2.16 & 8.55 & 8.52 & 1.20 & 0.57 & 0.85 & 2.30 & 2.72\tabularnewline
$\widehat{\textrm{RTNV}}_{\textrm{MAX}}$ & 1 & 3.95 & 3.94 & 0.56 & 0.26 & 0.39 & 1.06 & 1.25\tabularnewline
\hline 
$\widehat{\textrm{IACT}}_{\textrm{MEAN}}$ & 9.02 & 45.77 & 37.53 & 13.49 & 4.65 & 12.78 & 19.45 & 15.55\tabularnewline
$\widehat{\textrm{TNV}}_{\textrm{MEAN}}$ & 1.44 & 4.58 & 4.50 & 1.08 & 0.56 & 0.77 & 1.56 & 2.17\tabularnewline
$\widehat{\textrm{RTNV}}_{\textrm{MEAN}}$ & 1 & 3.18 & 3.12 & 0.75 & 0.39 & 0.53 & 1.08 & 1.51\tabularnewline
\hline 
Time & 0.16 & 0.10 & 0.12 & 0.08 & 0.12 & 0.05 & 0.08 & 0.14\tabularnewline
\hline 
\end{tabular}
\end{table}

\subsubsection*{Univariate OU model with exact transition density and 50 covariates}
We now consider the univariate OU model with an exact transition density and $m_{\beta}=50$
covariates. We compare the performance of the following samplers:
(1) $\textrm{PMMH}\left(\alpha,\tau^{2}\right)+\textrm{PG}\left(\mu,\beta\right)$,
(2) $\textrm{PGAT}\left(\mu,\tau^{2},\alpha,\beta\right)$, (3) $\textrm{PGBS}\left(\mu,\tau^{2},\alpha,\beta\right)$,
(4) $\textrm{PMMH-RW}\left(\mu,\tau^{2},\alpha,\beta\right)$, (5)
$\textrm{PMMH-MALA}\left(\mu,\tau^{2},\alpha,\beta\right)$, (6) $\textrm{Corr. PMMH-RW}\left(\mu,\tau^{2},\alpha,\beta\right)$,
(7) $\textrm{Corr. PMMH-MALA}\left(\mu,\tau^{2},\alpha,\beta\right)$,
and (8) $\textrm{PGDA}\left(\mu,\tau^{2},\alpha,\beta\right)$. We
use $N=500$ particles for the PMMH+PG, PGAT, PGBS, PMMH,
and PGDA samplers, and $N=50$ for the correlated PMMH sampler.
We simulated data with
$T=1000$ and set $\alpha=0.09$, $\mu=0.38$, $\tau^{2}=0.08$, and
$\beta_{i}=0.1$ for $i=1,...,m_{\beta}$. The covariates are $z_{t}\sim N\left(0,I_{50}\right)$. 

Table \ref{tab:Univariate-OU-model with covariate} shows the IACT,
TNV, and RTNV values for the parameters in the univariate OU model
with an exact transition density and 50 covariates estimated using the 8 different samplers listed above.
The table shows the following points.
(1) The best sampler for this example is the PMMH+PG sampler. This example shows how the PMMH and PG samplers can be combined in a flexible manner
to obtain good results. In this example, the vector of parameters
$\beta$ are high dimensional and not highly correlated with the states,
so it is important to generate them in a PG step.
Both $\alpha$ and $\tau^{2}$ are generated in a
PMMH step because they are highly correlated with the states. (2)
The standard and correlated PMMH with adaptive random walks are much
worse than the PMMH+PG sampler because the adaptive random walk proposal is inefficient in high
dimensions.
(3) The correlated PMMH with the MALA proposal is worse than
the correlated PMMH with an adaptive random walk proposal and is the
worst sampler in this example because the variance of the gradient of log-posterior is not sufficiently small with the
number of particles set to $N=50$.
(4) The PGDA sampler has very large IACT
values for all parameters indicating that the PGDA sampler does not
perform well for models with a large number of parameters.

Figure \ref{fig:The-Inefficiency-Factors log-volatilities} shows
the RTNV of the PMMH+PG sampler over other samplers for the log-volatilities
$h_{1:T}$ for all $t$. The figure shows that the PMMH+PG sampler
is much more efficient than the standard and correlated PMMH samplers and the PGDA sampler.
It is only slightly worse than the PGAT and PGBS samplers.

\begin{table}[H]
\caption{Inefficiency factors of $\alpha$, $\tau^{2}$, and $\mu$ for the Univariate OU model with an exact transition density and $m_{\beta}=50$ covariates for the simulated data with $T=1000$.
Sampler I: $\textrm{PMMH}\left(\alpha,\tau^{2}\right)+\textrm{PG}\left(\beta,\mu\right)$,
Sampler II: $\textrm{PGAT}\left(\beta,\mu,\tau^{2},\alpha\right)$, Sampler
III: $\textrm{PGBS}\left(\beta,\mu,\tau^{2},\alpha\right)$, Sampler IV:
$\textrm{PMMH-RW}\left(\beta,\mu,\tau^{2},\alpha\right)$, Sampler V: $\textrm{PMMH-MALA}\left(\beta,\mu,\tau^{2},\alpha\right)$,
Sampler VI: $\textrm{Correlated PMMH-RW}\left(\beta,\mu,\tau^{2},\alpha\right)$,
Sampler VII: $\textrm{Correlated PMMH-MALA}\left(\beta,\mu,\tau^{2},\alpha\right)$,
and Sampler VIII: $\textrm{PGDA}\left(\beta,\mu,\tau^{2},\alpha\right)$.
\label{tab:Univariate-OU-model with covariate}}

\centering{}%
\begin{tabular}{ccccccccc}
\hline 
Param & I & II & III & IV & V & VI & VII & VIII\tabularnewline
\hline 
$\alpha$ & 11.15 & 47.14 & 40.94 & 281.68 & 33.15 & 135.17 & 561.44 & 356.24\tabularnewline
$\mu$ & 1.58 & 1.73 & 1.81 & 377.59 & 17.79 & 84.31 & 931.89 & 211.48\tabularnewline
$\tau^{2}$ & 14.50 & 95.55 & 71.83 & 341.19 & 20.43 & 81.17 & 1368.65 & 296.52\tabularnewline
$\textrm{mean}\left(\beta\right)$ & 1.52 & 1.57 & 1.46 & 165.50 & 14.43 & 131.88 & 958.51 & 276.13\tabularnewline
$\max\left(\beta\right)$ & 1.80 & 1.95 & 1.71 & 545.26 & 21.76 & 434.50 & 1445.25 & 690.57\tabularnewline
\hline 
$\widehat{\textrm{IACT}}_{\textrm{MAX}}$ & 14.50 & 95.55 & 71.83 & 545.26 & 33.15 & 434.50 & 1445.25 & 690.57\tabularnewline
$\widehat{\textrm{TNV}}_{\textrm{MAX}}$ & 2.47 & 9.55 & 9.34 & 43.62 & 7.96 & 26.07 & 130.07 & 227.89\tabularnewline
$\widehat{\textrm{RTNV}}_{\textrm{MAX}}$ & 1 & 3.87 & 3.78 & 17.66 & 3.22 & 10.55 & 52.66 & 92.26\tabularnewline
\hline 
$\widehat{\textrm{IACT}}_{\textrm{MEAN}}$ & 1.95 & 4.21 & 3.53 & 175.00 & 14.96 & 130.10 & 958.26 & 276.81\tabularnewline
$\widehat{\textrm{TNV}}_{\textrm{MEAN}}$ & 0.33 & 0.42 & 0.46 & 14.00 & 3.59 & 7.81 & 86.24 & 91.35\tabularnewline
$\widehat{\textrm{RTNV}}_{\textrm{MEAN}}$ & 1 & 1.27 & 1.39 & 42.42 & 10.88 & 23.67 & 261.33 & 276.82\tabularnewline
\hline 
Time & 0.17 & 0.10 & 0.13 & 0.08 & 0.24 & 0.06 & 0.09 & 0.33\tabularnewline
\hline 
\end{tabular}
\end{table}

\begin{figure}[H]

\caption{The Inefficiency Factors for the log-volatilities $h_{1:T}$ for the
univariate OU model with 50 covariates for simulated data with $T=1000$.
The relative Time Normalised Variance (RTNV) is computed relative
to the PMMH+PG sampler\label{fig:The-Inefficiency-Factors log-volatilities}}

\centering{}\includegraphics[width=15cm,height=10cm]{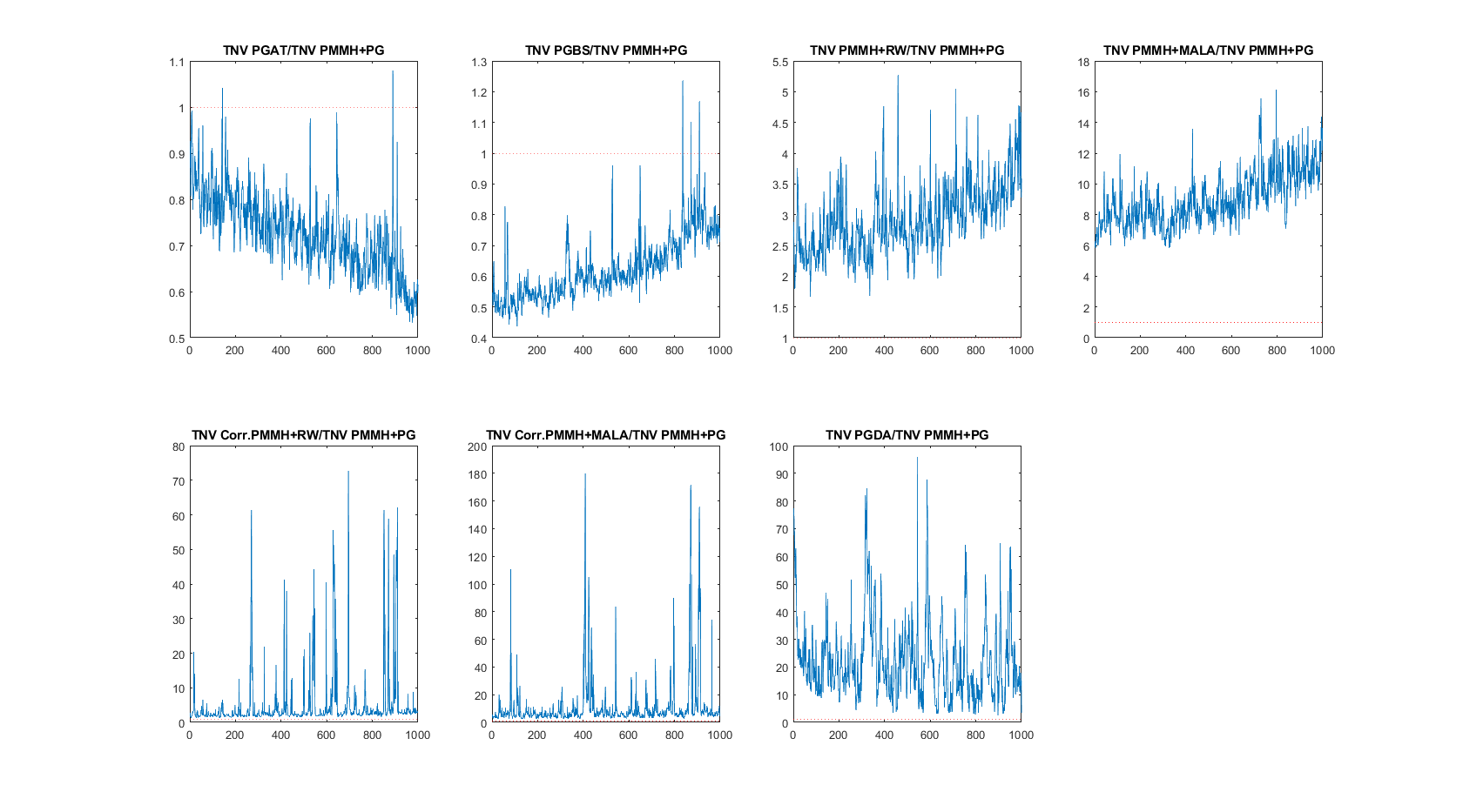}
\end{figure}

\subsubsection*{Univariate OU model with Euler approximation for the state transition density and 50 covariates}
Lastly, we consider the univariate OU model with an Euler approximation for the state transition density and $m_{\beta}=50$
covariates. We compare the performance of the following samplers:
(1) $\textrm{PMMH}\left(\mu,\alpha,\tau^{2}\right)+\textrm{PG}\left(\beta\right)$,
(2) $\textrm{PGAT}\left(\mu,\tau^{2},\alpha,\beta\right)$, (3) $\textrm{PGBS}\left(\mu,\tau^{2},\alpha,\beta\right)$.
We used $N=500$ particles for all samplers and $M=10$ latent points for the Euler approximation of the state transition density.

Table \ref{tab:Univariate-OU-model with covariate-Euler} shows the IACT, TNV, and RTNV values for the parameters in the univariate OU model with an Euler approximation for the state transition density and 50 covariates.
The table shows the following points.
(1) The PMMH+PG samplers with exact and approximate state transition densities have very similar IACT values suggesting
that the inefficiency of the PMMH+PG sampler does not deteriorate when the Euler approximation is used. However, both the PGAT and PGBS samplers using the Euler approximation are significantly worse than the
PGAT and PGBS samplers with exact transition densities.
(2) The best sampler is the PMMH+PG sampler.
(3)  It is interesting to see that when we use an Euler approximation for the diffusion the  PMMH+PG, PGAT,
and PGBT samplers all take approximately the same computing time. This is because the PGAT and PGBT samplers need to store and trace back all the latent log-volatilities $h_{t}$
and the $M$ latent data points between $t$ and $t+1$ for all  $t=1,...,T$, whereas the PMMH+PG sampler only needs to store and trace back the latent log-volatilities $h_{t}$ for all $t=1,...,T$.
Therefore, the PMMH+PG sampler is also more efficient in terms of memory usage if it is necessary to use an Euler approximation.

In summary, in this univariate example, we show the following points.
(1) The inefficiency of the PMMH+PG sampler does not deteriorate when the Euler approximation is used, whereas both the PGAS and PGAT samplers are significantly worse.
(2) PGDA is useful to improve the mixing of the parameters that are highly correlated with the states, but it does not work for models with many parameters.
(3) The PMMH+PG sampler is much more efficient than the standard and correlated PMMH samplers with adaptive random walk proposals because the random walk proposals are inefficient in high dimensions.
(4) There is no advantage of using particle
MALA over the random walk proposal when the variance of the estimate of the gradient of the log-posterior is not sufficiently small.
(5) It is desirable to generate
parameters that are highly correlated with the states using a PMMH step that
does not condition on the states. Conversely, if there is a subset of
parameters that is not highly correlated with the states, then it is
preferable to generate them using a particle Gibbs step, or a particle
Metropolis within Gibbs step, that conditions on the states, especially when
the subset is large. In general, using PG may be preferred to PMMH whenever possible, because it may be easier to obtain better proposals within a PG framework. 
(6) Our PMMH + PG approach can be further refined by using the data augmented PMMH and PG sampling schemes proposed by \citet{Fearnhead2016}
and the refined proposals for the PMMH sampling scheme by \citet{dahlin2015} and \citet{Nemeth2016}. 


\begin{table}[H]
\caption{Univariate OU model with $m_{\beta}=50$ covariates and Euler approximation
for the state transition density for the simulated data with $T=1000$.
Sampler I: $\textrm{PMMH}\left(\alpha,\tau^{2},\mu\right)+\textrm{PG}\left(\beta\right)$,
Sampler II: $\textrm{PGAT}\left(\beta,\mu,\tau^{2},\alpha\right)$,
Sampler III: $\textrm{PGBS}\left(\beta,\mu,\tau^{2},\alpha\right)$.
\label{tab:Univariate-OU-model with covariate-Euler}}

\centering{}%
\begin{tabular}{cccc}
\hline 
Param & I & II & III\tabularnewline
\hline 
$\alpha$ & 12.23 & 175.33 & 130.71\tabularnewline
$\mu$ & 13.56 & 18.09 & 15.22\tabularnewline
$\tau^{2}$ & 10.99 & 403.72 & 347.64\tabularnewline
$\textrm{mean}\left(\beta\right)$ & 1.52 & 1.55 & 1.46\tabularnewline
$\max\left(\beta\right)$ & 1.72 & 1.87 & 1.72\tabularnewline
\hline 
$\widehat{\textrm{IACT}}_{\textrm{MAX}}$ & 13.56 & 403.72 & 347.64\tabularnewline
$\widehat{\textrm{TNV}}_{\textrm{MAX}}$ & 3.53 & 117.08 & 111.24\tabularnewline
$\widehat{\textrm{RTNV}}_{\textrm{MAX}}$ & 1 & 33.17 & 31.51\tabularnewline
\hline 
$\widehat{\textrm{IACT}}_{\textrm{MEAN}}$ & 2.13 & 12.73 & 10.69\tabularnewline
$\widehat{\textrm{TNV}}_{\textrm{MEAN}}$ & 0.55 & 3.69 & 3.42\tabularnewline
$\widehat{\textrm{RTNV}}_{\textrm{MEAN}}$ & 1 & 6.71 & 6.22\tabularnewline
\hline 
Time & 0.26 & 0.29 & 0.32\tabularnewline
\hline 
\end{tabular}
\end{table}

\section{Multivariate Example\label{PMMH+PG SV factor}}
This section applies the ideas in this paper to the multivariate factor stochastic volatility model, which is a serious complex example.
It also shows how a complex particle MCMC scheme can be built from the basic PMMH + PG sampler in Section~\ref{s:pmwg}.
Section \ref{S: factor SV model explanation} discusses the multivariate factor stochastic volatility model.
Section \ref{s:simulations} compares the performance of the PMMH+PG sampler to other competing PMCMC methods to estimate multivariate factor SV models using both simulated and real datasets.

\subsection{The factor stochastic volatility model \label{S: factor SV model explanation}}
Factor stochastic volatility (SV) models are a popular approach to jointly model many co-varying financial time series,
as they are able to capture their common features using only a small number of latent factors (see, e.g., \cite{Chib2006} and \cite{Kastner:2017}).
However, estimating time-varying multivariate factor SV models can be very challenging because the likelihood involves calculating an integral over a very high-dimensional latent state space,
and the number of parameters in the model can be large. 

We consider a factor SV model with the volatilities of the factors following a traditional SV model \citep{Chib2006, Kastner:2017},
while the log volatilities of the idiosyncratic errors follow continuous time Ornstein-Uhlenbeck (OU) processes \citep{Stein1991} or GARCH diffusion processes \citep{Chib2004,  Kleppe2010}.
The log volatility of an OU process admits a closed form state transition density,
see Section~\ref{SSS: subsection cts time OU process}, whereas the GARCH diffusion process does not.
Our estimation methods are applied to Euler approximations of the diffusion process driving the log volatilities, and
hence can handle diffusions that do not admit closed form transition densities; see \cite{Ignatieva2015} for other diffusions whose transition equations need an Euler approximation
because they cannot be expressed in closed form.
It is informative to study the closed form and Euler approximation for the state transition density for the OU process in the multivariate case to see the relative loss due to the approximation.


%


Suppose that $\boldsymbol{P}_{t}$ is a $S\times1$ vector of daily stock prices and define $\boldsymbol{y}_{t}:=\log\boldsymbol{P}_{t}-\log\boldsymbol{P}_{t-1}$ as the log-return of the stocks.
We model ${\bs y_t}$ as the factor SV  model
\begin{equation}
\boldsymbol{y}_{t}=\boldsymbol{\beta}\boldsymbol{f}_{t}+\boldsymbol{V}_{t}^{\frac{1}{2}}\boldsymbol{\epsilon}_{t} \quad (t = 1, \ldots, T),\label{eq:discretisationfactor}
\end{equation}
where $\boldsymbol{f}_{t}$ is a $K\times1$ vector of latent factors (with $K\ll S$), $\boldsymbol{\beta}$ is a $S\times K$ factor loading matrix of unknown parameters.
Appendix \ref{SS: sampling loading matrix} gives further details on the restrictions on $\boldsymbol{\beta}$.
We model the latent factors as $\boldsymbol{f}_{t}\sim N\left(0,\boldsymbol{D}_{t}\right)$ and $\boldsymbol{\epsilon}_{t}\sim N\left(0,I\right)$,
so that $\boldsymbol{y}_{t}|(\boldsymbol{f}_{t},\boldsymbol{h}_{t})\sim N\left(\boldsymbol{\beta}\boldsymbol{f}_{t},\boldsymbol{V}_{t}\right)$.
The time-varying variance matrices $\boldsymbol{D}_{t}$ and $\boldsymbol{V}_{t}$ depend on unobserved random variables
$\boldsymbol{\lambda}_{t}=\left(\lambda_{1,t},...,\lambda_{K,t}\right)$ and $\boldsymbol{h}_{t}=\left(h_{1,t},...,h_{S,t}\right)$ such that
\[
\boldsymbol{D}_{t}  := \textrm{diag}\left(\exp\left(\lambda_{1,t}\right),...,\exp\left(\lambda_{K,t}\right)\right), \quad
\boldsymbol{V}_{t}  :=\textrm{diag}\left(\exp\left(h_{1,t}\right),...,\exp\left(h_{S,t}\right)\right).
\]
Each $\lambda_{k,t}$ is assumed to follow an independent autoregressive process
\begin{equation}
\lambda_{k,t}=\phi_{k}\lambda_{k,t-1}+\tau_{f,k}\eta_{k,t}, \quad k=1,...,K,\label{eq:SVtransition}
\end{equation}
with $\eta_{k,t}\sim N\left(0,1\right)$.
The log volatilities $h_{s,t}$ follow a either a Gaussian OU continuous time volatility process or a GARCH diffusion continuous time volatility process.

The continuous time Ornstein-Uhlenbeck (OU) process $\{h_{s,t}\}_{t\ge 1}$ discussed in Section~\ref{SSS: subsection cts time OU process} satisfies
\begin{equation}
    d h_{s,t}=\alpha_{s}\left(\mu_{s}-h_{s,t}\right) d t+\tau_{\epsilon,s} d W_{s,t},\quad \textrm{for}\quad \;s=1,...,S,\label{eq:transitiondensity}
\end{equation}
where $W_{s,t}$ is a Wiener process. The transition distribution for each $h_{s,t}$ is \citep[][p. 7]{Lunde2015}
\begin{align}
h_{s,t}|h_{s,t-1} & \sim N\left(\mu_{s}+\exp\left(-\alpha_{s}\right)\left(h_{s,t-1}-\mu_{s}\right),\frac{1-\exp\left(-2\alpha_{s}\right)}{2\alpha_{s}}\tau_{\epsilon,s}^{2}\right), \quad s = 1, \dots, S.\label{eq:exact transition}
\end{align}
with $h_{s,1} \sim N\( \mu_s , \frac{\tau_{\epsilon,s}^2}{2 \alpha_s}\)$. The parameters are $\alpha_s>0$, $\mu_s$ and $\tau_{\epsilon,s}^2>0$.

The Euler scheme approximates the evolution of the log-volatilities ${h}_{s,t}$ in equation \eqref{eq:transitiondensity}.
We use the approach in Section~\ref{SSS: subsection cts time OU process} by placing $M-1$ evenly spaced points between times $t$ and $t+1$.
The intermediate volatility components are denoted by $h_{s,t,1},...,h_{s,t,M-1}$, and it is convenient to set $h_{s,t,0}=h_{s,t}$ and $h_{s,t,M}=h_{s,t+1}$.
The equation for the Euler evolution, starting at $h_{s,t,0}$ is (see, for example, \cite{Stramer2011}, pg. 234)
\begin{align}
h_{s,t,j}|h_{s,t,j-1}\sim N\( h_{s,t,j-1}+\alpha_s\left(\mu_s-h_{s,t,j-1}\right)\delta, \tau_{\epsilon,s}^2\delta \),\label{eq:ou approximatetransition multiple}
\end{align}
for $j=1,...,M$, where $\delta = 1/M$.

The continuous time  GARCH diffusion process $\{h_{s,t}\}_{t\ge 1}$ \citep{Chib2004,  Kleppe2010} satisfies
\begin{equation}
dh_{s,t}=\left\{ \alpha_{s}\left(\mu_{s}-\exp\left(h_{s,t}\right)\right)\exp\left(-h_{s,t}\right)-\frac{\tau_{\epsilon,s}^{2}}{2}\right\} dt+\tau_{\epsilon,s}dW_{s,t},\;\;\textrm{for}\;s=1,...,S\label{eq:GARCHtransitiondensity},
\end{equation}
where the $W_{s,t}$ are independent Wiener processes.
The Euler approximation of the state transition density of equation \eqref{eq:GARCHtransitiondensity} yields the transition density between steps (see for example, \cite{Wuetal2018}, pg. 21)
\begin{align}
h_{s,t,j+1}|h_{s,t,j}\sim N\left(h_{s,t,j}+\left\{ \alpha_{s}\left(\mu_{s}-\exp\left(h_{s,t,j}\right)\right)\exp\left(-h_{s,t,j}\right)-\frac{\tau_{\epsilon,s}^{2}}{2}\right\} \delta,\tau_{\epsilon,s}^{2}\delta\right)
\label{eq:GARCH Euler transitiondensity}
\end{align}
for $j=0,...,M-1$, where $\delta = 1/M$.

We denote the parameter vector for the factor stochastic volatility model given by equations \eqref{eq:discretisationfactor}, \eqref{eq:SVtransition}
and either \eqref{eq:exact transition}, \eqref{eq:ou approximatetransition multiple} or \eqref{eq:GARCH Euler transitiondensity} by
\begin{align*}
\boldsymbol{\omega} = \left( \boldsymbol{\beta}; (\phi_k, \tau_{f,k}), k = 1, \ldots, K; (\alpha_s, \mu_s, \tau_{\epsilon,s}), s = 1, \ldots, S  \right).
\end{align*}


Although the factor SV model can be written in state space form as in Section~\ref{s:SSM},
it is more efficient to take advantage of the extra structure in the model and base the sampling scheme on multiple independent univariate state space models.
The next section outlines the conditional independence structure in the factor SV model.
Sections \ref{sec:Sampling-Schemes-Factor model} and \ref{ss: sampling scheme for the factor SV model} of the supplement give
the more complex target density and sampling schemes required for estimating the posterior distribution of the factor SV model.

\subsection*{Conditional independence in the factor SV model \label{sss: Conditional Independence}}

The key to making the estimation of the factor SV model tractable is that the factor SV model in equation \eqref{eq:discretisationfactor} separates into
independent components consisting of $K$ univariate SV models for the latent
factors and $S$ univariate state space models for the idiosyncratic errors
given the values of $\left(\boldsymbol{y}_{1:T},\boldsymbol{f}_{1:T},\boldsymbol{\omega}\right)$ and the conditional
independence of the innovations of the returns.
The sampling scheme generates the latent factors and factor loading matrix in PG steps and then,
conditioning on the them, estimates a series of univariate state space models.
For $k=1,...,K$, we have that
\begin{align}
f_{k,t}|\lambda_{k,t} \sim N\left(0,\exp\left(\lambda_{k,t}\right)\right), \label{eq:PMMH+PG SV obs density}
\end{align}
with the transition density in equation \eqref{eq:SVtransition}. For
$s=1,...,S$, we have
\begin{align}
y_{s,t}|\boldsymbol{f}_{t},h_{s,t}\sim N\left(\boldsymbol{\beta}_{s}\boldsymbol{f}_{t},\exp\left(h_{s,t}\right)\right),\label{eq:PMMH+PG OU obs density}
\end{align}
with the exact and approximate transition densities given in equations
\eqref{eq:exact transition}, \eqref{eq:ou approximatetransition multiple} or \eqref{eq:GARCH Euler transitiondensity}.

Section \ref{s:simulations} shows on both simulated and real data that the PMMH+PG sampler works well.
We note that our example merely illustrates our methods which can naturally handle multiple factors and most types of log-volatilites for both the factors and idiosyncratic errors.


\subsection{Empirical Studies\label{s:simulations}}

This section presents empirical results  for the factor SV model described
in Section \ref{S: factor SV model explanation} to illustrate the flexibility of the  sampling approach given in our article.
Section \ref{ss:simulationandapplication} presents a simulation study for the factor SV model with the idiosyncratic log-volatilities following Gaussian OU processes with exact and approximate transition densities.
Section \ref{SS: US stock returns} presents empirical results for the factor SV model with the idiosyncratic log-volatilities following
Gaussian OU processes and GARCH diffusion processes using a sample of daily US industry stock returns data.   

We use the same notation as Section \ref{SS: Empirical results for OU} to describe the algorithms in this study.
For example, the basic sampler, as used in Sampling Scheme 1, is
$\textrm{PMMH}\left(\theta_{1}\right)+\textrm{PG}\left(\theta_{2}\right)$
sampling the parameter vector $\theta_{1}$ in the PMMH step and $\theta_{2}$ in
the PG step.
Our general procedure to determine an efficient
sampling scheme is to first run a PG algorithm to identify
which parameters have large IACTs, or, in some cases, require a large amount of computational time to generate in the PG step. We then generate these parameters
in the PMMH step. 

\subsubsection{Simulation Study \label{ss:simulationandapplication}}
We conducted a simulation study for the factor SV model with the idiosyncratic
log-volatilities following Gaussian OU continuous time volatility processes
with exact and approximate transition densities. 

We compare the performance
of the samplers listed below.
Section~\ref{SS: Empirical results for OU} gives the notation for the samplers.
The samplers are: (I) $\textrm{PMMH}\left(\boldsymbol{\alpha},\boldsymbol{\tau}_{\epsilon}^{2},\boldsymbol{\tau}_{f}^{2}\right)+\textrm{PG}\left(\boldsymbol{f}_{1:T},\boldsymbol{\beta},\boldsymbol{\mu},\phi\right)$
for the Gaussian OU model with exact transition densities and $\textrm{PMMH}\left(\boldsymbol{\alpha},\boldsymbol{\tau}_{\epsilon}^{2},\boldsymbol{\tau}_{f}^{2},\boldsymbol{\mu}\right)+\textrm{PG}\left(\boldsymbol{f}_{1:T},\boldsymbol{\beta},\phi\right)$
for the Gaussian OU model with approximate transition densities, (II)
$\textrm{PGAT}\left(\boldsymbol{f}_{1:T},\boldsymbol{\beta},\boldsymbol{\alpha},\boldsymbol{\tau}_{\epsilon}^{2},\boldsymbol{\mu},\phi,\boldsymbol{\tau}_{f}^{2}\right)$,
(III) $\textrm{PGBS}\left(\boldsymbol{f}_{1:T},\boldsymbol{\beta},\boldsymbol{\alpha},\boldsymbol{\tau}_{\epsilon}^{2},\boldsymbol{\mu},\phi,\boldsymbol{\tau}_{f}^{2}\right)$,
(IV) $\textrm{PMMH-RW}\left(\boldsymbol{f}_{1:T},\boldsymbol{\beta},\boldsymbol{\alpha},\boldsymbol{\tau}_{\epsilon}^{2},\boldsymbol{\mu},\phi,\boldsymbol{\tau}_{f}^{2}\right)$,
(V) $\textrm{PMMH-MALA}\left(\boldsymbol{f}_{1:T},\boldsymbol{\beta},\boldsymbol{\alpha},\boldsymbol{\tau}_{\epsilon}^{2},\boldsymbol{\mu},\phi,\boldsymbol{\tau}_{f}^{2}\right)$,
(VI) $\textrm{Corr. PMMH-RW}\left(\boldsymbol{f}_{1:T},\boldsymbol{\beta},\boldsymbol{\alpha},\boldsymbol{\tau}_{\epsilon}^{2},\boldsymbol{\mu},\phi,\boldsymbol{\tau}_{f}^{2}\right)$,
(VII) $\textrm{Corr. PMMH-MALA}\left(\boldsymbol{f}_{1:T},\boldsymbol{\beta},\boldsymbol{\alpha},\boldsymbol{\tau}_{\epsilon}^{2},\boldsymbol{\mu},\phi,\boldsymbol{\tau}_{f}^{2}\right)$,
(VIII) $\textrm{PGDA}\left(\boldsymbol{f}_{1:T},\boldsymbol{\beta},\boldsymbol{\alpha},\boldsymbol{\tau}_{\epsilon}^{2},\boldsymbol{\mu},\phi,\boldsymbol{\tau}_{f}^{2}\right)$.
We first compare the three samplers PMMH+PG, PGAT, and PGBS and then discuss the PMMH and PGDA sampling schemes for the factor SV model.

We simulated data with $T=1,000$ observations, $S=20$ stocks, and $K=1$ factors
from the factor SV model in equation~\eqref{eq:discretisationfactor},
setting $\alpha_{s}=0.06$, and $\tau_{\epsilon,s}^{2}=0.1$ for all $s$, $\phi_{1}=0.98$, $\tau_{f,1}^{2}=0.1$ and $\beta_{s}=0.8$ for all $s$.
We chose independent Gaussian priors for every unrestricted element of the factor loading matrix $\boldsymbol{\beta}$,
i.e. $\beta_{s,k}\sim N\left(0,1\right)$.
The priors for the state transition density parameters are
$\alpha_{s}\sim IG\left(\frac{v_{0}}{2},\frac{s_{0}}{2}\right)$, $\tau_{\epsilon,s}^{2}\sim IG\left(\frac{v_{0}}{2},\frac{s_{0}}{2}\right)$, $\tau_{f,k}^{2}\sim IG\left(\frac{v_{0}}{2},\frac{s_{0}}{2}\right)$,
where $v_{0}=10, s_{0}=1$, and $\phi_{k}\sim U\left(-1,1\right)$. These prior densities cover most possible values in practice.
The initial state of $\lambda_{k,t}$ is assumed normally distributed $N\left(0,\frac{\tau_{f,k}^{2}}{1-\phi_{k}^{2}}\right)$, for $k=1,...,K$.
The initial state of $h_{s,t}$ is also assumed normally distributed $N\left(\mu_{s},\frac{\tau_{\epsilon,s}^{2}}{2\alpha_{s}}\right)$, for $s=1,...,S$.
We ran all the sampling schemes for $11,000$ iterations and discarded the initial $1,000$ iterates as warmup. We used $M=10$ latent points for the Euler approximations to the state transition densities.

\subsection*{Gaussian OU process with exact transition density}
Table \ref{tab:Inefficiency-factor-of simulation} in Section~\ref{S:FSV tables and figures} of the supplement shows the IACT estimates for the parameters in the factor SV model estimated
for three different samplers using the exact transition density, (I)
$\textrm{PMMH}\left(\boldsymbol{\alpha}, \boldsymbol{\tau}_{\epsilon}^{2},\boldsymbol{\tau}_{f}^{2}\right)+\textrm{PG}\left(\boldsymbol{\mu},\boldsymbol{\beta},\boldsymbol{f}_{1:T},
\phi\right)$,  (II) $\textrm{PGAT}\left(\boldsymbol{f}_{1:T},\boldsymbol{\beta},\boldsymbol{\alpha},\boldsymbol{\tau}_{\epsilon}^{2},\boldsymbol{\tau}_{f}^{2},\phi\right)$
and (III) $\textrm{PGBS}\left(\boldsymbol{f}_{1:T},\boldsymbol{\beta},\boldsymbol{\alpha},\boldsymbol{\tau}_{\epsilon}^{2},\boldsymbol{\tau}_{f}^{2}, \phi\right)$.
All three samplers estimate the factor loading matrix $\boldsymbol{\beta}$ and $\boldsymbol{\mu}$ with comparable IACT values.
The PMMH+PG sampler always has lower IACT values than both PG samplers for the parameters $\boldsymbol{\alpha}$, $\boldsymbol{\tau}_{\epsilon}^{2}$, $\boldsymbol{\tau}_{f}^{2}$, and $\phi$.
There are some improvements in terms of IACT obtained by using PGBS compared to PGAT.
Table \ref{tab:Comparison-between-different simulationexact} summarises the estimation results when the exact transition density is used and shows that in terms of $\textrm{TNV}_{\textrm{MAX}}$,
the PMMH+PG sampler is 9.25 and 4.19 times better than PGAT and PGBS, respectively, and in terms of $\textrm{TNV}_{\textrm{MEAN}}$, the PMMH+PG is 2.69 and 2.55 times better than PGAT and PGBS, respectively.

\begin{table}[H]
\caption{Comparing different samplers in terms of Time Normalised
Variance (TNV) with the exact transition density used for the Gaussian OU model: Sampler I: $\textrm{PMMH}\left(\boldsymbol{\alpha},\boldsymbol{\tau}_{\epsilon}^{2},\boldsymbol{\tau}_{f}^{2}\right)+\textrm{PG}\left(\boldsymbol{f}_{1:T},\boldsymbol{\beta},\boldsymbol{\mu},\boldsymbol{\phi}\right)$,
Sampler $\textrm{II}$: $\textrm{PGAT}\left(\boldsymbol{f}_{1:T},\boldsymbol{\beta},\boldsymbol{\alpha},\boldsymbol{\tau}_{\epsilon}^{2},\boldsymbol{\mu},\boldsymbol{\phi},\boldsymbol{\tau}_{f}^{2}\right)$,
sampler III: $\textrm{PGBS}\left(\boldsymbol{f}_{1:T},\boldsymbol{\beta},\boldsymbol{\alpha},\boldsymbol{\tau}_{\epsilon}^{2},\boldsymbol{\mu},\boldsymbol{\phi},\boldsymbol{\tau}_{f}^{2}\right)$.
The data was simulated with $T=1000$, $S=20$, and $K=1$, and number
of particles $N=500$. Time denotes the time taken in seconds for one iteration of the method.
\label{tab:Comparison-between-different simulationexact}}

\centering{}%
\begin{tabular}{cccc}
\hline
 & $I$ & $II$ & $III$\tabularnewline
\hline
{\footnotesize{}$\widehat{\textrm{IACT}}_{\textrm{MAX}}$} & {\footnotesize{}$18.07$} & {\footnotesize{}$283.23$} & {\footnotesize{}$101.64$}\tabularnewline
{\footnotesize{}$\textrm{TNV}_{\max}$} & {\footnotesize{}$33.97$} & {\footnotesize{}$314.39$} & {\footnotesize{}$142.30$}\tabularnewline
{\footnotesize{}$\textrm{RTNV}_{\max}$} & {\footnotesize{}$1$} & {\footnotesize{}$9.25$} & {\footnotesize{}$4.19$}\tabularnewline
\hline
{\footnotesize{}$\widehat{\textrm{IACT}}_{\textrm{MEAN}}$} & {\footnotesize{}$8.54$} & {\footnotesize{}$38.96$} & {\footnotesize{}$29.26$}\tabularnewline
{\footnotesize{}$\textrm{TNV}_{\textrm{MEAN}}$} & {\footnotesize{}$16.06$} & {\footnotesize{}$43.25$} & {\footnotesize{}$40.96$}\tabularnewline
{\footnotesize{}$\textrm{RTNV}_{MEAN}$} & {\footnotesize{}$1$} & {\footnotesize{}$2.69$} & {\footnotesize{}$2.55$}\tabularnewline
\hline
{\footnotesize{}$\textrm{Time}$} & {\footnotesize{}$1.88$} & {\footnotesize{}$1.11$} & {\footnotesize{}$1.40$}\tabularnewline
\hline
\end{tabular}
\end{table}

\subsubsection*{Gaussian OU process with an Euler evolution transition density}
Table \ref{tab:Inefficiency-factor-of simulation-1}  in Section~\ref{S:FSV tables and figures} of the supplement shows the IACT values for all the parameters in
the model for the three samplers, (I) $\textrm{PMMH}\left(\boldsymbol{\mu},\boldsymbol{\alpha}, \boldsymbol{\tau}_{\epsilon}^{2},\boldsymbol{\tau}_{f}^{2}\right)+\textrm{PG}\left(\boldsymbol{\beta},\boldsymbol{f}_{1:T},
\phi\right)$, (II) $\textrm{PGAT}\left(\boldsymbol{f}_{1:T},\boldsymbol{\beta},\boldsymbol{\alpha},\boldsymbol{\tau}_{\epsilon}^{2},\boldsymbol{\tau}_{f}^{2},\phi\right)$
and (III) $\textrm{PGBS}\left(\boldsymbol{f}_{1:T},\boldsymbol{\beta},\boldsymbol{\alpha},\boldsymbol{\tau}_{\epsilon}^{2},\boldsymbol{\tau}_{f}^{2}, \phi\right)$,
using the Euler approximation scheme for the transition density.
The table shows that the PMMH+PG samplers with the exact and approximate state transition densities have very similar IACT values for all the parameters suggesting that the inefficiency of the PMMH+PG sampler
does not deteriorate when the Euler approximation is used.
However, both PG samplers, PGAT and PGBS, using the Euler approximation are significantly worse than the PGAT and PGBS samplers with the exact transition density.
For example, the IACT of $\tau_{4}^{2}$ in PGAT with the exact transition density is 283.23, compared to 977.93 for PGAT with the Euler approximation.

Table \ref{tab:Comparison-between-different simulationapproximate} summarises the estimation results with the Euler approximation of the transition density and shows that in terms of $\textrm{TNV}_{\textrm{MAX}}$,
the PMMH+PG sampler is 60.57 and 50.72 times better than PGAT and PGBS, respectively, and in terms of $\textrm{TNV}_{\textrm{MEAN}}$, the PMMH+PG sampler is 14.67 and 12.95 times better than the PGAT and PGBS samplers,
respectively. Similarly to the univariate case in Section \ref{SS: Empirical results for OU}, we note that if Euler approximations are used for the state transition densities then all three samplers PMMH+PG, PGAT,
and PGBT take approximately the same computing time because the PG samplers need to store and trace back all the latent log-volatilities $h_{s,t}$
and the $M$ latent data points between $t$ and $t+1$ for all $s=1,...,S$ and $t=1,...,T$, whereas the PMMH+PG sampler only needs to store and trace back the latent log-volatilities $h_{s,t}$ for all $s=1,...,S$
and $t=1,...,T$. 

\begin{table}[H]
\caption{Comparing different samplers in terms of Time Normalised
Variance using an Euler approximation for the state transition density for the Gaussian OU model: Sampler
I: $\textrm{PMMH}\left(\boldsymbol{\alpha},\boldsymbol{\tau}_{\epsilon}^{2},\boldsymbol{\mu},\boldsymbol{\tau}_{f}^{2}\right)+\textrm{PG}\left(\boldsymbol{f}_{1:T},\boldsymbol{\beta},\boldsymbol{\phi}\right)$,
Sampler $\textrm{II}$: $\textrm{PGAT}\left(\boldsymbol{f}_{1:T},\boldsymbol{\beta},\boldsymbol{\alpha},\boldsymbol{\tau}_{\epsilon}^{2},\boldsymbol{\mu},\boldsymbol{\phi},\boldsymbol{\tau}_{f}^{2}\right)$,
sampler III: $\textrm{PGBS}\left(\boldsymbol{f}_{1:T},\boldsymbol{\beta},\boldsymbol{\alpha},\boldsymbol{\tau}_{\epsilon}^{2},\boldsymbol{\mu},\boldsymbol{\phi},\boldsymbol{\tau}_{f}^{2}\right)$
for the simulated data with $T=1,000$, $S=20$, and $K=1$, and the number
of particles $N=1,000$. Time denotes the time taken in seconds for  one iteration of the method.\label{tab:Comparison-between-different simulationapproximate}}

\centering{}%
\begin{tabular}{cccc}
\hline
 & $I$ & $II$ & $III$\tabularnewline
\hline
{\footnotesize{}$\widehat{\textrm{IACT}}_{\textrm{MAX}}$} & {\footnotesize{}$17.57$} & {\footnotesize{}$977.93$} & {\footnotesize{}$792.88$}\tabularnewline
{\footnotesize{}$\textrm{TNV}_{\max}$} & {\footnotesize{}$113.50$} & {\footnotesize{}$6874.85$} & {\footnotesize{}$5756.31$}\tabularnewline
{\footnotesize{}$\textrm{RTNV}_{\max}$} & {\footnotesize{}$1$} & {\footnotesize{}$60.57$} & {\footnotesize{}$50.72$}\tabularnewline
\hline
{\footnotesize{}$\widehat{\textrm{IACT}}_{\textrm{MEAN}}$} & {\footnotesize{}$14.17$} & {\footnotesize{}$191.04$} & {\footnotesize{}$163.26$}\tabularnewline
{\footnotesize{}$\textrm{TNV}_{\textrm{MEAN}}$} & {\footnotesize{}$91.54$} & {\footnotesize{}$1343.01$} & {\footnotesize{}$1185.27$}\tabularnewline
{\footnotesize{}$\textrm{RTNV}_{MEAN}$} & {\footnotesize{}$1$} & {\footnotesize{}$14.67$} & {\footnotesize{}$12.95$}\tabularnewline
\hline
{\footnotesize{}$\textrm{Time}$} & {\footnotesize{}$6.46$} & {\footnotesize{}$7.03$} & {\footnotesize{}$7.26$}\tabularnewline
\hline
\end{tabular}
\end{table}

\subsection*{The PMMH and PGDA Sampling Schemes for the Factor SV Model}
This section discusses the PMMH samplers, both the standard and correlated PMMH, and the PGDA sampler of \citet{Fearnhead2016} to estimate the factor SV model which are denoted by sampling schemes IV to VIII.
The PMMH method generates the parameters by integrating out all the latent factors, so that the observation equation is given by
\begin{equation}
\boldsymbol{y}_{t}|\boldsymbol{\lambda}_{t}, \boldsymbol{h}_{t}, \boldsymbol{\omega} \sim N\left(\boldsymbol{0},\boldsymbol{\beta}\boldsymbol{D}_{t}\boldsymbol{\beta}^{'}+\boldsymbol{V}_{t}\right).\label{eq:PMMHmeasurement density}
\end{equation}
The state transition equations are given by equations~\eqref{eq:SVtransition} and either equation \eqref{eq:exact transition} for the closed form case or equation \eqref{eq:ou approximatetransition multiple}
for the Euler scheme for the OU model and equation \eqref{eq:GARCH Euler transitiondensity} for the Euler scheme for the GARCH model.
The PMMH method uses the observation density, which includes all $(K+S)$ dimensional latent log-volatilities simultaneously.
This becomes a high dimensional (21 dimensional) state space model.
The performance of the standard PMMH sampler depends critically on the number of particles $N$ used to estimate the likelihood.
\cite{pittetal2012} suggest selecting the number of particles $N$ such that the variance of the log of the estimated likelihood is around 1 to obtain an optimal tradeoff between computing time and statistical efficiency.
Table \ref{tab:The-Variance-of log-likelihood} gives the variance of the log of the estimated likelihood for different numbers of particles using the bootstrap filter and shows that even with 5,000 particles,
the log of the estimated likelihood still has a large variance and the Markov chain for the standard PMMH approach (sampling schemes IV and V) would get stuck.
We therefore do not report results for the standard PMMH method as it is computationally very expensive and its TNV would be significantly higher than the PG and PMMH+PG methods.

From Section \ref{SS: Empirical results for OU}, we need $\log \left(Z_{1:T}\left(\boldsymbol{\theta}^{'},\boldsymbol{u}^{'}\right)\right)$ and $\log \left(Z_{1:T}\left(\boldsymbol{\theta},\boldsymbol{u}\right)\right)$
to be highly correlated to reduce the variance of the difference between them for the correlated PMMH method.
We now set the correlation between the individual elements of $\bs u$ and $\bs u^{'}$ to $\textrm{corr}\left(u_{i},u_{i}^{'}\right)=0.999999$.
We then obtained $1,000$ independent estimates of $\log \left(Z_{1:T}\left(\boldsymbol{\theta},\boldsymbol{u}^{'}\right)\right)$ and $\log \left(Z_{1:T}\left(\boldsymbol{\theta},\boldsymbol{u}\right)\right)$
at the true value of $\boldsymbol{\theta}$ and computed their sample correlation. The sample correlation was $0.06$, showing that it is difficult to preserve the correlation in such a high dimensional state space model and
that the correlated PMMH Markov chain would still get stuck unless enough particles are used to ensure that the variance of the log of the  estimator of the likelihood is close to 1.

A second problem with the PMMH approach is the large number of parameters to be estimated. Constructing proposals in high dimensions is remarkably difficult,
and often requires estimating gradients and Hessian matrices. On the other hand, simpler approaches such as the adaptive random walk are very inefficient in large dimensions,
 as we showed in Section \ref{SS: Empirical results for OU}.
Hence, it is natural to use a parameter splitting strategy and hybrid samplers.

Finally, we do not report results for the PGDA method applied to the factor stochastic volatility model as it is very clear that its TNV would be significantly higher than the PMMH+PG method.
This sampler updates pseudo observations of the parameters by MCMC and updates the latent states and parameters jointly using a particle filter.
Section \ref{SS: Empirical results for OU} shows that this sampler does not work well when the model has many parameters.
Note that \citet{Fearnhead2016} only apply their method to a simple univariate SV model.
The factor SV model considered in this section is more complex with a large number of parameters and high dimensional latent states.



\begin{table}[H]
\caption{The Variance of the log of the estimated likelihood for the PMMH method with the exact transition
density for different numbers of particles for the simulated dataset
with $T=1,000$, $S=20$, and $K=1$ evaluated at the true values of
the parameters. CPU time to  estimate the
likelihood is in seconds .\label{tab:The-Variance-of log-likelihood}}

\centering{}%
\begin{tabular}{ccc}
\hline
Number of Particles & Variance of log-likelihood & CPU time\tabularnewline
\hline
\hline
250 & 1672.07 & 4.39\tabularnewline
500 & 766.38 & 8.57\tabularnewline
2500 & 331.65 & 45.03\tabularnewline
5000 & 243.82 & 130.53\tabularnewline
\hline
\end{tabular}
\end{table}

\subsubsection{Application to US stock returns\label{SS: US stock returns}}
We now apply our methods to a sample of daily US industry stock returns
data. The data, obtained from the Kenneth French website\footnote{ http://mba.tuck.dartmouth.edu/pages/faculty/ken.french/datalibrary.html}
consists of daily returns for $S=20$ value-weighted industry portfolios,
using a sample from  January 3rd, 2001 to the 24th of December, 2003,
a total of 1,000 observations. 


We  compare the PMMH+PG, PGAT, and PGBS samplers for the factor SV model with the idiosyncratic log-volatilities following Gaussian OU processes with exact and approximate transition densities
and GARCH diffusion processes and show that the performance of the PMMH+PG sampler does not deteriorate for the real
data, whereas both PGAT and PGBS samplers get worse in terms of the IACT values
of the parameters, especially with the Euler approximation.
This section does not compare the PMMH+PG sampler with either of the standard or correlated PMMH samplers or the PGDA sampler because of the problems discussed in Section \ref{ss:simulationandapplication}.

\subsubsection*{Gaussian OU process with exact and Euler evolution transition densities{\label{GaussianOUrealdata}}}
This section compares the following samplers: (I) $\textrm{PMMH}\left(\boldsymbol{\alpha},\boldsymbol{\tau}_{\epsilon}^{2},\boldsymbol{\tau}_{f}^{2}\right)+\textrm{PG}\left(\boldsymbol{f}_{1:T},\boldsymbol{\beta},\boldsymbol{\mu},\phi\right)$
for the Gaussian OU model with exact transition densities and $\textrm{PMMH}\left(\boldsymbol{\alpha},\boldsymbol{\tau}_{\epsilon}^{2},\boldsymbol{\tau}_{f}^{2},\boldsymbol{\mu}\right)+\textrm{PG}\left(\boldsymbol{f}_{1:T},\boldsymbol{\beta},\phi\right)$
for the Gaussian OU model with approximate transition densities, (II)
$\textrm{PGAT}\left(\boldsymbol{f}_{1:T},\boldsymbol{\beta},\boldsymbol{\alpha},\boldsymbol{\tau}_{\epsilon}^{2},\boldsymbol{\mu},\phi,\boldsymbol{\tau}_{f}^{2}\right)$, and
(III) $\textrm{PGBS}\left(\boldsymbol{f}_{1:T},\boldsymbol{\beta},\boldsymbol{\alpha},\boldsymbol{\tau}_{\epsilon}^{2},\boldsymbol{\mu},\phi,\boldsymbol{\tau}_{f}^{2}\right)$
for the factor SV model with the idiosyncratic log-volatilities following Gaussian OU processes with exact and approximate transition densities.
Tables \ref{tab:Inefficiency-factor-of real data} and \ref{tab:Inefficiency-factor-of real data-1} in Section~\ref{S:FSV tables and figures} of the supplement show the IACT
estimates for all the parameters in the factor SV model estimated with exact transition densities for the Gaussian OU model and Euler
approximations for the transition densities for the Gaussian OU processes.
As for the simulated data, all three samplers estimate the factor loading matrix $\boldsymbol{\beta}$
and $\boldsymbol{\mu}$ efficiently and with comparable IACT values.
The performance of the PMMH+PG sampler does not deteriorate for the real
data, whereas both PGAT and PGBS samplers get worse in terms of the IACT values
of the parameters, especially for the Euler approximation model.
Overall, the PMMH+PG samplers always have smaller IACT values than both the PGAT and PGBS
samplers for all the state transition parameters.

Tables \ref{tab:Comparison-between-different real data-exact}
and \ref{tab:Comparison-between-different real dataapproximate} summarise
the estimation results for the Gaussian OU model and show that in terms of $\textrm{TNV}_{\textrm{MAX}}$, the PMMH+PG
sampler is 20.87 and 13.91 times better than the PGAT and PGBS samplers
with the exact transition density, respectively, and the  PMMH+PG sampler
is 53.94 and 58.71 times, respectively, better than the PGAT and PGBS with the Euler approximation.
In terms of $\textrm{TNV}_{\textrm{MEAN}}$, the PMMH+PG
sampler is 5.61 and 4.73 times better than the PGAT and PGBS samplers
with the exact transition density, respectively, and the PMMH+PG sampler
is 22.17 and 22.40 times, respectively, better than the PGAT and PGBS samplers when using the
Euler approximation.

Figures~\ref{fig:The-kernel-density alpha real data} and
\ref{fig:The-kernel-density tau real data}
in Section~\ref{S:FSV tables and figures} of the supplement present
the kernel density estimates of marginal posterior densities of four representative $\alpha$ and $\tau_{\epsilon}^2$ parameters, respectively,
for the US stock returns data. The density estimates are for PMMH+PG using exact and approximate transition densities
and PG with  approximate transition densities using ancestral tracing and  backward simulation for the Gaussian OU model.
The figures show that both PMMH+PG samplers
produce estimates that are close to each other, whereas the PG samplers are much less reliable and suggest that the PG estimators did not converge.
This confirms the usefulness of the PMMH+PG samplers
for this class of model.


\begin{table}[H]
\caption{Comparing different samplers in terms of Time Normalised
Variance with the exact transition density for the Gaussian OU model: Sampler I: $\textrm{PMMH}\left(\boldsymbol{\alpha},\boldsymbol{\tau}_{\epsilon}^{2},\boldsymbol{\tau}_{f}^{2}\right)+\textrm{PG}\left(\boldsymbol{f}_{1:T},\boldsymbol{\beta},\boldsymbol{\mu},\boldsymbol{\phi}\right)$,
Sampler $\textrm{II}$: $\textrm{PGAT}\left(\boldsymbol{f}_{1:T},\boldsymbol{\beta},\boldsymbol{\alpha},\boldsymbol{\tau}_{\epsilon}^{2},\boldsymbol{\mu},\boldsymbol{\phi},\boldsymbol{\tau}_{f}^{2}\right)$,
sampler III: $\textrm{PGBS}\left(\boldsymbol{f}_{1:T},\boldsymbol{\beta},\boldsymbol{\alpha},\boldsymbol{\tau}_{\epsilon}^{2},\boldsymbol{\mu},\boldsymbol{\phi},\boldsymbol{\tau}_{f}^{2}\right)$
for US stock returns data with $T=1,000$, $S=20$, and $K=1$, and
number of particles $N=500$. Time denotes the time taken in seconds for one iteration of the method.\label{tab:Comparison-between-different real data-exact}}

\centering{}%
\begin{tabular}{cccc}
\hline
 & $I$ & $II$ & $III$\tabularnewline
\hline
{\footnotesize{}$\widehat{\textrm{IACT}}_{\textrm{MAX}}$} & {\footnotesize{}$20.57$} & {\footnotesize{}$682.49$} & {\footnotesize{}$382.86$}\tabularnewline
{\footnotesize{}$\textrm{TNV}_{\max}$} & {\footnotesize{}$38.26$} & {\footnotesize{}$798.51$} & {\footnotesize{}$532.18$}\tabularnewline
{\footnotesize{}$\textrm{RTNV}_{\max}$} & {\footnotesize{}$1$} & {\footnotesize{}$20.87$} & {\footnotesize{}$13.91$}\tabularnewline
\hline
{\footnotesize{}$\widehat{\textrm{IACT}}_{\textrm{MEAN}}$} & {\footnotesize{}$8.54$} & {\footnotesize{}$76.19$} & {\footnotesize{}$54.06$}\tabularnewline
{\footnotesize{}$\textrm{TNV}_{\textrm{MEAN}}$} & {\footnotesize{}$15.88$} & {\footnotesize{}$89.14$} & {\footnotesize{}$75.14$}\tabularnewline
{\footnotesize{}$\textrm{RTNV}_{MEAN}$} & {\footnotesize{}$1$} & {\footnotesize{}$5.61$} & {\footnotesize{}$4.73$}\tabularnewline
\hline
{\footnotesize{}$\textrm{Time}$} & {\footnotesize{}$1.86$} & {\footnotesize{}$1.17$} & {\footnotesize{}$1.39$}\tabularnewline
\hline
\end{tabular}
\end{table}

\begin{table}[H]
\caption{Comparing different samplers in terms of Time Normalised
Variance with the Euler approximation for state transition density for the Gaussian OU model: Sampler
I: $\textrm{PMMH}\left(\boldsymbol{\alpha},\boldsymbol{\tau}_{\epsilon}^{2},\boldsymbol{\mu},\boldsymbol{\tau}_{f}^{2}\right)+\textrm{PG}\left(\boldsymbol{f}_{1:T},\boldsymbol{\beta},\boldsymbol{\phi}\right)$,
Sampler $\textrm{II}$: $\textrm{PGAT}\left(\boldsymbol{f}_{1:T},\boldsymbol{\beta},\boldsymbol{\alpha},\boldsymbol{\tau}_{\epsilon}^{2},\boldsymbol{\mu},\boldsymbol{\phi},\boldsymbol{\tau}_{f}^{2}\right)$,
sampler III: $\textrm{PGBS}\left(\boldsymbol{f}_{1:T},\boldsymbol{\beta},\boldsymbol{\alpha},\boldsymbol{\tau}_{\epsilon}^{2},\boldsymbol{\mu},\boldsymbol{\phi},\boldsymbol{\tau}_{f}^{2}\right)$
with backward simulation for US stock returns data with $T=1,000$,
$S=20$, and $K=1$, and number of particles $N=1,000$. Time denotes the time taken in seconds for one iteration of the method.\label{tab:Comparison-between-different real dataapproximate}}

\centering{}%
\begin{tabular}{cccc}
\hline
 & $I$ & $II$ & $III$\tabularnewline
\hline
{\footnotesize{}$\widehat{\textrm{IACT}}_{\textrm{MAX}}$} & {\footnotesize{}$23.99$} & {\footnotesize{}$1215.77$} & {\footnotesize{}$1228.99$}\tabularnewline
{\footnotesize{}$\textrm{TNV}_{\max}$} & {\footnotesize{}$152.82$} & {\footnotesize{}$8242.92$} & {\footnotesize{}$8971.63$}\tabularnewline
{\footnotesize{}$\textrm{RTNV}_{\max}$} & {\footnotesize{}$1$} & {\footnotesize{}$53.94$} & {\footnotesize{}$58.71$}\tabularnewline
\hline
{\footnotesize{}$\widehat{\textrm{IACT}}_{\textrm{MEAN}}$} & {\footnotesize{}$12.99$} & {\footnotesize{}$270.58$} & {\footnotesize{}$253.90$}\tabularnewline
{\footnotesize{}$\textrm{TNV}_{\textrm{MEAN}}$} & {\footnotesize{}$82.75$} & {\footnotesize{}$1834.53$} & {\footnotesize{}$1853.47$}\tabularnewline
{\footnotesize{}$\textrm{RTNV}_{MEAN}$} & {\footnotesize{}$1$} & {\footnotesize{}$22.17$} & {\footnotesize{}$22.40$}\tabularnewline
\hline
{\footnotesize{}$\textrm{Time}$} & {\footnotesize{}$6.37$} & {\footnotesize{}$6.78$} & {\footnotesize{}$7.30$}\tabularnewline
\hline
\end{tabular}
\end{table}

\subsubsection*{GARCH diffusion process with an Euler evolution transition density {\label{GARCHrealdata}}}
This section compares the following samplers: (I) $\textrm{PMMH}\left(\boldsymbol{\alpha},\boldsymbol{\tau}_{\epsilon}^{2},\boldsymbol{\tau}_{f}^{2},\boldsymbol{\mu}\right)+\textrm{PG}\left(\boldsymbol{f}_{1:T},\boldsymbol{\beta},\phi\right)$, (II)
$\textrm{PGAT}\left(\boldsymbol{f}_{1:T},\boldsymbol{\beta},\boldsymbol{\alpha},\boldsymbol{\tau}_{\epsilon}^{2},\boldsymbol{\mu},\phi,\boldsymbol{\tau}_{f}^{2}\right)$, and
(III) $\textrm{PGBS}\left(\boldsymbol{f}_{1:T},\boldsymbol{\beta},\boldsymbol{\alpha},\boldsymbol{\tau}_{\epsilon}^{2},\boldsymbol{\mu},\phi,\boldsymbol{\tau}_{f}^{2}\right)$
for the factor SV model with the idiosyncratic log-volatilities following GARCH diffusion processes which do not have closed form state transition densities.

Table \ref{tab:Inefficiency-factor-of real data-GARCH} in Section~\ref{S:FSV tables and figures} of the supplement shows the IACT
estimates for all the parameters 
for the factor SV model with the idiosyncratic log-volatilities following GARCH diffusion processes which do not have closed form state transition densities.
As for the models with Gaussian OU processes, all three samplers estimate the factor loading matrix $\boldsymbol{\beta}$ efficiently and with comparable IACT values.
The performance of the PMMH+PG sampler does not deteriorate for the real
data, whereas both the PGAT and PGBS samplers get worse in terms of the IACT values
for the remaining parameters.
Overall, the PMMH+PG sampler always has smaller IACT values than both the PGAT and PGBS
samplers for all the state transition parameters.

Table \ref{tab:Comparison-between-different real dataapproximateGARCH} summarises the estimation results for the GARCH diffusion model
and shows that in terms of  $\textrm{TNV}_{\textrm{MAX}}$, the PMMH+PG  is 19.56 and 22.11 times better than PGAT and PGBS samplers.
In terms of $\textrm{TNV}_{\textrm{MEAN}}$, the PMMH+PG is 25.84 and 28.01 times better than PGAT and PGBS, respectively.
This confirms the usefulness of the PMMH+PG samplers
for this class of the model.


\begin{table}[H]
\caption{Comparing different samplers in terms of Time Normalised Variance
with the Euler approximation for the state transition density for the GARCH diffusion
model. Sampler I: $\textrm{PMMH}\left(\boldsymbol{\alpha},\boldsymbol{\tau}_{\epsilon}^{2},\boldsymbol{\mu},\boldsymbol{\tau}_{f}^{2}\right)+\textrm{PG}\left(\boldsymbol{f}_{1:T},\boldsymbol{\beta},\boldsymbol{\phi}\right)$,
Sampler II: $\textrm{PGAT}\left(\boldsymbol{f}_{1:T},\boldsymbol{\beta},\boldsymbol{\alpha},\boldsymbol{\tau}_{\epsilon}^{2},\boldsymbol{\mu},\boldsymbol{\phi},\boldsymbol{\tau}_{f}^{2}\right)$,
Sampler III: $\textrm{PGBS}\left(\boldsymbol{f}_{1:T},\boldsymbol{\beta},\boldsymbol{\alpha},\boldsymbol{\tau}_{\epsilon}^{2},\boldsymbol{\mu},\boldsymbol{\phi},\boldsymbol{\tau}_{f}^{2}\right)$
for US stock returns data with $T=1000$, $S=20$, and $K=1$, and
number of particles $N=1000$. Time denotes the time taken in seconds
for one iteration of the method. \label{tab:Comparison-between-different real dataapproximateGARCH}}

\centering{}%
\begin{tabular}{cccc}
\hline 
 & {\small{}$I$} & {\small{}$II$} & {\small{}$III$}\tabularnewline
\hline 
{\small{}$\widehat{\textrm{IACT}}_{\textrm{MAX}}$} & {\small{}$147.16$} & {\small{}$3098.27$} & {\small{}$3257.52$}\tabularnewline
{\small{}$\textrm{TNV}_{\textrm{MAX}}$} & {\small{}$1392.13$} & {\small{}$27233.79$} & {\small{}$30783.56$}\tabularnewline
{\small{}$\textrm{RTNV}_{\textrm{MAX}}$} & {\small{}$1$} & {\small{}$19.56$} & {\small{}$22.11$}\tabularnewline
\hline 
{\small{}$\widehat{\textrm{IACT}}_{\textrm{MEAN}}$} & {\small{}$17.38$} & {\small{}$483.37$} & {\small{}$487.28$}\tabularnewline
{\small{}$\textrm{TNV}_{\textrm{MEAN}}$} & {\small{}$164.41$} & {\small{}$4248.82$} & {\small{}$4604.80$}\tabularnewline
{\small{}$\textrm{RTNV}_{\textrm{MEAN}}$} & {\small{}$1$} & {\small{}$25.84$} & {\small{}$28.01$}\tabularnewline
\hline 
{\small{}Time} & {\small{}$9.46$} & {\small{}$8.79$} & {\small{}$9.45$}\tabularnewline
\hline 
\end{tabular}
\end{table}

\section{Discussion\label{s:discussion}}

Our article introduces a flexible particle Markov chain Monte Carlo sampling
scheme for state space models where some parameters are generated without
conditioning on the states (PMMH) while other parameters are generated
conditional on the states (PG). Previous sampling schemes
 used PMMH or PG exclusively without combining both strategies. The
technical contribution of our article is to set out the required particle
framework for the flexible sampler and to obtain uniform ergodicity under
given assumptions. Our examples demonstrate that it is advantageous to use
this flexible sampling scheme to generate the parameters that are highly
correlated with the states without conditioning on the states (the PMMH
component) while the other parameters are generated by particle Gibbs (PG).

As we note in the introduction, in general, there are likely to be a number of
different sampling schemes that can solve the same problems addressed in our article, and which
sampler is best depends on a number of factors such as the model, the data set and  the number of observations.
We also note that our PMMH + PG approach can be further refined by using the data augmented PMMH and PG sampling schemes proposed
by \citet{Fearnhead2016} and the refined proposals for the PMMH sampling scheme by \citet{dahlin2015} and \citet{Nemeth2016}.

\section*{Acknowledgement}

The work of the authors was partially supported by an ARC Research Council
Grant DP120104014. The work of Robert Kohn and David Gunawan was also
partially supported by the ARC Center of Excellence grant CE140100049


\bibliographystyle{abbrvnat}
\bibliography{pfmcmc}
\pagebreak
\renewcommand{\thesscheme}{S\arabic{sscheme}}
\renewcommand{\thealgorithm}{S\arabic{algorithm}}
\renewcommand{\theremark}{S\arabic{remark}}
\renewcommand{\theequation}{S\arabic{equation}}
\renewcommand{\thetheorem}{S\arabic{theorem}}
\renewcommand{\thesection}{S\arabic{section}}
\renewcommand{\thepage}{S\arabic{page}}
\renewcommand{\thetable}{S\arabic{table}}
\renewcommand{\thefigure}{S\arabic{figure}}
\renewcommand{\theassumption}{S\arabic{assumption}}
\setcounter{page}{1}
\setcounter{section}{0}
\setcounter{equation}{0}
\setcounter{algorithm}{0}
\setcounter{table}{0}
\setcounter{figure}{0}

\section*{Online Supplement for \textquotedblleft A Flexible Particle Markov chain Monte Carlo method\textquotedblright}

We use the following notation in the supplement. Equation (1), Algorithm 1,
and Sampling Scheme 1, etc, refer to the main paper, while equation (S1),
Algorithm S1, and Sampling Scheme S1, etc, refer to the supplement. Section %
\ref{s:algorithms} lists some of the algorithms used in the main paper.
These algorithms are used in \cite{andrieuetal2010} and are included here
for notational consistency.  Section~\ref{s:theory} discusses the
convergence of Sampling Scheme~\ref{ssch:pmwg} to its target distribution.
Section~\ref{appendixbsi} discusses other choices of target distribution and
how it is straightforward to modify the results in the main paper to apply
to these distributions. Section \ref{sec:Sampling-Schemes-Factor model} discusses the target density of the PMMH+PG sampler for the multivariate factor SV model. Section \ref{ss: sampling scheme for the factor SV model} discusses the PMMH+PG sampling schemes for the factor SV model.
Section~\ref{S:FSV tables and figures} presents some additional tables and plots based on the analysis reported
in Sections~\ref{ss:simulationandapplication} and \ref{SS: US stock returns}.

\section{Algorithms\label{s:algorithms}}

The Sequential Monte Carlo algorithm used here is the same one as in \cite%
{andrieuetal2010} and is defined as follows.

\begin{algorithm}[Sequential Monte Carlo]
\label{alg:smc}\ \newline

\begin{enumerate}
\item For $t=1$:

\begin{enumerate}
\item Sample $X_1^i$ from $m_1^\theta(x)$, for $i=1,\dots,N$

\item Calculate the importance weights
\begin{equation*}
w_1^i = \frac{f_1^\theta(x_1^i)\,g_{\theta}(y_1|x_1^i)}{ m_1^\theta(x_1^i)}
\quad(i=1,\dots,N),
\end{equation*}
and normalize them to obtain $\bar w_1^{1:N}$.
\end{enumerate}

\item For $t=2,3,\dots$:

\begin{enumerate}
\item Sample the ancestral indices $A_{t-1}^{1:N} \sim \mathcal{M}\left(
a_{t-1}^{1:N}|\bar w_{t-1}^{1:N} \right)$

\item Sample $X_{t}^{i}$ from $m_{t}^{\theta }\left(
x|x_{t-1}^{a_{t-1}^{i}}\right)$, $i=1,\dots ,N$

\item Calculate the importance weights
\begin{equation*}
w_{t}^{i}=\frac{f_{\theta }\left( x_{t}^{i}|x_{t-1}^{a_{t-1}^{i}}\right)
\,g_{\theta }\left( y_{t}|x_{t}^{i}\right) }{m_{t}^{\theta }\left(
x_{t}^{i}|x_{t-1}^{a_{t-1}^{i}}\right) }\quad (i=1,\dots ,N)
\end{equation*}%
and normalize them to obtain $\overline{w}_{t}^{1:N}=w_{t}^{1:N}/\sum_{i=1}^{N}w_{t}^{i}$.
\end{enumerate}
\end{enumerate}
\end{algorithm}

\bigskip

Algorithm \ref{alg:condsmc} is the conditional sequential Monte Carlo
algorithm (as in \cite{andrieuetal2010}), consistent with $%
(x_{1:T}^{j},a_{1:T-1}^{j},j)$.

\begin{algorithm}[Conditional Sequential Monte Carlo]
\label{alg:condsmc}\ \newline
\end{algorithm}

\begin{enumerate}
\item Fix $X_{1:T}^{j}=x_{1:T}^{j}$ and $A_{1:T-1}^{j}=b_{1:T-1}^{j}$.

\item For $t=1$

\begin{enumerate}
\item Sample $X_{1}^{i}$ from $m_{1}^{\theta }(x)\mathrm{d}x$, for $i\in
\{1,\dots ,N\}\setminus \{b_{1}^{j}\}$.

\item Calculate the importance weights
\begin{equation*}
w_{1}^{i}=\frac{f_1^{\theta }(x_{1}^{i})\,g_{\theta }(y_{1}|x_{1}^{i})}{%
m_{1}^{\theta }(x_{1}^{i})}\quad (i=1,\dots ,N),
\end{equation*}%
and normalize them to obtain $\bar{w}_{1}^{1:N}$.
\end{enumerate}

\item For $t=2,\dots ,T$

\begin{enumerate}
\item Sample the ancestral indices
\begin{equation*}
A_{t-1}^{-(b_{t}^{j})}\sim \mathcal{M}\left( a^{(-b_{t}^{j})}|\bar{w}%
_{t-1}^{1:N}\right) .
\end{equation*}

\item Sample $X_{t}^{i}$ from $m_{t}^{\theta }\left(
x|x_{t-1}^{a_{t-1}^{i}}\right) \mathrm{d}x$, $i\in \{1,\dots ,N\}\setminus
\{b_{t}^{j}\}$.

\item Calculate the importance weights
\begin{equation*}
w_{t}^{i}=\frac{f_{\theta }\left( x_{t}^{i}|x_{t-1}^{a_{t-1}^{i}}\right)
\,g_{\theta }\left( y_{t}|x_{t}^{i}\right) }{m_{t}^{\theta }\left(
x_{t}^{i}|x_{t-1}^{a_{t-1}^{i}}\right) }\quad (i=1,\dots ,N)
\end{equation*}%
and normalized them to obtain $\bar{w}_{t}^{1:N}$.
\end{enumerate}
\end{enumerate}

\section{Ergodicity\label{s:theory}}

This section discusses the assumptions required for the particle filter.
We then discuss
convergence of Sampling Scheme~\ref{ssch:pmwg} in total variation norm and then
consider the stronger condition of uniform convergence.

We will use the generalization of Sampling Scheme \ref{ssch:pmwg} to the case where there may be multiple PMMH steps and there may be multiple Gibbs steps.
This was discussed in Section~\ref{s:pmwg}.
Let $\theta :=(\theta _{1},\ldots ,\theta _{p})$ be a partition of the parameter vector into $p$ components where each component may be a vector and
let $0\leq p_{1}\leq p$. Let $\Theta =\Theta_{1}\times \ldots \times \Theta _{p}$ be the corresponding partition of the parameter space.
We use the notation $\theta _{-i}:=(\theta _{1},\ldots ,\theta _{i-1},\theta _{i+1},\ldots,\theta_{p})$.
Sampling Scheme \ref{ssch:pmwg general} generates the parameters $\theta _{1},\ldots,\theta _{p_{1}}$ using PMMH steps and the parameters $\theta _{p_{1}+1},\ldots ,\theta _{p}$ using PG steps.
To simplify the discussion, we assume that both particle marginal Metropolis-Hastings steps and particle Gibbs steps are used, i.e., $0<p_{1}<p $.

\begin{sscheme}[PMMH+PG Sampler]
\label{ssch:pmwg general}
Given initial values for $U_{1:T}$, $J$ and $\theta$, one iteration of the MCMC involves the following steps.

\begin{enumerate}
\item (PMMH sampling) For $i=1,\ldots ,p_{1}$

Step $i$:
\begin{enumerate}
\item Sample $\theta _{i}^{\ast }\sim q_{i,1}(\cdot |U_{1:T}, J, \theta_{-i},\theta _{i}).$
\item Sample $U_{1:T}^{\ast }\sim \psi(\cdot |\theta_{-i},\theta _{i}^{\ast }).$
\item Sample $J^{\ast } \sim \tilde{\pi}^N (\cdot |U_{1:T}^{\ast}, \theta _{-i},\theta _{i}^{\ast }).$
\item Set $(\theta _{i}, U_{1:T}, J ) \leftarrow (\theta _{i}^{\ast }, U_{1:T}^{\ast},J^{\ast })$ with probability
\begin{align}
\alpha _{i} & \left( U_{1:T}, J, \theta _{i};U_{1:T}^{\ast },J^{\ast},\theta _{i}^{\ast }|\theta _{-i}\right)  =  1\wedge \nonumber \\
& \frac{\tilde{\pi}^{N}\left( U_{1:T}^{\ast } , \theta _{i}^{\ast}|\theta _{-i}\right) }{\tilde{\pi}^{N}\left( U_{1:T}, \theta _{i}|\theta_{-i}\right) }\,
\frac{q_{i}(U_{1:T}, \theta _{i}|U_{1:T}^{\ast }, J^{\ast}, \theta _{-i},\theta _{i}^{\ast })}{q_{i}(U_{1:T}^{\ast },  \theta _{i}^{\ast}|U_{1:T}, J,\theta _{-i},\theta _{i})} ,  \label{eq:PMwGaccprob}
\end{align}%
where%
\begin{eqnarray*}
q_{i}( U_{1:T}^{\ast },\theta _{i}^{\ast } | U_{1:T}, J, \theta
_{-i},\theta _{i}) & = & q_{i,1}(\theta _{i}^{\ast }|U_{1:T}, J, \theta _{-i},\theta
_{i})
 \psi(U_{1:T}^\ast|\theta _{-i},\theta _{i}^{\ast}).
\end{eqnarray*}

\end{enumerate}

\item (PG sampling) For $i=p_{1}+1,\ldots ,p$

Step $i$:
\begin{enumerate}
\item Sample $\theta _{i}^{\ast }\sim q_{i}(\cdot
|X_{1:T}^{J},B_{1:T-1}^{J},J,\theta _{-i},\theta _{i}).$

\item Set $\theta_i \leftarrow \theta_{i}^{\ast }$ with probability
\begin{eqnarray}
\lefteqn{\alpha_{i}\left(\theta_{i};\theta_{i}^{\ast}|X_{1:T}^{J},B_{1:T-1}^{J},J,\theta _{-i}\right) =}  \notag \\
&&1\wedge \frac{\tilde{\pi}^{N}\left( \theta_{i}^{\ast}|X_{1:T}^{J},B_{1:T-1}^{J},J,\theta _{-i}\right) }{\tilde{\pi}^{N}\left(\theta_{i}|X_{1:T}^{J},B_{1:T-1}^{J},J,\theta_{-i}\right) }
\times \frac{q_{i}(\theta _{i}|X_{1:T}^{J},B_{1:T-1}^{J},J,\theta_{-i},\theta_{i}^{\ast })}{q_{i}(\theta_{i}^{\ast }|X_{1:T}^{J},B_{1:T-1}^{J},J,\theta_{-i},\theta_{i}) } .
\label{eq:PMwGaccproba}
\end{eqnarray}
\end{enumerate}

\item Sample $U_{1:T}^{(-J)}\sim \tilde{\pi}^{N}(\cdot|X_{1:T}^{J},B_{1:T-1}^{J},J,\theta )$ using the conditional sequential Monte Carlo algorithm (CSMC) discussed in Section ~\ref{s:CSMC}.

\item Sample $J\sim \tilde{\pi}^{N}\left( \cdot |U_{1:T},\theta \right)$.
\end{enumerate}
\end{sscheme}

We now discuss the assumptions required for the particle filter.
For $t\geq1$, we define,
\begin{eqnarray*}
S_{t}^{\boldsymbol{\theta}}=\left(\boldsymbol{x}_{1:t}\in\boldsymbol{\chi}^{t}:\pi\left(\boldsymbol{x}_{1:t}|\boldsymbol{\theta}\right)>0\right) & \textrm{and} & Q_{t}^{\boldsymbol{\theta}}=\left\{ \boldsymbol{x}_{1:t}\in\boldsymbol{\chi}^{t}:\pi\left(\boldsymbol{x}_{1:t-1}|\boldsymbol{\theta}\right)m_{t}^{\boldsymbol{\theta}}\left(\boldsymbol{x}_{t}|\boldsymbol{x}_{1:t-1},\boldsymbol{y}_{1:t}\right)>0\right\} .
\end{eqnarray*}
Assumption \ref{assu:propstatespace} ensures that the proposal densities
 $\pi\left(\boldsymbol{x}_{1:t-1}|\boldsymbol{\theta}\right)m_{t}^{\boldsymbol{\theta}}\left(\boldsymbol{x}_{t}|\boldsymbol{x}_{1:t-1},\boldsymbol{y}_{1:t}\right)$
can be used to approximate $\pi\left(\boldsymbol{x}_{1:t}|\boldsymbol{\theta}\right)$
for $t\geq1$.
\begin{assumption}
\citep{andrieuetal2010} We assume that $S_{t}^{\boldsymbol{\theta}}\subseteq Q_{t}^{\boldsymbol{\theta}}$
for any $\boldsymbol{\theta}\in\boldsymbol{\Theta}$ and $t=1,...,T$
\label{assu:propstatespace}
\end{assumption}
Assumption \ref{assu:propstatespace} is always satisfied in our implementation
because we use the bootstrap filter with $p\left(\boldsymbol{x}_{t}|\boldsymbol{x}_{t-1},\boldsymbol{\theta}\right)$
as a proposal density which are positive everywhere.

We also require Assumption \ref{assu:resampling} given below.
\begin{assumption}
\citep{andrieuetal2010} For any $k=1,...,N$ and $t=1,..,T$, the resampling
scheme $\mathcal{M}\left(a_{t-1}^{1:N}|\bar{w}_{t-1}^{1:N}\right)$
satisfies $\mathcal{M}\left(a_{t-1}^{k}=j|\bar{w}_{t-1}^{1:N}\right) = \bar{w}_{t-1}^{j}$.\label{assu:resampling}
\end{assumption}
Assumption \ref{assu:resampling} is satisfied by the popular resampling schemes, such
as multinomial, systematic, residual resampling.

Under Assumption \ref{assu:resampling}, it is straightforward to show
that the algorithm samples from the target density of the random variable $%
U_{1:T}^{\left( -J\right) }=\left( X_{1}^{(-B_{1}^{J})},\ldots
,X_{T}^{(-B_{T}^{J})},A_{1}^{(-B_{2}^{J})},\ldots
,A_{T-1}^{(-B_{T}^{J})}\right) ,$ conditional on $U_{1:T}^{J}$ and index $J$
given by
\begin{eqnarray*}
\lefteqn{\tilde{\pi}^{N}\left( u_{1:T}^{(-j)}|x_{1:T},b_{1:T-1},j,\theta
\right) =} \\
&&\frac{\psi \left( u_{1:T}|\theta \right) }{m_{1}^{\theta }\left(
x_{1}^{b_{1}}\right) \text{ }\prod_{t=2}^{T}\bar{w}%
_{t-1}^{a_{t-1}^{i}}m_{t}^{\theta }\left(
x_{t}^{b_{t}}|x_{t-1}^{a_{t-1}^{b_{t}}}\right) };
\end{eqnarray*}%
see \cite{andrieuetal2010} for details.

We now discuss convergence of Sampling Scheme~\ref{ssch:pmwg general} in total variation norm and then
consider the stronger condition of uniform convergence.
Note that, by construction, Sampling Scheme~\ref{ssch:pmwg general} has the
stationary distribution
\begin{equation*}
\tilde{\pi}^{N}\left( x_{1:T},b_{1:T-1},j,u_{1:T}^{(-j)},\theta \right)
\end{equation*}%
defined in (\ref{eq:targetdist}). From \cite{robertsrosenthal2004} Theorem
4, irreducibility and aperiodicity are sufficient conditions for the Markov
chain obtained using Sampling Scheme~\ref{ssch:pmwg general} to converge to its
stationary distribution in total variation norm for $\tilde{\pi}^{N}$-almost
all starting values. These conditions must be checked for a particular
sampler and it is often straightforward to do so. We will relate Sampling
Scheme \ref{ssch:pmwg general} to the particle Metropolis within Gibbs sampling
scheme defined below.

\begin{sscheme}[Ideal]
\label{ssch:idealpmwg}

\begin{description}
\item Given initial values for $U_{1:T}$, $J$ and $\theta $, one iteration
of the MCMC sampling scheme involves the following steps

\item[1.] (PMMH sampling) For $i=1,\ldots ,p_{1}$

Step $i$:

\begin{description}
\item[(a)] Sample $\theta _{i}^{\ast }\sim q_{i,1}(\cdot |U_{1:T},J,\theta
_{-i},\theta _{i}).$

\item[(b)] Sample $\left( J^{\ast },U_{1:T}^{\ast }\right) \sim \tilde{\pi}%
^{N}\left( \cdot |\theta _{-i},\theta _{i}^{\ast }\right) $.

\item[(c)] Set $\left( \theta _{i}, U_{1:T}, J\right) \leftarrow \left( \theta _{i}^{\ast }, U_{1:T}^{\ast}, J^{\ast }\right)$ with probability
\begin{eqnarray}
\lefteqn{\widetilde{\alpha _{i}}\left( U_{1:T},J,\theta _{i};U_{1:T}^{\ast
},J^{\ast },\theta _{i}^{\ast }|\theta _{-i}\right) =}  \notag \\
& & 1\wedge \frac{\tilde{\pi}^{N}\left( \theta _{i}^{\ast }|\theta
_{-i}\right) }{\tilde{\pi}^{N}\left( \theta _{i}|\theta _{-i}\right) }\,%
\frac{q_{i,1}(\theta _{i}|U_{1:T}^{\ast },J^{\ast },\theta _{-i},\theta
_{i}^{\ast })}{q_{i,1}(\theta _{i}^{\ast }|U_{1:T},J,\theta _{-i},\theta
_{i})}  \label{eq:idealPMwGaccprob}
\end{eqnarray}
\end{description}

\item[2.] (PG sampling) For $i=p_{1}+1,\ldots ,p$

Step $i$:

\begin{description}
\item[(a)] Sample $\theta _{i}^{\ast }\sim q_{i}(\cdot
|X_{1:T}^{J},B_{1:T-1}^{J},J,\theta _{-i},\theta _{i}).$

\item[(b)] Set $\theta_{i} \leftarrow \theta_{i}^{\ast }$ with probability
\begin{eqnarray}
\lefteqn{\alpha _{i}\left\{ \theta _{i};\theta _{i}^{\ast
}|X_{1:T}^{J},B_{1:T-1}^{J},J,\theta _{-i}\right\} =}  \notag \\
&&1\wedge \frac{\tilde{\pi}^{N}\left( \theta _{i}^{\ast
}|X_{1:T}^{J},B_{1:T-1}^{J},J,\theta _{-i}\right) }{\tilde{\pi}^{N}\left(
\theta _{i}|X_{1:T}^{J},B_{1:T-1}^{J},J,\theta _{-i}\right) }  \notag \\
&&\frac{q_{i}(\theta _{i}|X_{1:T}^{J},B_{1:T-1}^{J},J,\theta _{-i},\theta
_{i}^{\ast })}{q_{i}(\theta _{i}^{\ast }|X_{1:T}^{J},B_{1:T-1}^{J},J,\theta
_{-i},\theta _{i})}.  \label{eq:idealPMwGaccproba}
\end{eqnarray}
\end{description}

\item[3.]

\begin{description}
\item Sample $U_{1:T}^{(-J)}\sim \tilde{\pi}^{N}(\cdot
|X_{1:T}^{J},B_{1:T-1}^{J},J,\theta )$ using Algorithm~\ref{alg:condsmc}.
\end{description}

\item[4.]

\begin{description}
\item Sample $J\sim \tilde{\pi}^{N}\left( \cdot |U_{1:T},\theta \right) $.
\end{description}
\end{description}
\end{sscheme}

We call Sampling Scheme~\ref{ssch:idealpmwg} an \textit{ideal} particle
sampling scheme because in Part 1 Step $i$(b) it generates the particles $%
U_{1:T}^{\ast }$ from their conditional distribution $\tilde{\pi}^{N}\left(
\cdot |\theta _{-i},\theta _{i}^{\ast }\right) $ instead of using a
Metropolis-Hastings proposal. Thus comparing Sampling Schemes \ref{ssch:pmwg general}
and \ref{ssch:idealpmwg} allows us to concentrate on the effect of the
Metropolis-Hastings proposal for the particles on the convergence of the
sampler.

\begin{remark}
\cite{andrieuroberts2009} and \cite{andrieuvihola2012} discuss the
relationship between PMMH sampling schemes with one block of parameters and
an \textit{ideal} Metropolis-Hastings sampling scheme not involving the
particles. Sampling Schemes \ref{ssch:pmwg general} and \ref{ssch:idealpmwg} are
more general. Our approach is similar to,  but generalizes, the
results in \cite{andrieuroberts2009} and \cite{andrieuvihola2012} to more
complex sampling schemes.
\end{remark}

To develop the theory of Sampling Schemes \ref{ssch:pmwg general} and \ref%
{ssch:idealpmwg} we require the following definitions. Let $\left\{
V^{\left( n\right) },n=1,2,\ldots \right\} $ be the iterates of a Markov
chain defined on the state space $\mathcal{V}:=\mathcal{U}\times \mathbb{N}%
\times \Theta $. For $i=1,\ldots ,p$, let $K_{i}(v;\cdot )$ be the
substochastic transition kernel of the $i$th step of Sampling Scheme~\ref%
{ssch:pmwg general} that defines the probabilities for accepted Metropolis-Hastings
moves and define%
\begin{equation*}
K:=K_{1}K_{2}\ldots K_{p}
\end{equation*}%
to be the substochastic transition kernel that defines the probabilities for
accepted Metropolis-Hastings moves. Note that probabilities involving the
substochastic kernels provide lower bounds on the probabilities for the
transition kernel of the corresponding Markov chain.

\bigskip

For $i=1,\ldots ,p_{1}$%
\begin{eqnarray*}
\lefteqn{K_{i}\left( U_{1:T},J,\theta _{-i},\theta _{i};U_{1:T}^{\ast
},J_{i}^{\ast },\theta _{-i},\theta _{i}^{\ast }\right) =} \\
&&\tilde{\pi}^{N}(J^{\ast }|U_{1:T}^{\ast },\theta _{-i},\theta _{i}^{\ast
})q_{i}(U_{1:T}^{\ast },\theta _{i}^{\ast }|U_{1:T},J,\theta _{-i},\theta
_{i})\times \alpha _{i}\left( U_{1:T},J,\theta _{i};U_{1:T}^{\ast },J^{\ast },\theta
_{i}^{\ast }|\theta _{-i}\right) .
\end{eqnarray*}%
Similarly, for $i=1,\ldots ,p$, let $\widetilde{K}_{i}(v;\cdot )$ be the
substochastic transition kernel of the $i$th step of Sampling Scheme~\ref%
{ssch:idealpmwg} that defines the probabilities for accepted
Metropolis-Hastings moves and define%
\begin{equation*}
\widetilde{K}=\widetilde{K}_{1}\widetilde{K}_{2}\ldots \widetilde{K}_{p},
\end{equation*}%
where the kernels $K_{i}$ and $\widetilde{K}_{i}$ only differ for $%
i=1,\ldots ,p_{1}$.

The next theorem gives a sufficient condition for Sampling Scheme~\ref{ssch:pmwg general} to be irreducible and aperiodic and
is similar to Theorem 1 of \cite{andrieuroberts2009}).
\begin{theorem}
\label{th:irrandaperiodic} If $\widetilde{K}$ is irreducible and
aperiodic then $K$ is irreducible and aperiodic.
\begin{proof}
For $i=1,\ldots ,p_{1}$,
$\tilde{\pi}^{N}\left( \cdot |\theta _{-i},\theta _{i}^{\ast }\right)
 \ll \psi\left (\cdot|\theta _{-i},\theta _{i}^{\ast }
\right )
$
and the result now follows from Assumption 1 of \cite{andrieuetal2010}.
\end{proof}
\end{theorem}

We now follow the approach in \cite{andrieuroberts2009} and show the uniform
erdogicity of the sampling schemes by giving sufficient conditions for the
existence of minorization conditions for Sampling Scheme~\ref{ssch:pmwg general}.
These minorization conditions are equivalent to uniform ergodicity by
Theorem 8 of \cite{robertsrosenthal2004}. The results use the following
technical lemmas.

\begin{lemma}
\label{l:ergtechnical1}For $i=1,\ldots ,p_{1},$%
\begin{eqnarray*}
\alpha _{i}\left( U_{1:T},J,\theta _{i};U_{1:T}^{\ast },J^{\ast
},\theta _{i}^{\ast }|\theta _{-i}\right) \geq
\left\{ 1\wedge \frac{\tilde{\pi}^{N}\left( U_{1:T}^{\ast }|\theta
_{-i},\theta _{i}^{\ast }\right) \psi(U_{1:T}|
\theta _{-i},\theta _{i})}{\tilde{\pi}^{N}\left( U_{1:T}|\theta
_{-i},\theta _{i}\right) \psi(U_{1:T}^{\ast }|\theta
_{-i},\theta _{i}^{\ast })}\right\}
 \times
\widetilde{\alpha _{i}}\left( U_{1:T},J,\theta _{i};U_{1:T}^{\ast
},J^{\ast },\theta _{i}^{\ast }|\theta _{-i}\right)
\end{eqnarray*}
\end{lemma}

\begin{proof}
From (\ref{eq:PMwGaccprob}),%
\begin{eqnarray*}
\lefteqn{\alpha _{i}\left( U_{1:T},J,\theta _{i};U_{1:T}^{\ast },J^{\ast
},\theta _{i}^{\ast }|\theta _{-i}\right) } \\
&=&1\wedge \frac{\tilde{\pi}^{N}\left( U_{1:T}^{\ast },\theta _{i}^{\ast
}|\theta _{-i}\right) }{\tilde{\pi}^{N}\left( U_{1:T},\theta _{i}|\theta
_{-i}\right) }\,\frac{q_{i}(U_{1:T},\theta _{i}|U_{1:T}^{\ast },J^{\ast
},\theta _{-i},\theta _{i}^{\ast })}{q_{i}(U_{1:T}^{\ast },\theta _{i}^{\ast
}|U_{1:T},J,\theta _{-i},\theta _{i})} \\
&=&  1\wedge \frac{\tilde{\pi}^{N}\left( U_{1:T}^{\ast }|\theta
_{-i},\theta _{i}^{\ast }\right) \psi(U_{1:T}|
\theta _{-i},\theta _{i})}{\tilde{\pi}^{N}\left( U_{1:T}|\theta
_{-i},\theta _{i}\right) \psi(U_{1:T}^{\ast }|\theta
_{-i},\theta _{i}^{\ast })}
\times
\frac{\tilde{\pi}^{N}\left( \theta _{i}^{\ast }|\theta _{-i}\right)
q_{i,1}(\theta _{i}|U_{1:T}^{\ast },J^{\ast },\theta _{-i},\theta _{i}^{\ast
})}{\tilde{\pi}^{N}\left( \theta _{i}|\theta _{-i}\right) q_{i,1}(\theta
_{i}^{\ast }|U_{1:T},J,\theta _{-i},\theta _{i})}\\
&\geq & 1\wedge \frac{\tilde{\pi}^{N}\left( U_{1:T}^{\ast }|\theta
_{-i},\theta _{i}^{\ast }\right) \psi(U_{1:T}|
\theta _{-i},\theta _{i})}{\tilde{\pi}^{N}\left( U_{1:T}|\theta
_{-i},\theta _{i}\right) \psi(U_{1:T}^{\ast }|\theta
_{-i},\theta _{i}^{\ast })}
\times
1\wedge \frac{\tilde{\pi}^{N}\left( \theta _{i}^{\ast }|\theta _{-i}\right)
q_{i,1}(\theta _{i}|U_{1:T}^{\ast },J^{\ast },\theta _{-i},\theta _{i}^{\ast
})}{\tilde{\pi}^{N}\left( \theta _{i}|\theta _{-i}\right) q_{i,1}(\theta
_{i}^{\ast }|U_{1:T},J,\theta _{-i},\theta _{i})}\\
&=&
\left\{ 1\wedge \frac{\tilde{\pi}^{N}\left( U_{1:T}^{\ast }|\theta
_{-i},\theta _{i}^{\ast }\right) \psi(U_{1:T}|
\theta _{-i},\theta _{i})}{\tilde{\pi}^{N}\left( U_{1:T}|\theta
_{-i},\theta _{i}\right) \psi(U_{1:T}^{\ast }|\theta
_{-i},\theta _{i}^{\ast })}
\right\} \times
\widetilde{\alpha _{i}}\left( U_{1:T},J,\theta _{i};U_{1:T}^{\ast
},J^{\ast },\theta _{i}^{\ast }|\theta _{-i}\right)
\end{eqnarray*}
\end{proof}

\begin{lemma}
\label{l:ergtechnical2}Suppose that
\begin{equation}
\frac{\tilde{\pi}^{N}\left( U_{1:T}^{\ast }|\theta \right) }{%
\psi(U_{1:T}^{\ast }|\theta )} \leq \gamma<\infty
\label{eq:convcond}
\end{equation}%
for all $U_{1:T}^{\ast }\in \mathcal{U}, \theta  \in
\mathcal{S}$. Then, for $i=1,\ldots ,p_{1}$, each Markov transition kernel $%
K_i$ satisfies
\begin{equation}
K_{i}\geq \gamma^{-1}\widetilde{K}_{i}  \label{eq:kernelbound1}
\end{equation}%
and hence%
\begin{equation}
K\geq \gamma ^{-p_1}\widetilde{K}.
\label{eq:kernelbound2}
\end{equation}
\end{lemma}
\begin{proof}
Fix $i\in \left\{ 1,\ldots ,p_{1}\right\} $ and let $A\in \mathcal{B}\left(
\mathcal{U}\right) $, $J,J^{\ast }\in \left\{ 1,\ldots ,N\right\} $ and $%
B\in \mathcal{B}\left( \Theta _{i}\right) $. Then%
\begin{eqnarray*}
\lefteqn{K_{i}\left( U_{1:T},J,\theta _{-i},\theta _{i};A,J^{\ast },\theta
_{-i},B\right) } \\
&=&\dint_{A\times B}\tilde{\pi}^{N}(J^{\ast }|U_{1:T}^{\ast },\theta
_{-i},\theta _{i}^{\ast })q_{i}(U_{1:T}^{\ast },\theta _{i}^{\ast
}|U_{1:T},J,\theta _{-i},\theta _{i})\times \\
&&\alpha _{i}\left( U_{1:T},J,\theta _{i};U_{1:T}^{\ast },J^{\ast },\theta
_{i}^{\ast }|\theta _{-i}\right) dU_{1:T}^{\ast }d\theta _{i}^{\ast
} \\
&\geq &\dint_{A\times B}\tilde{\pi}^{N}(J^{\ast }|U_{1:T}^{\ast },\theta
_{-i},\theta _{i}^{\ast })q_{i}(U_{1:T}^{\ast },\theta _{i}^{\ast
}|U_{1:T},J,\theta _{-i},\theta _{i})\times \\
&&\left\{ 1\wedge \frac{\tilde{\pi}^{N}\left( U_{1:T}^{\ast }|\theta
_{-i},\theta _{i}^{\ast }\right) \psi(U_{1:T}|
\theta _{-i},\theta _{i})}{\tilde{\pi}^{N}\left( U_{1:T}|\theta
_{-i},\theta _{i}\right) \psi(U_{1:T}^{\ast }|\theta
_{-i},\theta _{i}^{\ast })}\right\} \times
\widetilde{\alpha _{i}}\left( U_{1:T},J,\theta _{i};U_{1:T}^{\ast
},J^{\ast },\theta _{i}^{\ast }|\theta _{-i}\right) dU_{1:T}^{\ast
}d\theta _{i}^{\ast } \\
&\geq &\gamma^{-1}\dint_{A\times B}\tilde{\pi}^{N}\left( U_{1:T}^{\ast
},J^{\ast }|\theta _{-i},\theta _{i}^{\ast }\right) q_{i,1}(\theta
_{i}^{\ast }|U_{1:T},J,\theta _{-i},\theta _{i})\times \widetilde{\alpha _{i}}\left( U_{1:T},J,\theta _{i};U_{1:T}^{\ast
},J^{\ast },\theta _{i}^{\ast }|\theta _{-i}\right) d U_{1:T}^{\ast
}d\theta _{i}^{\ast }\\
&=&\gamma^{-1}\widetilde{K}_{i}\left( U_{1:T},J,\theta _{-i},\theta
_{i};A,J^{\ast },\theta _{-i},B\right) \text{,}
\end{eqnarray*}
which proves (\ref{eq:kernelbound1}). Apply (\ref{eq:kernelbound1}) for each
$i$ to get (\ref{eq:kernelbound2})
\end{proof}
Lemma \ref{l:ergtechnical2} can be used to find sufficient conditions for
the existence of minorization conditions for Sampling Scheme~\ref{ssch:pmwg general}
as given in the theorem below, which is similar to \cite{andrieuroberts2009}
, Theorem 8. Let $\mathcal{L}_{N}\{V^{\left( n\right) }\in \cdot \}$ denote
the sequence of distribution functions of the random variables $\{V^{\left(
n\right) }:\,n=1,2,\dots \}$, generated by Sampling Scheme~\ref{ssch:pmwg general},
and let $|\cdot |_{TV}$ be total variation norm.

\begin{theorem}
\label{th:minorization}Suppose that Sampling Scheme~\ref{ssch:idealpmwg}
satisfies the following minorization condition: there exists a constant $%
\epsilon >0$, a number $n_{0}\geq 1$, and a probability measure $\nu $ on $%
\mathcal{V}$ such that $\widetilde{K}^{n_{0}}(v;A)\geq \epsilon \,\nu (A)$
for all $v\in \mathcal{V},A\in \mathcal{B}\left( \mathcal{V}\right) $.
Suppose also that the conditions of Lemma \ref{l:ergtechnical2} are
satisfied. Then Sampling Scheme~\ref{ssch:pmwg general} satisfies the minorization
condition%
\begin{equation*}
K^{n_{0}}(v;A)\geq \gamma
^{-p_1 n_{0}}\epsilon \nu (A)
\end{equation*}%
and for all starting values for the Markov Chain%
\begin{equation*}
\left\vert \mathcal{L}_{N}\{V^{\left( n\right) }\in \cdot \}-\tilde{\pi}%
^{N}\left\{ V^{\left( n\right) }\in \cdot \right\} \right\vert _{TV}\leq
\left( 1-\delta \right) ^{\left\lfloor n/n_{0}\right\rfloor },
\end{equation*}%
where $0<\delta <1$ and $\left\lfloor n/n_{0}\right\rfloor $ is the greatest
integer not exceeding $n/n_{0}$.
\end{theorem}

\begin{proof}
To show the first part, suppose $\widetilde{K}^{n_{0}}(v;A)\geq \epsilon
\,\nu (A)$ for all $v\in \mathcal{V},A\in \mathcal{B}\left( \mathcal{V}
\right) $. Fix $v\in \mathcal{V},A\in \mathcal{B}\left( \mathcal{V}\right) $%
. Applying Lemma \ref{l:ergtechnical2} repeatedly gives%
\begin{eqnarray*}
K^{n_{0}}(v;A) &\geq &\gamma^{-p_1 n_{0}}
\widetilde{K}^{n_{0}}(v;A)
\geq \,\gamma
^{-p_1 n_{0}}\epsilon \nu
(A)
\end{eqnarray*}%
as required. The second part follows from the first part and \cite%
{robertsrosenthal2004}, Theorem 8.
\end{proof}

Lemma \ref{l:boundedlikelihood} gives sufficient conditions for Lemma \ref%
{l:ergtechnical2} to hold.
The first condition is from \cite{andrieuetal2010}.

\begin{lemma}
\label{l:boundedlikelihood}Suppose

\begin{description}
\item[(i)] There is a sequence of finite, positive constants $%
\{c_{t}:t=1,\dots ,T\}$ such that for any $x_{1:t}\in \mathcal{S}_{t}(\theta
)$ and all $\theta \in \mathcal{S}$, $f_{\theta }(x_{t}|x_{t-1})g_{\theta
}(y_{t}|x_{t})\leq c_{t}\,m_{t}^{\theta }(x_{t}|x_{t-1})$.

\item[(ii)] There exists an $\varepsilon >0$ such that for all $\theta \in
\mathcal{S}$, $p\left( y_{1:T}|\theta \right) >\varepsilon $.
\end{description}

If (i)\ and (ii)\ hold, then the conditions in Lemma \ref{l:ergtechnical2}
are satisfied.
\end{lemma}
\begin{proof}
Part (i) implies that for all $\theta \in \mathcal{S}$ and all $U_{1:T}\in
\mathcal{U}$, $Z(U_{1:T},\theta )\leq \dprod_{t=1}^{T}c_{t}$. Hence Part
(ii) implies that%
\begin{equation*}
\frac{Z(U_{1:T},\theta )}{p\left( y_{1:T}|\theta \right) }<\frac{%
\dprod_{t=1}^{T}c_{t}}{\varepsilon }.
\end{equation*}%
 From \eqref{eq:accprobsimplify1},
\begin{align*}\frac{\tilde{\pi}^{N}\left(
U_{1:T}^{\ast }|\theta \right) }{\psi \left( U_{1:T}^{\ast }|\theta \right) }
&=\frac{Z(U_{1:T},\theta )}{p\left( y_{1:T}|\theta \right) }
\end{align*}
 giving the result.
\end{proof}

\begin{remark}

The results above can be modified for the factor stochastic volatility model given in Section \ref{PMMH+PG SV factor} in a straightforward way.
Details are
available from the authors on request.
\end{remark}

\begin{remark}
If the states are sampled using backward simulation, similar arguments can be
applied to obtain corresponding results (see Section \ref{appendixbsi}).
The mathematical details of the derivation use the results in \cite%
{olssonryden2011} and \cite{lindstenschon2012}.
\end{remark}

\section{Backward simulation\label{appendixbsi}}

\cite{godsilletal2004} introduce the \textit{backward simulation} algorithm
which samples the indices\linebreak $J_{T},J_{T-1},\dots ,J_{1}$
sequentially, and differs from \textit{ancestral tracing }which samples one
index $J$ and traces back its ancestral lineage. The \textit{backward
simulation} algorithm (Algorithm \ref{alg:bsim} below) is used in the PMCMC
setting by \cite{olssonryden2011} (in the PMMH algorithm) and \cite%
{lindstenschon2012} (in the PG algorithm). \cite{chopinsingh2013} studied
the PG algorithm with \textit{backward simulation} and found that it yields
a smaller autocorrelation than the corresponding algorithm using \textit{\
ancestral tracing}. Moreover, it is more robust to the resampling scheme
(multinomial resampling, systematic resampling, residual resampling or
stratified resampling) used in the resampling step of the algorithm.

\begin{algorithm}[Backward Simulation]
\label{alg:bsim}

\begin{enumerate}
\item Sample $J_{T}=j_{t}$ conditional on $u_{1:T}$, with probability
proportional to $w_{T}^{j_{T}}$, and choose $x_{T}^{j_{T}}$;

\item For $t=T-1,\dots ,1$, sample $J_{t}=j_{t}$ conditional on%
\begin{equation*}
(u_{1:t},j_{t+1:T},x_{t+1}^{j_{t+1}},\dots ,x_{T}^{j_{T}}),
\end{equation*}%
with probability proportional to $w_{t}^{j_{t}}f_{\theta
}(x_{t+1}^{j_{t+1}}|x_{t}^{j_{t}})$, and choose $x_{t}^{j_{t}}$.
\end{enumerate}
\end{algorithm}

We denote the particles selected and the trajectory selected by $%
x_{1:T}^{j_{1:T}}=(x_{1}^{j_{1}},\dots ,x_{T}^{j_{T}})$ and $j_{1:T}$,
respectively. With some abuse of notation, we denote
\begin{equation*}
x_{1:T}^{(-j_{1:T})}=\left\{ x_{1}^{(-j_{1})},\ldots
,x_{T}^{(-j_{T})}\right\} .
\end{equation*}%
It will simplify the notation to sometimes use the following one-to-one
transformation
\begin{equation*}
\left( u_{1:T},j_{1:T}\right) \leftrightarrow \left\{
x_{1:T}^{j_{1:T}},j_{1:T},x_{1:T}^{(-j_{1:T})},a_{1:T-1}\right\} ,
\end{equation*}%
and switch between the two representations and use whichever is more
convenient.

The augmented space in this case consists of the particle filter variables $%
U_{1:T}$ and the sampled trajectory $J_{1:T}$ and PMCMC methods using
\textit{backward simulation} target the following density
\begin{eqnarray}
\lefteqn{\tilde{\pi}_{BSi}^{N}\left(
x_{1:T},j_{1:T},x_{1:T}^{(-j_{1:T})},a_{1:T-1},\theta \right) \mathrel{:=} }
\notag \\
&&\frac{p(x_{1:T},\theta |y_{1:T})}{N^{T}}\frac{\psi \left( u_{1:T}|\theta
\right) }{m_{1}^{\theta }\left( x_{1}^{b_{1}}\right) \text{ }\prod_{t=2}^{T}%
\bar{w}_{t-1}^{a_{t-1}^{i}}m_{t}^{\theta }\left(
x_{t}^{b_{t}}|x_{t-1}^{a_{t-1}^{b_{t}}}\right) } \times  \notag \\
& & \prod_{t=2}^{T}\frac{%
w_{t}^{a_{t-1}^{j_{t}}}f(x_{t}^{j_{t}}|x_{t-1}^{a_{t-1}^{j_{t}}})}{%
\sum_{i=1}^{N}w_{t}^{a_{t-1}^{i}}f(x_{t}^{i}|x_{t-1}^{a_{t-1}^{i}})}.
\label{eq:targetbsi}
\end{eqnarray}

\cite{olssonryden2011} show that, under Assumption 2 of \cite%
{andrieuetal2010},%
\begin{equation*}
\tilde{\pi}_{BSi}^{N}\left(
x_{1:T},j_{1:T},x_{1:T}^{(-j_{1:T})},a_{1:T-1},\theta \right)
\end{equation*}%
has the following marginal distribution
\begin{equation*}
\tilde{\pi}_{BSi}^{N}\left( x_{1:T},j_{1:T},\theta \right) =\frac{%
p(x_{1:T},\theta |y_{1:T})}{N^{T}},
\end{equation*}%
and hence
\begin{equation*}
\tilde{\pi}_{BSi}^{N}\left( x_{1:T},\theta \right) =p(x_{1:T},\theta
|y_{1:T}).
\end{equation*}%

The conditional sequential Monte Carlo algorithm used in the backward
simulation also changes. It is given in \cite{lindstenetal2014} and
generates from the full conditional distribution%
\begin{equation*}
\tilde{\pi}_{BSi}^{N}\left(
x_{1:T}^{(-j_{1:T})},a_{1:T-1}|x_{1:T},j_{1:T},\theta \right) \text{.}
\end{equation*}
The general sampler using \textit{backward simulation} is analogous to the
\textit{ancestral tracing} general sampler, but on an expanded space.

\begin{sscheme}[general-BSi]
\label{ssch:pmwgbsi}
\begin{description}
\item Given initial values for $U_{1:T}$, $J_{1:T}$ and $\theta $, one
iteration of the MCMC involves the following steps

\item[1.] (PMMH sampling) For $i=1,\ldots ,p_{1}$

Step $i$:

\begin{description}
\item[(a)] Sample $\theta _{i}^{\ast }\sim q_{BSi,i,1}(\cdot
|U_{1:T},J_{1:T},\theta _{-i},\theta _{i}).$
\item[(b)] Sample $U_{1:T}^{\ast }\sim \psi(\cdot
|\theta _{-i},\theta _{i}^{\ast }).$
\item[(c)] Sample $J_{1:T}^{\ast }$ from $\tilde{\pi}_{BSi}^{N}(\cdot
|U_{1:T}^{\ast },\theta _{-i},\theta _{i}^{\ast }).$
\item[(d)] Set $\left(\theta _{i}, U_{1:T}, J_{1:T}\right) \leftarrow \left(\theta _{i}^{\ast }, U_{1:T}^{\ast}, J_{1:T}^{\ast }\right)$ with probability
\begin{eqnarray}
\lefteqn{\alpha _{i}\left( U_{1:T},J_{1:T},\theta _{i};U_{1:T}^{\ast
},J_{1:T}^{\ast },\theta _{i}^{\ast }|\theta _{-i}\right) =}
\label{eq:PMwGBSiaccprob} \\
&&1\wedge \frac{\tilde{\pi}_{BSi}^{N}\left( U_{1:T}^{\ast },\theta
_{i}^{\ast }|\theta _{-i}\right) }{\tilde{\pi}_{BSi}^{N}\left(
U_{1:T},\theta _{i}|\theta _{-i}\right) }\,\frac{q_{BSi,i}(U_{1:T},\theta
_{i}|U_{1:T}^{\ast },J_{1:T}^{\ast },\theta _{-i},\theta _{i}^{\ast })}{%
q_{BSi,i}(U_{1:T}^{\ast },\theta _{i}^{\ast }|U_{1:T},J_{1:T},\theta
_{-i},\theta _{i})}  \notag
\end{eqnarray}%
where%
\begin{align*}
q_{Bsi,i}(U_{1:T}^{\ast },\theta _{i}^{\ast
}|U_{1:T},J_{1:T},\theta _{-i},\theta _{i})=
& q_{BSi,i,1}(\theta _{i}^{\ast }|U_{1:T},J_{1:T},\theta _{-i},\theta
_{i})\psi(U_{1:T}^{\ast }|\theta _{-i},\theta
_{i}^{\ast }).
\end{align*}
\end{description}

\item[2.] (PG or PMwG sampling) For $i=p_{1}+1,\ldots ,p$

Step $i$:

\begin{description}
\item[(a)] Sample $\theta _{i}^{\ast }\sim q_{i}(\cdot
|X_{1:T}^{J},B_{1:T-1}^{J},J,\theta _{-i},\theta _{i}).$

\item[(b)] Set $\theta _{i} \leftarrow \theta _{i}^{\ast }$ with probability
\begin{eqnarray*}
&&\alpha _{i}\left( \theta _{i};\theta _{i}^{\ast
}|X_{1:T}^{J},B_{1:T-1}^{J},J,\theta _{-i}\right) \\
&=&1\wedge \frac{\tilde{\pi}^{N}\left( \theta _{i}^{\ast
}|X_{1:T}^{J},B_{1:T-1}^{J},J,\theta _{-i}\right) }{\tilde{\pi}^{N}\left(
\theta _{i}|X_{1:T}^{J},B_{1:T-1}^{J},J,\theta _{-i}\right) }\,\frac{%
q_{i}(\theta _{i}|X_{1:T}^{J},B_{1:T-1}^{J},J,\theta _{-i},\theta _{i}^{\ast
})}{q_{i}(\theta _{i}^{\ast }|X_{1:T}^{J},B_{1:T-1}^{J},J,\theta
_{-i},\theta _{i})}.
\end{eqnarray*}
\end{description}

\item[3.] Sample $U_{1:T}^{(-J),\ast }\sim \tilde{\pi}^{N}(\cdot
|X_{1:T}^{J},B_{1:T-1}^{J},J,\theta _{-i},\theta _{i}^{\ast })$.

\item[4.] Sample $J\sim \tilde{\pi}^{N}\left( \cdot |U_{1:T},\theta \right) $%
\end{description}
\end{sscheme}

The PMMH steps in Sampling Scheme~\ref{ssch:pmwgbsi} simplify similarly to
Sampling Scheme~\ref{ssch:pmwg general}. \cite{olssonryden2011} show that%
\begin{equation*}
\frac{\tilde{\pi}_{BSi}^{N}\left( U_{1:T},\theta _{i}|\theta _{-i}\right) }{%
\psi \left( U_{1:T}|\theta _{-i},\theta _{i}\right) }=\frac{Z(U_{1:T},\theta
)p(\theta _{i}|\theta _{-i})}{p\left( y_{1:T}|\theta _{-i}\right) }\text{,}
\end{equation*}%
which is the same expression as (\ref{eq:accprobsimplify1}).
Hence, the Metropolis-Hastings acceptance probability in (\ref%
{eq:PMwGBSiaccprob}) simplifies to%
\begin{equation*}
1\wedge \frac{Z(\theta _{i}^{\ast },\theta _{-i},U_{1:T}^{\ast })}{Z(\theta
_{i},\theta _{-i},U_{1:T})}\,\frac{q_{BSi,i,1}(\theta _{i}|U_{1:T}^{\ast
},J^{\ast },\theta _{-i},\theta _{i}^{\ast })p(\theta _{i}^{\ast }|\theta
_{-i})}{q_{BSi,i,1}(\theta _{i}^{\ast }|U_{1:T},J,\theta _{-i},\theta
_{i})p(\theta _{i}|\theta _{-i})}.
\end{equation*}

The results in Section \ref{s:theory} can be modified for the distribution $%
\tilde{\pi}_{BSi}^{N}\left( \cdot \right) $, instead of the distribution $%
\tilde{\pi}^{N}\left( \cdot \right) $ in a straightforward way. Details are
available from the authors on request.

\section{Target density for the factor SV model \label{sec:Sampling-Schemes-Factor model}}
This section discusses the target density of the PMMH+PG sampler for the multivariate factor SV model outlined in Section \ref{S: factor SV model explanation}. Section \ref{sss: exact transition density case} discusses an appropriate target density for the closed form density case and Section \ref{sss: Euler scheme case} discusses an appropriate target density for a factor SV model with the Euler approximation.

\subsection{The closed form density case}
\label{sss: exact transition density case}
This section provides an appropriate target density for a factor SV model with the closed form state transition density given in equation~\eqref{eq:exact transition}.
The target density includes all the random variables produced by $K + S$ univariate particle filters that generate the factor log volatilities $\boldsymbol{\lambda}_{k,1:T}$ for $k=1,...,K$
and the idiosyncratic log volatilities $\boldsymbol{h}_{s,1:T}$ for $s=1,...,S$, as well as the factors $\boldsymbol{f}_{1:T}$ and the parameters $\boldsymbol{\omega}$.
It is convenient in the developments below to define $\boldsymbol{\theta} = (\boldsymbol{f}_{1:T}, \boldsymbol{\omega})$.

To specify the univariate particle filters that generate the factor log volatilities $\boldsymbol{\lambda}_{k,1:T}$ for $k=1,...,K$,
 we use equations  \eqref{eq:SVtransition} and \eqref{eq:PMMH+PG SV obs density}
and to generate the idiosyncratic log volatilities $\boldsymbol{h}_{s,1:T},$ for $s=1,...,S,$
 we use equations \eqref{eq:exact transition} and \eqref{eq:PMMH+PG OU obs density}.
We denote the weighted samples by $\left(\boldsymbol{\lambda}_{k,t}^{1:N},\overline{w}_{f,k,t}^{1:N}\right)$
and $\left(\boldsymbol{h}_{s,t}^{1:N},\overline{w}_{\epsilon,s,t}^{1:N}\right)$.
We denote the proposal densities by $m_{f,k,1}^{\boldsymbol{\theta}}\left(\lambda_{k,1}\right)$,
$m_{f,k,t}^{\boldsymbol{\theta}}\left(\lambda_{k,t}|\lambda_{k,t-1}\right)$,
$m_{\epsilon,s,1}^{\boldsymbol{\theta}}\left(h_{s,1}\right)$ and
$m_{\epsilon,s,t}^{\boldsymbol{\theta}}\left(h_{s,t}|h_{s,t-1}\right)$ for $t=2,...,T$.
We denote the resampling schemes by $\mathcal{M}_{f}\left(\boldsymbol{a}_{f,k,t-1}^{1:N}|\overline{w}_{f,k,t-1}^{1:N}\right)$
for $k=1,...,K$, where each $a_{f,k,t-1}^{i}=j$ indexes a particle
in $\left(\boldsymbol{\lambda}_{k,t}^{1:N},\overline{w}_{f,k,t}^{1:N}\right)$ and is chosen
with probability $\overline{w}_{f,k,t}^{j}$; the resampling scheme
$\mathcal{M}_{\epsilon}\left(\boldsymbol{a}_{\epsilon,s,t-1}^{1:N}|\overline{w}_{\epsilon,s,t-1}^{1:N}\right)$
for $s=1,...,S$ is defined similarly. We denote the vector of particles by
\begin{align}
\boldsymbol{U}_{f,1:K,1:T}& =\left(\boldsymbol{\lambda}_{1:K,1:T}^{1:N}, \boldsymbol{A}_{f,1:K,1:T-1}^{1:N}\right),\label{eq:particlefactor}\\
\intertext{and}
\boldsymbol{U}_{\epsilon, 1:S,1:T}&=\left(\boldsymbol{h}_{1:S,1:T}^{1:N},\boldsymbol{A}_{\epsilon, 1:S,1:T-1}^{1:N}\right).\label{eq:particlesidio}
\end{align}
The joint distribution of the particles given the parameters is
\begin{equation}
\psi_{f,k}\left(\boldsymbol{U}_{f,k,1:T}|\boldsymbol{\theta}\right)=\\
\prod_{i=1}^{N}m_{f,k,1}^{\boldsymbol{\theta}}\left(\lambda_{k,1}^{i}\right)\prod_{t=2}^{T}\left\{ \mathcal{M}_{f}\left(\boldsymbol{a}_{f,k,t-1}^{1:N}|\overline{w}_{f,k,t-1}^{1:N}\right)\prod_{i=1}^{N}m_{f,k,t}^{\boldsymbol{\theta}}\left(\lambda_{f,k,t}^{i}|
\lambda_{f,k,t-1}^{a_{f,k,t-1}^{i}}\right)\right\},  \label{eq:factor}
\end{equation}
for $k=1,...,K,$
and
\begin{equation}
\psi_{\epsilon,s}\left(\boldsymbol{U}_{\epsilon,s,1:T}|\boldsymbol{\theta}\right)=
\prod_{i=1}^{N}m_{\epsilon,s,1}^{\boldsymbol{\theta}}\left(h_{s,1}^{i}\right)\prod_{t=2}^{T}\left\{ \mathcal{M}_{\epsilon}\left(\boldsymbol{a}_{\epsilon,s,t-1}^{1:N}|
\overline{w}_{\epsilon,s,t-1}^{1:N}\right)\prod_{i=1}^{N}m_{\epsilon,s,t}^{\boldsymbol{\theta}}\left(h_{s,t}^{i}|h_{s,t-1}^{a_{\epsilon,s,t-1}^{i}}\right)\right\}, \label{eq:idio}
\end{equation}
for $s=1,...,S$.

Next, we define indices $J_{f,k}=j$ for each $k=1,...,K$,
then trace back its ancestral
lineage $b_{f,k,1:T}^{j}$ $\left(b_{f,k,T}^{j}=j,b_{f,k,t-1}^{j}=a_{f,k,t-1}^{b_{f,k,t}^{j}}\right)$,
and select the particle trajectory $\boldsymbol{\lambda}_{k,1:T}^j=\left(\lambda_{k,1}^{b_{f,k,1}^{j}},...,\lambda_{k,T}^{b_{f,k,T}^{j}}\right)$.
Similarly, we define indices $J_{\epsilon s}=j$ for each
$s=1,...,S$, then trace back
its ancestral lineage $b_{\epsilon,s,1:T}^{j}$ $\left(b_{\epsilon,s,T}^{j}=j,b_{\epsilon,s,t-1}^{j}=a_{\epsilon,s,t-1}^{b_{\epsilon,s,t}^{j}}\right)$,
and select the particle trajectory $\boldsymbol{h}_{s,1:T}^j=\left(h_{s,1}^{b_{\epsilon,s,1}^{j}},...,h_{s,T}^{b_{\epsilon,s,T}^{j}}\right)$.

The augmented target density of the factor model is defined as
\begin{multline}
\tilde{\pi}^{N}\left(\boldsymbol{U}_{f,1:K,1:T}, \boldsymbol{U}_{\epsilon,1:S,1:T},
\boldsymbol{J}_{f},\boldsymbol{J}_{\epsilon},\boldsymbol{\theta}\right):=\\
\frac{\pi\left(\boldsymbol{\lambda}_{1:K,1:T}^{\boldsymbol{J}_{f}},\boldsymbol{h}_{1:S,1:T}^{\boldsymbol{J}_{\epsilon}},\boldsymbol{\theta}\right)}{N^{T\left(K+S\right)}}
\prod_{k=1}^{K}
\frac{\psi_{f,k}\left(\boldsymbol{U}_{f,k,1:T}|\boldsymbol{\theta}\right)}
{m_{f,k,1}^{\boldsymbol{\theta}}\left(\lambda_{k,1}^{b_{f,k,1}}\right)\prod_{t=2}^{T}\overline{w}_{f,k,t-1}^{a_{f,k,t-1}^{b_{f,k,t}}}m_{f,k,t}^{\theta}\left(\lambda_{k,t}^{b_{f,k,t}}|\lambda_{k,t-1}^{a_{f,k,t-1}^{b_{f,k,t}}}\right)}\\
\prod_{s=1}^{S}
\frac{\psi_{\epsilon,s}\left(\boldsymbol{U}_{\epsilon,s,1:T}|\boldsymbol{\theta}\right)}
{m_{\epsilon,s,1}^{\boldsymbol{\theta}}\left(h_{s,1}^{b_{\epsilon,s,1}}\right)\prod_{t=2}^{T}\overline{w}_{\epsilon,s,t-1}^{a_{\epsilon,s,t-1}^{b_{\epsilon,s,t}}}m_{\epsilon,s,t}^{\theta}\left(h_{s,t}^{b_{\epsilon,s,t}}|h_{s,t-1}^{a_{\epsilon,s,t-1}^{b_{\epsilon,s,t}}}\right)}. \label{eq:target distribution}
\end{multline}

\subsection{Approximating the transition density by an Euler scheme}
\label{sss: Euler scheme case}

This section provides an appropriate target density for a factor model with the Euler approximation given in Eq. \eqref{eq:ou approximatetransition multiple}
or Eq. \eqref{eq:GARCH Euler transitiondensity}.
We follow the approach in \cite{Lindstenetal2015} and introduce state vectors for $s=1,...,S$ defined as $x_{s,1} = h_{s,1}$ and $x_{s,t} = \( h_{s,t}, h_{s,t-1,M-1}, \dots,h_{s,t-1,1} \)^{\transp}$, for $t=2, \dots, T$. The state transition densities are given by
\begin{align}
f_{s,t}^{\theta}(x_{s,t}|x_{s,t-1}) & = \prod_{j=1}^{M}f_{s,t-1,j}^{\theta} (h_{s,t-1,j}|h_{s,t-1,j-1})\,\, (t=2, \dots, T), \label{eq: factor gen transition eqn}
\end{align}
where the densities $f_{s,t,j}^{\theta} (h_{s,t,j}|h_{s,t,j-1})$ for $j = 1, \ldots, M$, $t = 1, \ldots, T-1$ and $s=1, \ldots, S$
are defined by equation~\eqref{eq:ou approximatetransition multiple} or equation \eqref{eq:GARCH Euler transitiondensity}.
We use the proposal densities
\begin{align*}
m_{\epsilon,s,t}^{\theta} (x_{s,t} | x_{s,t-1}) = f_{s,t}^{\theta}(x_{s,t} | x_{s,t-1}) \quad (t = 2, \ldots, T \mbox{ and } s = 1, \ldots, S)
\end{align*}
which can be generated using equation \eqref{eq:ou approximatetransition multiple} or equation \eqref{eq:GARCH Euler transitiondensity}.
With these modifications, we use the same construction as Section~\ref{sss: exact transition density case}.
The modifications give
\begin{align}
\boldsymbol{U}_{\epsilon,1:S,1:T} = \left(\boldsymbol{x}_{1:S,1:T}^{1:N},\boldsymbol{A}_{\epsilon, 1:S,1:T-1}^{1:N}\right) \label{eq:particlesidio Euler}
\end{align}
\begin{align}
\psi_{\epsilon,s}\left(\boldsymbol{U}_{\epsilon,s,1:T}|\boldsymbol{\theta}\right) =
\prod_{i=1}^{N}m_{\epsilon,s,1}^{\boldsymbol{\theta}}\left(x_{s,1}^{i}\right)\prod_{t=2}^{T}\left\{ \mathcal{M}_{\epsilon}\left(\boldsymbol{a}_{\epsilon,s,t-1}^{1:N}|
\overline{w}_{\epsilon,s,t-1}^{1:N}\right)\prod_{i=1}^{N}m_{\epsilon,s,t}^{\boldsymbol{\theta}}\left(x_{s,t}^{i}|x_{s,t-1}^{a_{\epsilon,s,t-1}^{i}}\right)\right\} \label{eq:idio Euler}
\end{align}
\begin{multline}
\tilde{\pi}^{N}\left(\boldsymbol{U}_{f,1:K,1:T}, \boldsymbol{U}_{\epsilon,1:S,1:T},
\boldsymbol{J}_{f},\boldsymbol{J}_{\epsilon},\boldsymbol{\theta}\right):=\\
\frac{\pi\left(\boldsymbol{\lambda}_{1:K,1:T}^{\boldsymbol{J}_{f}},\boldsymbol{x}_{1:S,1:T}^{\boldsymbol{J}_{\epsilon}},\boldsymbol{\theta}\right)}{N^{T\left(K+S\right)}}
\prod_{k=1}^{K}
\frac{\psi_{f,k}\left(\boldsymbol{U}_{f,k,1:T}|\boldsymbol{\theta}\right)}
{m_{f,k,1}^{\boldsymbol{\theta}}\left(\lambda_{k,1}^{b_{f,k,1}}\right)\prod_{t=2}^{T}\overline{w}_{f,k,t-1}^{a_{f,k,t-1}^{b_{f,k,t}}}m_{f,k,t}^{\theta}\left(\lambda_{k,t}^{b_{f,k,t}}|\lambda_{k,t-1}^{a_{f,k,t-1}^{b_{f,k,t}}}\right)} \\
\prod_{s=1}^{S}
\frac{\psi_{\epsilon,s}\left(\boldsymbol{U}_{\epsilon,s,1:T}|\boldsymbol{\theta}\right)}
{m_{\epsilon,s,1}^{\boldsymbol{\theta}}\left(x_{s,1}^{b_{\epsilon,s,1}}\right)\prod_{t=2}^{T}\overline{w}_{\epsilon,s,t-1}^{a_{\epsilon,s,t-1}^{b_{\epsilon,s,t}}}m_{\epsilon,s,t}^{\theta}\left(x_{s,t}^{b_{\epsilon,s,t}}|x_{s,t-1}^{a_{\epsilon,s,t-1}^{b_{\epsilon,s,t}}}\right)} \label{eq:target distribution Euler}
\end{multline}


\section{PMMH+PG sampling scheme for the factor SV model}
\label{ss: sampling scheme for the factor SV model}
Similarly to Section \ref{SS: Empirical results for OU}, we use the following notation to describe the algorithms used in the examples.
The basic samplers, as used in Sampling Schemes~\ref{ssch:pmwg} or \ref{ssch:factor SV}, are $\textrm{PMMH}(\cdot)$ and $\textrm{PG}(\cdot)$.
These samplers can be used alone or in combination.
For example, $\textrm{PMMH}(\theta )$ means using a PMMH step to sample the parameter vector $\theta $; $\textrm{PMMH}(\theta _{1})+\textrm{PG}(\theta _{2})$
means sampling $\theta _{1}$ in the PMMH step and $\theta _{2}$ in the PG step; and $\textrm{PG}(\theta)$ means sampling $\theta $ using the PG sampler.

We illustrate our methods using the
$PMMH\left(\boldsymbol{\alpha},\boldsymbol{\tau}_f^{2},
\boldsymbol{\tau}_{\epsilon}^{2}\right)+PG\left(\boldsymbol{\beta}, \boldsymbol{f}_{1:T}, \boldsymbol{\phi}, \boldsymbol{\mu}\right)$
sampler, which we found to give good performance in the empirical studies in Section~\ref{s:simulations}.
It is straightforward to modify the sampling scheme for other choices of which parameters to sample with a PMMH step and which to sample with a PG step.
Our procedure to determine an efficient sampling scheme
is to run the PG algorithm first to identify which parameters have large IACT, or, in some cases, require a large amount of computational time to generate in the PG step.
We then generate these parameters in the PMMH step.
See, for example, our discussion of the univariate OU model in Section \ref{SS: Empirical results for OU}.
In particular, we note that if an Euler approximation is used, then generating any parameter in the OU or GARCH model is very
time intensive as it is necessary to determine, store and use the ancestor history of the entire state vector.

The sampling schemes for the factor SV model with the closed form transition density given by equation \eqref{eq:exact transition} and the model with the Euler scheme given
by equation \eqref{eq:ou approximatetransition multiple} or equation \eqref{eq:GARCH Euler transitiondensity} have the same structure, so Sampling Scheme \ref{ssch:factor SV}
is given below in a generic form and the appropriate state space models are used for the different cases;
see Sections~\ref{sss: exact transition density case} and \ref{sss: Euler scheme case} for details.
We have simplified the conditional distributions in Sampling Scheme~\ref{ssch:factor SV} wherever possible using the conditional independence properties discussed in Section~\ref{sss: Conditional Independence}.
The Metropolis-Hastings proposal densities for Sampling scheme \ref{ssch:factor SV} are given in Section \ref{sss:proposal densities}. We use the notation $\theta _{-i}:=(\theta _{1},\ldots ,\theta _{i-1},\theta _{i+1},\ldots,\theta_{p})$, where $p$ is the total number of parameters.

\begin{sscheme}[$PMMH\left(\boldsymbol{\alpha},\boldsymbol{\tau}_f^{2},\boldsymbol{\tau}_{\epsilon}^{2}\right)+PG\left(\boldsymbol{\beta}, \boldsymbol{f}_{1:T}, \boldsymbol{\phi}, \boldsymbol{\mu}\right)$]
\label{ssch:factor SV}
Given initial values for $U_{f,1:T}$, $U_{\epsilon, 1:T}$, $J_f$, $J_{\epsilon}$ and $\theta$, one iteration of
the MCMC involves the following steps.

\begin{enumerate}
\item (PMMH sampling),

\begin{enumerate}
\item For $k=1,...,K$

\begin{enumerate}
\item Sample $\left(\tau_{f,k}^{2*}\right)\sim q_{\tau^2_{f,k}}\left(\cdot|\boldsymbol{U}_{f,k,1:T},\tau_{f,k}^{2},\boldsymbol{\theta}_{-\tau_{f,k}^{2}}\right)$
\item Sample $\boldsymbol{U}_{f,k,1:T}^{*}\sim \psi_{f,k}\left(\cdot|\tau_{f,k}^{2*},\boldsymbol{\theta}_{-\tau_{f,k}^{2}}\right)$
\item Sample $J_{f,k}^{*}$ from $\tilde{\pi}^{N}\left(\cdot|\boldsymbol{U}_{f,k,1:T}^{*},\tau_{f,k}^{2*},\boldsymbol{\theta}_{-\tau_{f,k}^{2}}\right)$
\item Set $\left(\tau_{f,k}^{2}, \boldsymbol{U}_{f,k,1:T}, J_{f,k}\right) \leftarrow \left(\tau_{f,k}^{2*}, \boldsymbol{U}_{f,k,1:T}^{*}, J_{f,k}^{*}\right)$ with probability
\begin{multline*}
\alpha\left(\boldsymbol{U}_{f,k,1:T},J_{f,k},\tau_{f,k}^{2};
\boldsymbol{U}_{f,k,1:T}^{*},J_{f,k}^{*},\tau_{f,k}^{2*} | \boldsymbol{\theta}_{-\tau_{f,k}^{2}} \right)=\\
1\wedge\frac{Z\left(U_{f,k,1:T}^{*},\tau_{f,k}^{2*},
\boldsymbol{\theta}_{-\tau_{f,k}^{2}}\right)p\left(\tau_{f,k}^{2*}\right)}{Z\left(U_{f,k,1:T},\tau_{f,k}^{2},
\boldsymbol{\theta}_{-\tau_{f,k}^{2}}\right)p\left(\tau_{f,k}^{2}\right)}\times\frac{q_{\tau^2_{f,k}}\left(\tau_{f,k}^{2}|
\boldsymbol{U}_{f,k,1:T}^{*},\tau_{f,k}^{2*},\boldsymbol{\theta}_{-\tau_{f,k}^{2}}\right)}
{q_{\tau^2_{f,k}}\left(\tau_{f,k}^{2*}|\boldsymbol{U}_{f,k,1:T},\tau_{f,k}^{2},
\boldsymbol{\theta}_{-\tau_{f,k}^{2}}\right)}.
\end{multline*}

\end{enumerate}

\item For $s=1,...,S$,

\begin{enumerate}
\item Sample $\left(\alpha_{s}^{*},\tau_{\epsilon,s}^{2*}\right)\sim q_{\alpha_s,\tau_{\epsilon,s}^2}\left(\cdot|\boldsymbol{U}_{\epsilon,s,1:T},\alpha_{s},\tau_{\epsilon,s}^{2},
    \boldsymbol{\theta}_{-\alpha_{s},\tau_{\epsilon,s}^{2}}\right)$
\item Sample $\boldsymbol{U}_{\epsilon,s,1:T}^{*}\sim \psi_{\epsilon,s} \left(\cdot|\alpha_{s}^{*},\tau_{\epsilon,s}^{2*},
    \boldsymbol{\theta}_{-\alpha_{s},\tau_{\epsilon,s}^{2}}\right)$
\item Sample $J_{\epsilon,s}^{*}$ from $\tilde{\pi}^{N}\left(\cdot|\boldsymbol{U}_{\epsilon,s,1:T}^{*},\alpha_{s}^{*},\tau_{\epsilon,s}^{2*},
    \boldsymbol{\theta}_{-\alpha_{s},\tau_{\epsilon,s}^{2}}\right)$
\item Set $\left(\alpha_{s},\tau_{\epsilon,s}^{2}, \boldsymbol{U}_{\epsilon,s,1:T}, J_{\epsilon,s}\right) \leftarrow \left(\alpha_{s}^{*},\tau_{\epsilon,s}^{2*}, \boldsymbol{U}_{\epsilon,s,1:T}^{*}, J_{\epsilon,s}^{*}\right)$ with
probability
\begin{multline*}
\alpha\left(\boldsymbol{U}_{\epsilon,s,1:T},J_{s},\left(\alpha_{s},\tau_{\epsilon,s}^{2}\right);
\boldsymbol{U}_{\epsilon,s,1:T}^{*},J_{\epsilon,s}^{*},\left(\alpha_{s}^{*},\tau_{\epsilon,s}^{2*}\right) | \boldsymbol{\theta}_{-\alpha_{s},\tau_{\epsilon,s}^{2}} \right)=\\
1\wedge\frac{Z\left(U_{\epsilon,s,1:T}^{*},\alpha_{s}^{*},\tau_{\epsilon,s}^{2*},
\boldsymbol{\theta}_{-\alpha_{s},\tau_{\epsilon,s}^{2}}\right)
p\left(\alpha_{s}^{*},\tau_{\epsilon,s}^{2*}\right)}{Z\left(U_{\epsilon,s,1:T},\alpha_{s},\tau_{\epsilon,s}^{2},
\boldsymbol{\theta}_{-\alpha_{s},\tau_{s}^{2}}\right)p\left(\alpha_{s},\tau_{s}^{2}\right)}
\times\frac{q_{\alpha_s,\tau_{\epsilon,s}^2}\left(\alpha_{s},\tau_{\epsilon,s}^{2}|
\boldsymbol{U}_{\epsilon,s,1:T}^{*}, \alpha_{s}^{*},\tau_{\epsilon,s}^{2*},
\boldsymbol{\theta}_{-\alpha_{s},\tau_{\epsilon,s}^{2}}\right)}{q_{\alpha_s,\tau_{\epsilon,s}^2}\left(\alpha_{s}^{*},\tau_{\epsilon,s}^{2*}|\boldsymbol{U}_{\epsilon,s,1:T},\alpha_{s},
\tau_{\epsilon s}^{2},\boldsymbol{\theta}_{-\alpha_{s},\tau_{\epsilon,s}^{2}}\right)}.
\end{multline*}

\end{enumerate}
\end{enumerate}

\item (PG sampling)

\begin{enumerate}
\item Sample $\boldsymbol{\beta}|\boldsymbol{\lambda}_{1:T}^{\boldsymbol{J}_{f}},\boldsymbol{h}_{1:T}^{\boldsymbol{J}_{\epsilon}},\boldsymbol{B}_{f,1:T-1}^{\boldsymbol{J}_f},\boldsymbol{B}_{\epsilon,1:T-1}^{\boldsymbol{J}_{\epsilon}},\boldsymbol{J}_f,\boldsymbol{J}_{\epsilon},\boldsymbol{\theta}_{-\boldsymbol{\beta}},\boldsymbol{y}_{1:T}$
using equation \eqref{eq:Bfactor} in Appendix~\ref{SS: sampling loading matrix}.
\item Redraw the diagonal elements of $\boldsymbol{\beta}$ through the deep
interweaving procedure described in Appendix~\ref{SS: deep interweaving}.
This step is necessary to improve the mixing of the factor
loading matrix $\boldsymbol{\beta}$.
\item Sample $\boldsymbol{f}_{1:T}|\boldsymbol{\lambda}_{1:T}^{\boldsymbol{J}_{f}},\boldsymbol{h}_{1:T}^{\boldsymbol{J}_{\epsilon}},\boldsymbol{B}_{f,1:T-1}^{\boldsymbol{J}_f},\boldsymbol{B}_{\epsilon,1:T-1}^{\boldsymbol{J}_{\epsilon}},\boldsymbol{J}_f,\boldsymbol{J}_{\epsilon},\boldsymbol{\theta}_{-{\boldsymbol{f}_{1:T}}},\boldsymbol{y}_{1:T}$
using equation~\eqref{eq:factordraws} in Appendix~\ref{SS: sampling latent factors}.

\item For $k=1,...,K$

\begin{enumerate}
\item Sample $\phi_{k}^{*}$ from the proposal $q_{\phi_{k}}\left(\cdot|\boldsymbol{\lambda}_{k,1:T}^{J_{f,k}},
\boldsymbol{\theta}_{-\phi_{k}}\right)$
and set $\phi_{k} \leftarrow \phi_{k}^{*}$ with probability
\begin{align*}
1\land\frac{\tilde{\pi}^{N}\left(\phi_{k}^{*}|\boldsymbol{\lambda}_{k,1:T}^{J_{f,k}},\boldsymbol{B}_{f,k,1:T-1},J_{f,k},\boldsymbol{\theta}_{-\phi_{k}}\right)}{\tilde{\pi}^{N}\left(\phi_{k}|\boldsymbol{\lambda}_{k,1:T}^{J_{f,k}},\boldsymbol{B}_{f,k,1:T-1},J_{f,k},\boldsymbol{\theta}_{-\phi_{k}}\right)}\times
\frac{q_{\phi_{k}}\left(\phi_{k}|\boldsymbol{\lambda}_{k,1:T}^{J_{f,k}},
\boldsymbol{\theta}_{-\phi_{k}}\right)}{q_{\phi_{k}}\left(\phi_{k}^{*}|\boldsymbol{\lambda}_{k,1:T}^{J_{f,k}},
\boldsymbol{\theta}_{-\phi_{k}}\right)}.
\end{align*}

\item Sample $\boldsymbol{U}_{f,k,1:T}^{\left(-J_{f,k}\right)}\sim\tilde{\pi}^{N}\left(\cdot|\boldsymbol{\lambda}_{k,1:T}^{J_{f,k}},\boldsymbol{B}_{f,k,1:T-1},J_{f,k},\boldsymbol{\theta}\right)$
using the conditional sequential Monte Carlo algorithm (CSMC) discussed in Section ~\ref{alg:condsmc}.
\item Sample $J_{f,k}\sim\tilde{\pi}^{N}\left(\cdot|\boldsymbol{U}_{f,k,1:T},\boldsymbol{\theta}\right)$.
\end{enumerate}

\item For $s=1,...,S$,

\begin{enumerate}
\item Sample $\mu_{s}^{*}$ from the proposal $q_{\mu_{s}}\left(\cdot|\boldsymbol{h}_{s,1:T}^{J_{\epsilon,s}},
\boldsymbol{\theta}_{-\mu_{s}}\right)$
and set $\mu_{s} \leftarrow \mu_{s}^{*}$ with probability
\begin{align*}
1\land\frac{\tilde{\pi}^{N}\left(\mu_{s}^{*}|\boldsymbol{h}_{s,1:T}^{J_{\epsilon,s}},\boldsymbol{B}_{\epsilon,s,1:T-1},J_{\epsilon,s},
\boldsymbol{\theta}_{-\mu_{s}}\right)}
{\tilde{\pi}^{N}\left(\mu_{s}|\boldsymbol{h}_{s,1:T}^{J_{s}},\boldsymbol{B}_{\epsilon,s,1:T-1},J_{\epsilon,s},
\boldsymbol{\theta}_{-\mu_{s}}\right)}\times
\frac{q_{\mu_{s}}\left(\mu_{s}|\boldsymbol{h}_{s,1:T}^{J_{\epsilon,s}},
\boldsymbol{\theta}_{-\mu_{s}}\right)}{q_{\mu_{s}}\left(\mu_{s}^{*}|\boldsymbol{h}_{s,1:T}^{J_{\epsilon,s}},
\boldsymbol{\theta}_{-\mu_{s}}\right)}
\end{align*}

\item Sample $\boldsymbol{U}_{\epsilon,s,1:T}^{\left(-J_{\epsilon,s}\right)}
    \sim\tilde{\pi}^{N}\left(\cdot|
    \boldsymbol{h}_{s,1:T}^{J_{s}},\boldsymbol{B}_{\epsilon,s,1:T-1},J_{\epsilon,s},
\boldsymbol{\theta}\right)$
using the conditional sequential Monte Carlo algorithm (CSMC) discussed in Section~\ref{s:CSMC}.
\item Sample $J_{\epsilon,s}\sim\tilde{\pi}^{N}\left(\cdot|\boldsymbol{U}_{\epsilon,s,1:T},\boldsymbol{\theta}\right)$.
\end{enumerate}

\end{enumerate}

\end{enumerate}
\end{sscheme}

\subsection{Proposal densities\label{sss:proposal densities}}
This section details the proposal densities used in Sampling Scheme \ref{ssch:factor SV}
for the exact OU model given by equation \eqref{eq:exact transition}.
We will specify other cases such as the Euler evolution given by equation \eqref{eq:ou approximatetransition multiple}
and the GARCH diffusion model given by equation \eqref{eq:GARCH Euler transitiondensity} when describing the sampling scheme.
\begin{itemize}

\item
For $k=1, \dots, K$, $q_{\tau_{f,k}^{2}}$ is an adaptive random walk.
\item
For $s=1,\dots, S$, $q_{\alpha_{s},\tau_{\epsilon,s}^{2}}$
is an adaptive random walk.
\item
For $k=1, \dots, K$, $q_{\phi_{k}}\left(\cdot|\boldsymbol{\lambda}_{k,1:T}^{J_{f,k}},
\boldsymbol{\theta}_{-\phi_{k}}\right)= N\left(c_{\phi_{k}},d_{\phi_{k}}\right)$,
where
\begin{align*}
c_{\phi_{k}} &=\frac{d_{\phi_{k}}}{\tau_{f,k}^{2}} \sum_{t=2}^{T}\lambda_{k,t}\lambda_{k,t-1},\,\,\,
\text{and}\,\,\,
d_{\phi_{k}} =\frac{\tau_{f,k}^{2}}{\sum_{t=2}^{T-1}\lambda_{k,t}^{2}},
\end{align*}
\item
For $s=1,\dots, S$, $q_{\mu_{s}}\left(\cdot|\boldsymbol{h}_{s,1:T}^{J_{\epsilon s}},
\boldsymbol{\theta}_{-\mu_{s}}\right)= N\left(c_{\mu_{s}},d_{\mu_{s}}\right)$, where
\begin{align*}
c_{\mu_{s}}&=\frac{d_{\mu_{s}}}{\tau_{\epsilon,s}^2}\bigg (
h_{s,1}\left(2\alpha_{s}\right)+
\left ( \frac{2\alpha_{s}}{1-\exp\left(-2\alpha_{s}\right)}\right)
\left (  \sum_{t=2}^{T}
\left( h_{s,t}-\exp\left(-\alpha_{s}\right)h_{s,t}+ \right .\right .\\
& \left . \exp\left(-2\alpha_{s}\right)h_{s,t-1}
-\exp\left(-\alpha_{s}\right)h_{s,t-1}\right )
\bigg).
\\
d_{\mu_{s}}& =\frac{\tau_{\epsilon,s}^{2}}{\left(2\alpha_{s}\right)+\left(\frac{2\alpha_{s}}
{1-\exp\left(-2\alpha_{s}\right)}\right)\left(T-1\right)\left(1-2\exp\left(-\alpha_{s}\right)+
\exp\left(-2\alpha_{s}\right)\right)^{2}},
\end{align*}

\end{itemize}


\subsection{Sampling the factor loading matrix $\boldsymbol{\beta}$\label{SS: sampling loading matrix}}
First, to identify the parameters for the factor loading matrix
$\boldsymbol{\beta}$, we follow the usual convention and set the
upper triangular part of $\boldsymbol{\beta}$ to zero (\cite{gewekezhou1996}). This parameterisation imposes
an order dependence. Second, the model is also not identified without further
constraining either the scale of the $k$th column of $\boldsymbol{\beta}$
or the variance of $f_{k,t}$. The usual solution is to set the
diagonal elements of the factor loading matrix $\boldsymbol{\beta}_{k,k}$
to one, for $k=1,..,K$, while the level $\mu_{f,k}$ of the factor
volatility $\lambda_{k,t}$ is modeled to be unknown. However, \cite{Kastner:2017} note that
this approach makes the variable ordering dependence stronger. We therefore
follow \cite{Kastner:2017} and leave the diagonal elements $\boldsymbol{\beta}_{k,k}$
unrestricted and set the level $\mu_{f,k}$ of the factor volatility
$\lambda_{k,t}$ to zero for $k=1,...,K$.

Let $k_{s}$ denote the number of unrestricted
elements in row $s$ of $\bs \beta$ and define
\[
\boldsymbol{F}_{s}=\left[\begin{array}{ccc}
f_{1,1} & \cdots & f_{k_{s},1}\\
\vdots &  & \vdots\\
f_{1,T} & \cdots & f_{k_{s},T}
\end{array}\right],
\quad \text{and} \quad
\boldsymbol{\wt V}_{s}=\left[\begin{array}{ccc}
\exp\left(h_{s,1}\right) & \cdots & 0\\
0 & \ddots & 0\\
0 & \cdots & \exp\left(h_{s,T}\right)
\end{array}\right].
\]
We sample the factor loadings $\boldsymbol{\beta}_{s,.}=\left(\beta_{s,1},...,\beta_{s,k_{s}}\right)^{\transp}$,
for $s=1,...,S$,  independently for each $s$ using the Gibbs-update
\begin{equation}
\boldsymbol{\beta}_{s,.}|\boldsymbol{f},\boldsymbol{y}_{s,.},\boldsymbol{h}_{s,.}\sim N_{k_{s}}\left(a_{s,T},b_{s,T}\right),\label{eq:Bfactor}
\end{equation}
where
$
b_{s,T}=\left(\boldsymbol{F}_{s}^{\transp}\boldsymbol{\wt V}_{s}^{-1}\boldsymbol{F}_{s}+I_{k_{s}}\right)^{-1}$
and $
a_{s,T}=b_{s,T}\boldsymbol{F}_{s}^{\transp}\boldsymbol{\wt V}_{p}^{-1}\boldsymbol{y}_{s,1:T}$.

\subsection{Deep Interweaving\label{SS: deep interweaving}}

To improve the mixing in the draws
of the factor loading matrix we employ the following deep interweaving strategy
introduced by \cite{Kastner:2017}.
\begin{itemize}
\item Determine the vector $\boldsymbol{\beta}_{.,k}^{*}$, where $\beta_{s,k}^{*}=\beta_{s,k}^{old}/\beta_{k,k}^{old}$
in the $k$th column of the transformed factor loading matrix $\boldsymbol{\beta}^{*}$.
\item Define $\boldsymbol{\lambda_{k,.}}^{*}={\bs \lambda}_{k,.}^{old}+2\log|\beta_{k,k}^{old}|$
and sample $\beta_{k,k}^{new}$ from $p\left(\beta_{k,k}|\beta_{.,k}^{*},\boldsymbol{\lambda}_{k,.}^{*},\phi_{k},\tau_{f,k}^{2}\right)$.
\item Update $\boldsymbol{\beta}_{.,k}=\frac{\beta_{k,k}^{new}}{\beta_{k,k}^{old}}\boldsymbol{\beta}_{.,k}^{old}$,
$\boldsymbol{f}_{k,.}=\frac{\beta_{k,k}^{old}}{\beta_{k,k}^{new}}\boldsymbol{f}_{k,.}^{old}$,
and ${\bs \lambda}_{k,.}={\bs \lambda}_{k,.}^{old}+2\log|\frac{\beta_{k,k}^{old}}{\beta_{k,k}^{new}}|$.
\end{itemize}
In the deep interweaving representation  the scaling parameter
$\beta_{k,k}$ is sampled indirectly through $\mu_{f,k}=\log\beta_{k,k}^{2}$, $k=1,...,K$.
The implied prior $p\left(\mu_{f,k}\right)\propto\exp\left(\mu_{f,k}/2-\exp\left(\mu_{f,k}\right)/2\right)$
and the density $p\left(\boldsymbol{\beta}_{.,k}^{*}|\mu_{f,k}\right)\sim N_{k_{l}}\left(0,\exp\left(-\mu_{f,k}\right)I_{k_{l}}\right)$
and the likelihood yields the posterior
\[
p\left(\mu_{f,k}|\boldsymbol{\beta}_{.,k}^{*},\boldsymbol{\lambda}_{k,.}^{*},\phi_{k},\tau_{f,k}^{2}\right)\propto p\left(\boldsymbol{\lambda}_{k,.}^{*}|\mu_{f,k},\phi_{k},\tau_{f,k}^{2}\right)p\left(\boldsymbol{\beta}_{.,k}^{*}|\mu_{f,k}\right)p\left(\mu_{f,k}\right),
\]
which is not in recognisable form. We draw a proposal for $\mu_{f,k}^{prop}$
from $N\left(A,B\right)$ where
\[
A=\frac{\sum_{t=2}^{T-1}\lambda_{k,t}^{*}+\left(\lambda_{k,T}^{*}-\phi_{k}\lambda_{k,1}^{*}\right)/\left(1-\phi_{k}\right)}{T-1+1/B_{0}},B=\frac{\tau_{f,k}^{2}/\left(1-\phi_{k}\right)^{2}}{T-1+1/B_{0}}.
\]
Denoting the current value $\mu_{f,k}$ by $\mu_{f,k}^{old}$, the new
value $\mu_{f,k}^{prop}$ gets accepted with probability $\min\left(1,R\right)$,
where
\[
R=\frac{p\left(\mu_{f,k}^{prop}\right)p\left(\lambda_{k,1}^{*}|\mu_{f,k}^{prop},\phi_{k},\tau_{f,k}^{2}\right)p\left(\boldsymbol{\beta}_{.,k}^{*}|\mu_{f,k}^{prop}\right)}{p\left(\mu_{f,k}^{old}\right)\left(\lambda_{k,1}^{*}|\mu_{f,k}^{old},\phi_{k},\tau_{f,k}^{2}\right)p\left(\boldsymbol{\beta}{}_{.,k}^{*}|\mu_{f,k}^{old}\right)}\times\frac{p_{aux}\left(\mu_{f,k}^{old}|\phi_{k},\tau_{f,k}^{2}\right)}{p_{aux}\left(\mu_{f,k}^{prop}|\phi_{k},\tau_{f,k}^{2}\right)},
\]
where
\[
p_{aux}\left(\mu_{f,k}^{old}|\phi_{k},\tau_{f,k}^{2}\right)\sim N\left(0,B_{0}\tau_{f,k}^{2}/\left(1-\phi_{k}\right)^{2}\right).
\]
The constant $B_{0}$ is set to large value $10^{5}$ as in \cite{Kastner:2017}.

\subsection{Sampling the Latent Factors $\boldsymbol{f}_{1:T}$\label{SS: sampling latent factors}}
After some algebra, we obtain  that

\begin{align}
\left\{ \boldsymbol{f}_{t}\right\} |\boldsymbol{y},\left\{ \boldsymbol{h}_{t}\right\} ,\left\{ \boldsymbol{\lambda}_{t}\right\},\boldsymbol{\beta}
 & \sim N\left(a_{t},b_{t}\right),\label{eq:factordraws}
\end{align}
where
$b_{t}  =\left(\boldsymbol{\beta}^{\transp}\boldsymbol{V}_{t}^{-1} \boldsymbol{\beta}+\boldsymbol{D}_{t}^{-1}\right)^{-1}$
and
$a_{t}  =  b_{t}\boldsymbol{\beta}^{\transp}\boldsymbol{V}_{t}^{-1}\boldsymbol{y}_{t}$.

\section{Tables and figures for the factor stochastic volatility model in Sections \ref{ss:simulationandapplication} and \ref{SS: US stock returns}}
\label{S:FSV tables and figures} 

\begin{table}[H]
\caption{Inefficiency factor of $\boldsymbol{\beta}$, $\boldsymbol{\alpha}$,
$\boldsymbol{\mu}$, $\boldsymbol{\tau}^{2}$, $\boldsymbol{\phi}$,
and $\boldsymbol{\tau}_{f}^{2}$ with exact transition density for the Gaussian OU model: Sampler
I: $PMMH\left(\boldsymbol{\alpha},\boldsymbol{\tau}^{2},\boldsymbol{\tau}_{f}^{2}\right)+PG\left(\boldsymbol{\beta},\boldsymbol{\mu},\boldsymbol{\phi}\right)$,
Sampler $II$: $PGAT\left(\boldsymbol{\beta},\boldsymbol{\alpha},\boldsymbol{\tau}^{2},\boldsymbol{\mu},\boldsymbol{\phi},\boldsymbol{\tau}_{f}^{2}\right)$,
sampler III: $PGBS\left(\boldsymbol{\beta},\boldsymbol{\alpha},\boldsymbol{\tau}^{2},\boldsymbol{\mu},\boldsymbol{\phi},\boldsymbol{\tau}_{f}^{2}\right)$
for simulated data with $T=1000$, $S=20$,
and $K=1$, and number of particles $N=500$.\label{tab:Inefficiency-factor-of simulation}}

\centering{}%
\begin{tabular}{cccccccccccccccc}
\hline
 & {\footnotesize{}I} & {\footnotesize{}II} & {\footnotesize{}III} &  & {\footnotesize{}I} & {\footnotesize{}II} & {\footnotesize{}III} &  & {\footnotesize{}I} & {\footnotesize{}II} & {\footnotesize{}III} &  & {\footnotesize{}I} & {\footnotesize{}II} & {\footnotesize{}III}\tabularnewline
\hline
{\footnotesize{}$\beta_{1}$} & {\footnotesize{}$12.55$} & {\footnotesize{}$12.92$} & {\footnotesize{}$13.95$} & {\footnotesize{}$\alpha_{1}$} & {\footnotesize{}$12.64$} & {\footnotesize{}$66.69$} & {\footnotesize{}$39.94$} & {\footnotesize{}$\tau_{\epsilon,1}^{2}$} & {\footnotesize{}$14.70$} & {\footnotesize{}$136.58$} & {\footnotesize{}$99.80$} & {\footnotesize{}$\mu_{1}$} & {\footnotesize{}$1.29$} & {\footnotesize{}$1.47$} & {\footnotesize{}$1.39$}\tabularnewline
{\footnotesize{}$\beta_{2}$} & {\footnotesize{}$12.67$} & {\footnotesize{}$13.03$} & {\footnotesize{}$13.94$} & {\footnotesize{}$\alpha_{2}$} & {\footnotesize{}$11.76$} & {\footnotesize{}$44.67$} & {\footnotesize{}$35.59$} & {\footnotesize{}$\tau_{\epsilon,2}^{2}$} & {\footnotesize{}$14.36$} & {\footnotesize{}$72.64$} & {\footnotesize{}$74.03$} & {\footnotesize{}$\mu_{2}$} & {\footnotesize{}$1.28$} & {\footnotesize{}$1.43$} & {\footnotesize{}$1.33$}\tabularnewline
{\footnotesize{}$\beta_{3}$} & {\footnotesize{}$12.69$} & {\footnotesize{}$13.20$} & {\footnotesize{}$14.17$} & {\footnotesize{}$\alpha_{3}$} & {\footnotesize{}$11.89$} & {\footnotesize{}$64.76$} & {\footnotesize{}$61.08$} & {\footnotesize{}$\tau_{\epsilon,3}^{2}$} & {\footnotesize{}$12.01$} & {\footnotesize{}$92.80$} & {\footnotesize{}$101.64$} & {\footnotesize{}$\mu_{3}$} & {\footnotesize{}$1.56$} & {\footnotesize{}$1.72$} & {\footnotesize{}$1.59$}\tabularnewline
{\footnotesize{}$\beta_{4}$} & {\footnotesize{}$12.53$} & {\footnotesize{}$12.37$} & {\footnotesize{}$13.77$} & {\footnotesize{}$\alpha_{4}$} & {\footnotesize{}$13.13$} & {\footnotesize{}$107.58$} & {\footnotesize{}$59.69$} & {\footnotesize{}$\tau_{\epsilon,4}^{2}$} & {\footnotesize{}$14.70$} & {\footnotesize{}$283.23$} & {\footnotesize{}$93.35$} & {\footnotesize{}$\mu_{4}$} & {\footnotesize{}$1.41$} & {\footnotesize{}$1.40$} & {\footnotesize{}$1.33$}\tabularnewline
{\footnotesize{}$\beta_{5}$} & {\footnotesize{}$12.66$} & {\footnotesize{}$13.08$} & {\footnotesize{}$13.86$} & {\footnotesize{}$\alpha_{5}$} & {\footnotesize{}$15.21$} & {\footnotesize{}$76.45$} & {\footnotesize{}$35.94$} & {\footnotesize{}$\tau_{\epsilon,5}^{2}$} & {\footnotesize{}$14.56$} & {\footnotesize{}$123.53$} & {\footnotesize{}$81.58$} & {\footnotesize{}$\mu_{5}$} & {\footnotesize{}$1.29$} & {\footnotesize{}$1.37$} & {\footnotesize{}$1.25$}\tabularnewline
{\footnotesize{}$\beta_{6}$} & {\footnotesize{}$12.76$} & {\footnotesize{}$12.89$} & {\footnotesize{}$14.01$} & {\footnotesize{}$\alpha_{6}$} & {\footnotesize{}$14.80$} & {\footnotesize{}$37.25$} & {\footnotesize{}$30.74$} & {\footnotesize{}$\tau_{\epsilon,6}^{2}$} & {\footnotesize{}$14.84$} & {\footnotesize{}$76.76$} & {\footnotesize{}$56.96$} & {\footnotesize{}$\mu_{6}$} & {\footnotesize{}$1.25$} & {\footnotesize{}$1.29$} & {\footnotesize{}$1.18$}\tabularnewline
{\footnotesize{}$\beta_{7}$} & {\footnotesize{}$12.56$} & {\footnotesize{}$12.62$} & {\footnotesize{}$13.72$} & {\footnotesize{}$\alpha_{7}$} & {\footnotesize{}$14.11$} & {\footnotesize{}$27.87$} & {\footnotesize{}$24.29$} & {\footnotesize{}$\tau_{\epsilon,7}^{2}$} & {\footnotesize{}$13.36$} & {\footnotesize{}$58.61$} & {\footnotesize{}$43.39$} & {\footnotesize{}$\mu_{7}$} & {\footnotesize{}$1.23$} & {\footnotesize{}$1.28$} & {\footnotesize{}$1.18$}\tabularnewline
{\footnotesize{}$\beta_{8}$} & {\footnotesize{}$12.85$} & {\footnotesize{}$12.96$} & {\footnotesize{}$13.87$} & {\footnotesize{}$\alpha_{8}$} & {\footnotesize{}$13.65$} & {\footnotesize{}$40.08$} & {\footnotesize{}$19.94$} & {\footnotesize{}$\tau_{\epsilon,8}^{2}$} & {\footnotesize{}$13.37$} & {\footnotesize{}$98.49$} & {\footnotesize{}$42.14$} & {\footnotesize{}$\mu_{8}$} & {\footnotesize{}$1.24$} & {\footnotesize{}$1.27$} & {\footnotesize{}$1.20$}\tabularnewline
{\footnotesize{}$\beta_{9}$} & {\footnotesize{}$12.52$} & {\footnotesize{}$13.11$} & {\footnotesize{}$13.83$} & {\footnotesize{}$\alpha_{9}$} & {\footnotesize{}$13.58$} & {\footnotesize{}$96.90$} & {\footnotesize{}$47.77$} & {\footnotesize{}$\tau_{\epsilon,9}^{2}$} & {\footnotesize{}$15.06$} & {\footnotesize{}$144.72$} & {\footnotesize{}$81.66$} & {\footnotesize{}$\mu_{9}$} & {\footnotesize{}$1.99$} & {\footnotesize{}$1.86$} & {\footnotesize{}$1.54$}\tabularnewline
{\footnotesize{}$\beta_{10}$} & {\footnotesize{}$12.39$} & {\footnotesize{}$12.81$} & {\footnotesize{}$14.05$} & {\footnotesize{}$\alpha_{10}$} & {\footnotesize{}$18.07$} & {\footnotesize{}$23.49$} & {\footnotesize{}$32.13$} & {\footnotesize{}$\tau_{\epsilon,10}^{2}$} & {\footnotesize{}$16.56$} & {\footnotesize{}$58.06$} & {\footnotesize{}$57.03$} & {\footnotesize{}$\mu_{10}$} & {\footnotesize{}$1.29$} & {\footnotesize{}$1.28$} & {\footnotesize{}$1.23$}\tabularnewline
{\footnotesize{}$\beta_{11}$} & {\footnotesize{}$12.80$} & {\footnotesize{}$12.94$} & {\footnotesize{}$14.13$} & {\footnotesize{}$\alpha_{11}$} & {\footnotesize{}$17.31$} & {\footnotesize{}$41.43$} & {\footnotesize{}$31.13$} & {\footnotesize{}$\tau_{\epsilon,11}^{2}$} & {\footnotesize{}$14.33$} & {\footnotesize{}$75.79$} & {\footnotesize{}$66.30$} & {\footnotesize{}$\mu_{11}$} & {\footnotesize{}$1.33$} & {\footnotesize{}$1.37$} & {\footnotesize{}$1.27$}\tabularnewline
{\footnotesize{}$\beta_{12}$} & {\footnotesize{}$12.75$} & {\footnotesize{}$13.07$} & {\footnotesize{}$14.22$} & {\footnotesize{}$\alpha_{12}$} & {\footnotesize{}$16.33$} & {\footnotesize{}$30.14$} & {\footnotesize{}$47.93$} & {\footnotesize{}$\tau_{\epsilon,12}^{2}$} & {\footnotesize{}$14.18$} & {\footnotesize{}$53.80$} & {\footnotesize{}$74.84$} & {\footnotesize{}$\mu_{12}$} & {\footnotesize{}$1.42$} & {\footnotesize{}$1.35$} & {\footnotesize{}$1.31$}\tabularnewline
{\footnotesize{}$\beta_{13}$} & {\footnotesize{}$12.78$} & {\footnotesize{}$12.87$} & {\footnotesize{}$14.16$} & {\footnotesize{}$\alpha_{13}$} & {\footnotesize{}$16.24$} & {\footnotesize{}$38.37$} & {\footnotesize{}$27.31$} & {\footnotesize{}$\tau_{\epsilon,13}^{2}$} & {\footnotesize{}$13.67$} & {\footnotesize{}$67.67$} & {\footnotesize{}$47.37$} & {\footnotesize{}$\mu_{13}$} & {\footnotesize{}$1.25$} & {\footnotesize{}$1.31$} & {\footnotesize{}$1.25$}\tabularnewline
{\footnotesize{}$\beta_{14}$} & {\footnotesize{}$12.78$} & {\footnotesize{}$13.04$} & {\footnotesize{}$14.23$} & {\footnotesize{}$\alpha_{14}$} & {\footnotesize{}$14.41$} & {\footnotesize{}$38.38$} & {\footnotesize{}$21.61$} & {\footnotesize{}$\tau_{\epsilon,14}^{2}$} & {\footnotesize{}$15.88$} & {\footnotesize{}$83.16$} & {\footnotesize{}$46.09$} & {\footnotesize{}$\mu_{14}$} & {\footnotesize{}$1.27$} & {\footnotesize{}$1.30$} & {\footnotesize{}$1.26$}\tabularnewline
{\footnotesize{}$\beta_{15}$} & {\footnotesize{}$12.47$} & {\footnotesize{}$12.82$} & {\footnotesize{}$13.80$} & {\footnotesize{}$\alpha_{15}$} & {\footnotesize{}$12.72$} & {\footnotesize{}$34.25$} & {\footnotesize{}$22.16$} & {\footnotesize{}$\tau_{\epsilon,15}^{2}$} & {\footnotesize{}$15.39$} & {\footnotesize{}$60.91$} & {\footnotesize{}$44.90$} & {\footnotesize{}$\mu_{15}$} & {\footnotesize{}$1.22$} & {\footnotesize{}$1.25$} & {\footnotesize{}$1.19$}\tabularnewline
{\footnotesize{}$\beta_{16}$} & {\footnotesize{}$12.91$} & {\footnotesize{}$12.99$} & {\footnotesize{}$14.01$} & {\footnotesize{}$\alpha_{16}$} & {\footnotesize{}$15.19$} & {\footnotesize{}$70.11$} & {\footnotesize{}$42.38$} & {\footnotesize{}$\tau_{\epsilon,16}^{2}$} & {\footnotesize{}$13.60$} & {\footnotesize{}$110.75$} & {\footnotesize{}$66.36$} & {\footnotesize{}$\mu_{16}$} & {\footnotesize{}$1.40$} & {\footnotesize{}$1.62$} & {\footnotesize{}$1.34$}\tabularnewline
{\footnotesize{}$\beta_{17}$} & {\footnotesize{}$12.74$} & {\footnotesize{}$13.11$} & {\footnotesize{}$13.86$} & {\footnotesize{}$\alpha_{17}$} & {\footnotesize{}$11.17$} & {\footnotesize{}$22.16$} & {\footnotesize{}$27.11$} & {\footnotesize{}$\tau_{\epsilon,17}^{2}$} & {\footnotesize{}$11.43$} & {\footnotesize{}$53.60$} & {\footnotesize{}$51.73$} & {\footnotesize{}$\mu_{17}$} & {\footnotesize{}$1.37$} & {\footnotesize{}$1.31$} & {\footnotesize{}$1.21$}\tabularnewline
{\footnotesize{}$\beta_{18}$} & {\footnotesize{}$12.58$} & {\footnotesize{}$12.93$} & {\footnotesize{}$13.84$} & {\footnotesize{}$\alpha_{18}$} & {\footnotesize{}$12.74$} & {\footnotesize{}$28.17$} & {\footnotesize{}$28.51$} & {\footnotesize{}$\tau_{\epsilon,18}^{2}$} & {\footnotesize{}$15.66$} & {\footnotesize{}$59.10$} & {\footnotesize{}$75.58$} & {\footnotesize{}$\mu_{18}$} & {\footnotesize{}$1.33$} & {\footnotesize{}$1.32$} & {\footnotesize{}$1.30$}\tabularnewline
{\footnotesize{}$\beta_{19}$} & {\footnotesize{}$12.64$} & {\footnotesize{}$12.81$} & {\footnotesize{}$13.80$} & {\footnotesize{}$\alpha_{19}$} & {\footnotesize{}$12.67$} & {\footnotesize{}$40.38$} & {\footnotesize{}$29.96$} & {\footnotesize{}$\tau_{\epsilon,19}^{2}$} & {\footnotesize{}$15.17$} & {\footnotesize{}$74.87$} & {\footnotesize{}$59.19$} & {\footnotesize{}$\mu_{19}$} & {\footnotesize{}$1.44$} & {\footnotesize{}$1.57$} & {\footnotesize{}$1.41$}\tabularnewline
{\footnotesize{}$\beta_{20}$} & {\footnotesize{}$12.77$} & {\footnotesize{}$13.19$} & {\footnotesize{}$14.08$} & {\footnotesize{}$\alpha_{20}$} & {\footnotesize{}$12.85$} & {\footnotesize{}$27.12$} & {\footnotesize{}$22.34$} & {\footnotesize{}$\tau_{\epsilon,20}^{2}$} & {\footnotesize{}$12.84$} & {\footnotesize{}$73.02$} & {\footnotesize{}$44.80$} & {\footnotesize{}$\mu_{20}$} & {\footnotesize{}$1.26$} & {\footnotesize{}$1.38$} & {\footnotesize{}$1.30$}\tabularnewline
$\phi$ & {\footnotesize{}$8.03$} & {\footnotesize{}$20.12$} & {\footnotesize{}$18.62$} & {\footnotesize{}$\tau_{f,1}^{2}$} & {\footnotesize{}$14.76$} & {\footnotesize{}$73.76$} & {\footnotesize{}$79.14$} &  &  &  &  &  &  &  & \tabularnewline
\hline
\end{tabular}
\end{table}

\begin{sidewaystable}
\caption{Inefficiency factor of $\boldsymbol{\beta}$, $\boldsymbol{\alpha}$,
$\boldsymbol{\mu}$, $\boldsymbol{\tau}^{2}$, $\boldsymbol{\phi}$,
and $\boldsymbol{\tau}_{f}^{2}$ with Euler approximation for state
transition density for the Gaussian OU model: Sampler I: $PMMH\left(\boldsymbol{\alpha},\boldsymbol{\tau}^{2},\boldsymbol{\mu},\boldsymbol{\tau}_{f}^{2}\right)+PG\left(\boldsymbol{\beta},\boldsymbol{\phi}\right)$,
Sampler $II$: $PGAT\left(\boldsymbol{\beta},\boldsymbol{\alpha},\boldsymbol{\tau}^{2},\boldsymbol{\mu},\boldsymbol{\phi},\boldsymbol{\tau}_{f}^{2}\right)$,
sampler III: $PGBS\left(\boldsymbol{\beta},\boldsymbol{\alpha},\boldsymbol{\tau}^{2},\boldsymbol{\mu},\boldsymbol{\phi},\boldsymbol{\tau}_{f}^{2}\right)$
for simulated data with $T=1000$, $S=20$,
and $K=1$, and number of particles $N=1000$.\label{tab:Inefficiency-factor-of simulation-1}}

\centering{}%
\begin{tabular}{cccccccccccccccc}
\hline
 & {\footnotesize{}I} & {\footnotesize{}II} & {\footnotesize{}III} &  & {\footnotesize{}I} & {\footnotesize{}II} & {\footnotesize{}III} &  & {\footnotesize{}I} & {\footnotesize{}II} & {\footnotesize{}III} &  & {\footnotesize{}I} & {\footnotesize{}II} & {\footnotesize{}III}\tabularnewline
\hline
{\footnotesize{}$\beta_{1}$} & {\footnotesize{}$13.72$} & {\footnotesize{}$13.67$} & {\footnotesize{}$11.06$} & {\footnotesize{}$\alpha_{1}$} & {\footnotesize{}$12.85$} & {\footnotesize{}$159.79$} & {\footnotesize{}$181.78$} & {\footnotesize{}$\tau_{\epsilon,1}^{2}$} & {\footnotesize{}$13.10$} & {\footnotesize{}$374.25$} & {\footnotesize{}$444.82$} & {\footnotesize{}$\mu_{1}$} & {\footnotesize{}$13.27$} & {\footnotesize{}$12.92$} & {\footnotesize{}$11.82$}\tabularnewline
{\footnotesize{}$\beta_{2}$} & {\footnotesize{}$13.93$} & {\footnotesize{}$13.79$} & {\footnotesize{}$11.23$} & {\footnotesize{}$\alpha_{2}$} & {\footnotesize{}$15.49$} & {\footnotesize{}$92.87$} & {\footnotesize{}$335.05$} & {\footnotesize{}$\tau_{\epsilon,2}^{2}$} & {\footnotesize{}$13.49$} & {\footnotesize{}$340.28$} & {\footnotesize{}$792.88$} & {\footnotesize{}$\mu_{2}$} & {\footnotesize{}$13.42$} & {\footnotesize{}$11.00$} & {\footnotesize{}$11.98$}\tabularnewline
{\footnotesize{}$\beta_{3}$} & {\footnotesize{}$13.87$} & {\footnotesize{}$13.60$} & {\footnotesize{}$11.30$} & {\footnotesize{}$\alpha_{3}$} & {\footnotesize{}$12.43$} & {\footnotesize{}$300.77$} & {\footnotesize{}$272.34$} & {\footnotesize{}$\tau_{\epsilon,3}^{2}$} & {\footnotesize{}$12.46$} & {\footnotesize{}$733.43$} & {\footnotesize{}$682.28$} & {\footnotesize{}$\mu_{3}$} & {\footnotesize{}$15.41$} & {\footnotesize{}$13.09$} & {\footnotesize{}$13.23$}\tabularnewline
{\footnotesize{}$\beta_{4}$} & {\footnotesize{}$14.14$} & {\footnotesize{}$13.48$} & {\footnotesize{}$10.95$} & {\footnotesize{}$\alpha_{4}$} & {\footnotesize{}$13.35$} & {\footnotesize{}$530.99$} & {\footnotesize{}$303.41$} & {\footnotesize{}$\tau_{\epsilon,4}^{2}$} & {\footnotesize{}$13.46$} & {\footnotesize{}$977.93$} & {\footnotesize{}$654.65$} & {\footnotesize{}$\mu_{4}$} & {\footnotesize{}$13.44$} & {\footnotesize{}$13.87$} & {\footnotesize{}$14.16$}\tabularnewline
{\footnotesize{}$\beta_{5}$} & {\footnotesize{}$13.63$} & {\footnotesize{}$13.56$} & {\footnotesize{}$10.95$} & {\footnotesize{}$\alpha_{5}$} & {\footnotesize{}$15.72$} & {\footnotesize{}$93.77$} & {\footnotesize{}$140.44$} & {\footnotesize{}$\tau_{\epsilon,5}^{2}$} & {\footnotesize{}$16.23$} & {\footnotesize{}$514.24$} & {\footnotesize{}$339.73$} & {\footnotesize{}$\mu_{5}$} & {\footnotesize{}$13.24$} & {\footnotesize{}$11.83$} & {\footnotesize{}$12.87$}\tabularnewline
{\footnotesize{}$\beta_{6}$} & {\footnotesize{}$13.84$} & {\footnotesize{}$13.68$} & {\footnotesize{}$11.30$} & {\footnotesize{}$\alpha_{6}$} & {\footnotesize{}$16.81$} & {\footnotesize{}$190.71$} & {\footnotesize{}$152.17$} & {\footnotesize{}$\tau_{\epsilon,6}^{2}$} & {\footnotesize{}$16.20$} & {\footnotesize{}$539.97$} & {\footnotesize{}$418.23$} & {\footnotesize{}$\mu_{6}$} & {\footnotesize{}$14.00$} & {\footnotesize{}$13.23$} & {\footnotesize{}$13.45$}\tabularnewline
{\footnotesize{}$\beta_{7}$} & {\footnotesize{}$13.77$} & {\footnotesize{}$13.69$} & {\footnotesize{}$11.25$} & {\footnotesize{}$\alpha_{7}$} & {\footnotesize{}$17.57$} & {\footnotesize{}$79.74$} & {\footnotesize{}$102.55$} & {\footnotesize{}$\tau_{\epsilon,7}^{2}$} & {\footnotesize{}$13.75$} & {\footnotesize{}$592.05$} & {\footnotesize{}$352.65$} & {\footnotesize{}$\mu_{7}$} & {\footnotesize{}$13.80$} & {\footnotesize{}$10.85$} & {\footnotesize{}$11.77$}\tabularnewline
{\footnotesize{}$\beta_{8}$} & {\footnotesize{}$13.87$} & {\footnotesize{}$13.52$} & {\footnotesize{}$11.14$} & {\footnotesize{}$\alpha_{8}$} & {\footnotesize{}$13.33$} & {\footnotesize{}$134.56$} & {\footnotesize{}$136.97$} & {\footnotesize{}$\tau_{\epsilon,8}^{2}$} & {\footnotesize{}$13.99$} & {\footnotesize{}$392.80$} & {\footnotesize{}$376.86$} & {\footnotesize{}$\mu_{8}$} & {\footnotesize{}$16.46$} & {\footnotesize{}$11.48$} & {\footnotesize{}$11.67$}\tabularnewline
{\footnotesize{}$\beta_{9}$} & {\footnotesize{}$13.69$} & {\footnotesize{}$13.39$} & {\footnotesize{}$11.15$} & {\footnotesize{}$\alpha_{9}$} & {\footnotesize{}$13.50$} & {\footnotesize{}$395.91$} & {\footnotesize{}$161.91$} & {\footnotesize{}$\tau_{\epsilon,9}^{2}$} & {\footnotesize{}$14.65$} & {\footnotesize{}$803.36$} & {\footnotesize{}$457.15$} & {\footnotesize{}$\mu_{9}$} & {\footnotesize{}$16.12$} & {\footnotesize{}$13.72$} & {\footnotesize{}$12.55$}\tabularnewline
{\footnotesize{}$\beta_{10}$} & {\footnotesize{}$13.95$} & {\footnotesize{}$13.66$} & {\footnotesize{}$11.19$} & {\footnotesize{}$\alpha_{10}$} & {\footnotesize{}$12.46$} & {\footnotesize{}$128.96$} & {\footnotesize{}$117.10$} & {\footnotesize{}$\tau_{\epsilon,10}^{2}$} & {\footnotesize{}$13.10$} & {\footnotesize{}$408.40$} & {\footnotesize{}$357.97$} & {\footnotesize{}$\mu_{10}$} & {\footnotesize{}$14.72$} & {\footnotesize{}$11.70$} & {\footnotesize{}$11.94$}\tabularnewline
{\footnotesize{}$\beta_{11}$} & {\footnotesize{}$13.99$} & {\footnotesize{}$13.84$} & {\footnotesize{}$11.14$} & {\footnotesize{}$\alpha_{11}$} & {\footnotesize{}$13.55$} & {\footnotesize{}$273.87$} & {\footnotesize{}$98.71$} & {\footnotesize{}$\tau_{\epsilon,11}^{2}$} & {\footnotesize{}$15.56$} & {\footnotesize{}$667.52$} & {\footnotesize{}$402.61$} & {\footnotesize{}$\mu_{11}$} & {\footnotesize{}$12.55$} & {\footnotesize{}$11.51$} & {\footnotesize{}$12.62$}\tabularnewline
{\footnotesize{}$\beta_{12}$} & {\footnotesize{}$13.85$} & {\footnotesize{}$13.78$} & {\footnotesize{}$11.32$} & {\footnotesize{}$\alpha_{12}$} & {\footnotesize{}$16.34$} & {\footnotesize{}$105.64$} & {\footnotesize{}$204.73$} & {\footnotesize{}$\tau_{\epsilon,12}^{2}$} & {\footnotesize{}$16.09$} & {\footnotesize{}$356.37$} & {\footnotesize{}$438.96$} & {\footnotesize{}$\mu_{12}$} & {\footnotesize{}$12.56$} & {\footnotesize{}$13.00$} & {\footnotesize{}$13.25$}\tabularnewline
{\footnotesize{}$\beta_{13}$} & {\footnotesize{}$14.20$} & {\footnotesize{}$13.56$} & {\footnotesize{}$11.13$} & {\footnotesize{}$\alpha_{13}$} & {\footnotesize{}$13.56$} & {\footnotesize{}$262.15$} & {\footnotesize{}$136.41$} & {\footnotesize{}$\tau_{\epsilon,13}^{2}$} & {\footnotesize{}$12.73$} & {\footnotesize{}$511.17$} & {\footnotesize{}$378.67$} & {\footnotesize{}$\mu_{13}$} & {\footnotesize{}$13.18$} & {\footnotesize{}$14.97$} & {\footnotesize{}$11.28$}\tabularnewline
{\footnotesize{}$\beta_{14}$} & {\footnotesize{}$14.12$} & {\footnotesize{}$13.92$} & {\footnotesize{}$11.34$} & {\footnotesize{}$\alpha_{14}$} & {\footnotesize{}$12.60$} & {\footnotesize{}$188.22$} & {\footnotesize{}$177.73$} & {\footnotesize{}$\tau_{\epsilon,14}^{2}$} & {\footnotesize{}$12.00$} & {\footnotesize{}$530.42$} & {\footnotesize{}$428.24$} & {\footnotesize{}$\mu_{14}$} & {\footnotesize{}$16.19$} & {\footnotesize{}$12.18$} & {\footnotesize{}$11.69$}\tabularnewline
{\footnotesize{}$\beta_{15}$} & {\footnotesize{}$13.65$} & {\footnotesize{}$13.27$} & {\footnotesize{}$11.00$} & {\footnotesize{}$\alpha_{15}$} & {\footnotesize{}$14.79$} & {\footnotesize{}$200.20$} & {\footnotesize{}$162.37$} & {\footnotesize{}$\tau_{\epsilon,15}^{2}$} & {\footnotesize{}$12.79$} & {\footnotesize{}$574.45$} & {\footnotesize{}$578.06$} & {\footnotesize{}$\mu_{15}$} & {\footnotesize{}$15.09$} & {\footnotesize{}$13.01$} & {\footnotesize{}$12.46$}\tabularnewline
{\footnotesize{}$\beta_{16}$} & {\footnotesize{}$13.89$} & {\footnotesize{}$13.89$} & {\footnotesize{}$11.07$} & {\footnotesize{}$\alpha_{16}$} & {\footnotesize{}$14.62$} & {\footnotesize{}$271.96$} & {\footnotesize{}$337.69$} & {\footnotesize{}$\tau_{\epsilon,16}^{2}$} & {\footnotesize{}$15.67$} & {\footnotesize{}$470.91$} & {\footnotesize{}$672.67$} & {\footnotesize{}$\mu_{16}$} & {\footnotesize{}$13.51$} & {\footnotesize{}$15.99$} & {\footnotesize{}$11.88$}\tabularnewline
{\footnotesize{}$\beta_{17}$} & {\footnotesize{}$13.77$} & {\footnotesize{}$13.30$} & {\footnotesize{}$11.07$} & {\footnotesize{}$\alpha_{17}$} & {\footnotesize{}$16.29$} & {\footnotesize{}$139.51$} & {\footnotesize{}$87.63$} & {\footnotesize{}$\tau_{\epsilon,17}^{2}$} & {\footnotesize{}$13.62$} & {\footnotesize{}$467.94$} & {\footnotesize{}$330.15$} & {\footnotesize{}$\mu_{17}$} & {\footnotesize{}$16.63$} & {\footnotesize{}$12.34$} & {\footnotesize{}$13.24$}\tabularnewline
{\footnotesize{}$\beta_{18}$} & {\footnotesize{}$13.71$} & {\footnotesize{}$13.40$} & {\footnotesize{}$10.96$} & {\footnotesize{}$\alpha_{18}$} & {\footnotesize{}$15.69$} & {\footnotesize{}$55.90$} & {\footnotesize{}$107.32$} & {\footnotesize{}$\tau_{\epsilon,18}^{2}$} & {\footnotesize{}$17.08$} & {\footnotesize{}$262.38$} & {\footnotesize{}$317.31$} & {\footnotesize{}$\mu_{18}$} & {\footnotesize{}$15.03$} & {\footnotesize{}$10.81$} & {\footnotesize{}$11.65$}\tabularnewline
{\footnotesize{}$\beta_{19}$} & {\footnotesize{}$13.90$} & {\footnotesize{}$13.69$} & {\footnotesize{}$11.05$} & {\footnotesize{}$\alpha_{19}$} & {\footnotesize{}$15.73$} & {\footnotesize{}$284.70$} & {\footnotesize{}$194.08$} & {\footnotesize{}$\tau_{\epsilon,19}^{2}$} & {\footnotesize{}$14.97$} & {\footnotesize{}$649.26$} & {\footnotesize{}$537.12$} & {\footnotesize{}$\mu_{19}$} & {\footnotesize{}$15.39$} & {\footnotesize{}$13.53$} & {\footnotesize{}$11.72$}\tabularnewline
{\footnotesize{}$\beta_{20}$} & {\footnotesize{}$13.86$} & {\footnotesize{}$13.61$} & {\footnotesize{}$11.21$} & {\footnotesize{}$\alpha_{20}$} & {\footnotesize{}$13.76$} & {\footnotesize{}$311.20$} & {\footnotesize{}$125.72$} & {\footnotesize{}$\tau_{\epsilon,20}^{2}$} & {\footnotesize{}$14.99$} & {\footnotesize{}$667.49$} & {\footnotesize{}$331.18$} & {\footnotesize{}$\mu_{20}$} & {\footnotesize{}$14.64$} & {\footnotesize{}$16.43$} & {\footnotesize{}$15.96$}\tabularnewline
$\phi$ & {\footnotesize{}$7.11$} & {\footnotesize{}$20.88$} & {\footnotesize{}$17.01$} & {\footnotesize{}$\tau_{f,1}^{2}$} & {\footnotesize{}$12.66$} & {\footnotesize{}$78.23$} & {\footnotesize{}$67.92$} &  &  &  &  &  &  &  & \tabularnewline
\hline
\end{tabular}
\end{sidewaystable}


\begin{sidewaystable}
\caption{Inefficiency factors of $\boldsymbol{\beta}$, $\boldsymbol{\alpha}$,
$\boldsymbol{\mu}$, $\boldsymbol{\tau}^{2}$, $\boldsymbol{\phi}$,
and $\boldsymbol{\tau}_{f}^{2}$ with exact transition density for the Gaussian OU model: Sampler
I: $PMMH\left(\boldsymbol{\alpha},\boldsymbol{\tau}^{2},\boldsymbol{\tau}_{f}^{2}\right)+PG\left(\boldsymbol{\beta},\boldsymbol{\mu},\boldsymbol{\phi}\right)$,
Sampler $II$: $PGAT\left(\boldsymbol{\beta},\boldsymbol{\alpha},\boldsymbol{\tau}^{2},\boldsymbol{\mu},\boldsymbol{\phi},\boldsymbol{\tau}_{f}^{2}\right)$,
sampler III: $PGBS\left(\boldsymbol{\beta},\boldsymbol{\alpha},\boldsymbol{\tau}^{2},\boldsymbol{\mu},\boldsymbol{\phi},\boldsymbol{\tau}_{f}^{2}\right)$
for US stock returns data with $T=1000$,
$S=20$, and $K=1$, and number of particles $N=500$.\label{tab:Inefficiency-factor-of real data}}
\centering{}%
\begin{tabular}{cccccccccccccccc}
\hline
 & {\footnotesize{}I} & {\footnotesize{}II} & {\footnotesize{}III} &  & {\footnotesize{}I} & {\footnotesize{}II} & {\footnotesize{}III} &  & {\footnotesize{}I} & {\footnotesize{}II} & {\footnotesize{}III} &  & {\footnotesize{}I} & {\footnotesize{}II} & {\footnotesize{}III}\tabularnewline
\hline
{\footnotesize{}$\beta_{1}$} & {\footnotesize{}$2.18$} & {\footnotesize{}$2.05$} & {\footnotesize{}$1.91$} & {\footnotesize{}$\alpha_{1}$} & {\footnotesize{}$14.21$} & {\footnotesize{}$219.61$} & {\footnotesize{}$113.66$} & {\footnotesize{}$\tau_{\epsilon,1}^{2}$} & {\footnotesize{}$14.37$} & {\footnotesize{}$260.88$} & {\footnotesize{}$129.79$} & {\footnotesize{}$\mu_{1}$} & {\footnotesize{}$2.11$} & {\footnotesize{}$4.50$} & {\footnotesize{}$2.84$}\tabularnewline
{\footnotesize{}$\beta_{2}$} & {\footnotesize{}$1.68$} & {\footnotesize{}$1.85$} & {\footnotesize{}$1.90$} & {\footnotesize{}$\alpha_{2}$} & {\footnotesize{}$11.87$} & {\footnotesize{}$35.87$} & {\footnotesize{}$40.80$} & {\footnotesize{}$\tau_{\epsilon,2}^{2}$} & {\footnotesize{}$12.29$} & {\footnotesize{}$68.34$} & {\footnotesize{}$70.17$} & {\footnotesize{}$\mu_{2}$} & {\footnotesize{}$1.20$} & {\footnotesize{}$1.42$} & {\footnotesize{}$1.18$}\tabularnewline
{\footnotesize{}$\beta_{3}$} & {\footnotesize{}$1.80$} & {\footnotesize{}$1.76$} & {\footnotesize{}$1.70$} & {\footnotesize{}$\alpha_{3}$} & {\footnotesize{}$13.04$} & {\footnotesize{}$62.04$} & {\footnotesize{}$89.69$} & {\footnotesize{}$\tau_{\epsilon,3}^{2}$} & {\footnotesize{}$13.23$} & {\footnotesize{}$110.88$} & {\footnotesize{}$157.46$} & {\footnotesize{}$\mu_{3}$} & {\footnotesize{}$2.39$} & {\footnotesize{}$2.66$} & {\footnotesize{}$2.36$}\tabularnewline
{\footnotesize{}$\beta_{4}$} & {\footnotesize{}$1.79$} & {\footnotesize{}$1.76$} & {\footnotesize{}$1.83$} & {\footnotesize{}$\alpha_{4}$} & {\footnotesize{}$14.22$} & {\footnotesize{}$66.24$} & {\footnotesize{}$51.79$} & {\footnotesize{}$\tau_{\epsilon,4}^{2}$} & {\footnotesize{}$14.99$} & {\footnotesize{}$122.26$} & {\footnotesize{}$88.17$} & {\footnotesize{}$\mu_{4}$} & {\footnotesize{}$1.77$} & {\footnotesize{}$1.83$} & {\footnotesize{}$1.50$}\tabularnewline
{\footnotesize{}$\beta_{5}$} & {\footnotesize{}$1.87$} & {\footnotesize{}$1.76$} & {\footnotesize{}$1.69$} & {\footnotesize{}$\alpha_{5}$} & {\footnotesize{}$18.44$} & {\footnotesize{}$466.48$} & {\footnotesize{}$136.77$} & {\footnotesize{}$\tau_{\epsilon,5}^{2}$} & {\footnotesize{}$17.14$} & {\footnotesize{}$682.49$} & {\footnotesize{}$167.35$} & {\footnotesize{}$\mu_{5}$} & {\footnotesize{}$2.97$} & {\footnotesize{}$3.57$} & {\footnotesize{}$1.91$}\tabularnewline
{\footnotesize{}$\beta_{6}$} & {\footnotesize{}$1.66$} & {\footnotesize{}$1.74$} & {\footnotesize{}$1.67$} & {\footnotesize{}$\alpha_{6}$} & {\footnotesize{}$17.31$} & {\footnotesize{}$113.00$} & {\footnotesize{}$112.08$} & {\footnotesize{}$\tau_{\epsilon,6}^{2}$} & {\footnotesize{}$19.29$} & {\footnotesize{}$202.66$} & {\footnotesize{}$258.42$} & {\footnotesize{}$\mu_{6}$} & {\footnotesize{}$4.88$} & {\footnotesize{}$5.94$} & {\footnotesize{}$4.11$}\tabularnewline
{\footnotesize{}$\beta_{7}$} & {\footnotesize{}$1.61$} & {\footnotesize{}$1.67$} & {\footnotesize{}$1.66$} & {\footnotesize{}$\alpha_{7}$} & {\footnotesize{}$11.41$} & {\footnotesize{}$52.72$} & {\footnotesize{}$64.09$} & {\footnotesize{}$\tau_{\epsilon,7}^{2}$} & {\footnotesize{}$14.00$} & {\footnotesize{}$91.79$} & {\footnotesize{}$92.67$} & {\footnotesize{}$\mu_{7}$} & {\footnotesize{}$1.87$} & {\footnotesize{}$1.79$} & {\footnotesize{}$1.86$}\tabularnewline
{\footnotesize{}$\beta_{8}$} & {\footnotesize{}$1.82$} & {\footnotesize{}$1.93$} & {\footnotesize{}$1.70$} & {\footnotesize{}$\alpha_{8}$} & {\footnotesize{}$18.71$} & {\footnotesize{}$86.37$} & {\footnotesize{}$45.28$} & {\footnotesize{}$\tau_{\epsilon,8}^{2}$} & {\footnotesize{}$20.57$} & {\footnotesize{}$145.71$} & {\footnotesize{}$76.37$} & {\footnotesize{}$\mu_{8}$} & {\footnotesize{}$2.43$} & {\footnotesize{}$3.41$} & {\footnotesize{}$1.80$}\tabularnewline
{\footnotesize{}$\beta_{9}$} & {\footnotesize{}$1.89$} & {\footnotesize{}$1.96$} & {\footnotesize{}$1.74$} & {\footnotesize{}$\alpha_{9}$} & {\footnotesize{}$12.97$} & {\footnotesize{}$80.73$} & {\footnotesize{}$136.71$} & {\footnotesize{}$\tau_{\epsilon,9}^{2}$} & {\footnotesize{}$14.23$} & {\footnotesize{}$116.44$} & {\footnotesize{}$158.23$} & {\footnotesize{}$\mu_{9}$} & {\footnotesize{}$2.27$} & {\footnotesize{}$2.77$} & {\footnotesize{}$3.30$}\tabularnewline
{\footnotesize{}$\beta_{10}$} & {\footnotesize{}$1.65$} & {\footnotesize{}$1.73$} & {\footnotesize{}$1.66$} & {\footnotesize{}$\alpha_{10}$} & {\footnotesize{}$15.25$} & {\footnotesize{}$119.34$} & {\footnotesize{}$124.61$} & {\footnotesize{}$\tau_{\epsilon,10}^{2}$} & {\footnotesize{}$12.54$} & {\footnotesize{}$106.68$} & {\footnotesize{}$128.63$} & {\footnotesize{}$\mu_{10}$} & {\footnotesize{}$6.21$} & {\footnotesize{}$7.57$} & {\footnotesize{}$6.70$}\tabularnewline
{\footnotesize{}$\beta_{11}$} & {\footnotesize{}$1.63$} & {\footnotesize{}$1.74$} & {\footnotesize{}$1.67$} & {\footnotesize{}$\alpha_{11}$} & {\footnotesize{}$14.66$} & {\footnotesize{}$65.71$} & {\footnotesize{}$69.71$} & {\footnotesize{}$\tau_{\epsilon,11}^{2}$} & {\footnotesize{}$14.44$} & {\footnotesize{}$121.39$} & {\footnotesize{}$83.53$} & {\footnotesize{}$\mu_{11}$} & {\footnotesize{}$3.24$} & {\footnotesize{}$5.57$} & {\footnotesize{}$2.84$}\tabularnewline
{\footnotesize{}$\beta_{12}$} & {\footnotesize{}$1.65$} & {\footnotesize{}$1.89$} & {\footnotesize{}$1.69$} & {\footnotesize{}$\alpha_{12}$} & {\footnotesize{}$17.47$} & {\footnotesize{}$433.51$} & {\footnotesize{}$97.20$} & {\footnotesize{}$\tau_{\epsilon,12}^{2}$} & {\footnotesize{}$16.20$} & {\footnotesize{}$545.21$} & {\footnotesize{}$146.63$} & {\footnotesize{}$\mu_{12}$} & {\footnotesize{}$3.36$} & {\footnotesize{}$5.94$} & {\footnotesize{}$2.54$}\tabularnewline
{\footnotesize{}$\beta_{13}$} & {\footnotesize{}$1.94$} & {\footnotesize{}$2.02$} & {\footnotesize{}$1.92$} & {\footnotesize{}$\alpha_{13}$} & {\footnotesize{}$13.50$} & {\footnotesize{}$151.20$} & {\footnotesize{}$112.64$} & {\footnotesize{}$\tau_{\epsilon,13}^{2}$} & {\footnotesize{}$13.49$} & {\footnotesize{}$189.17$} & {\footnotesize{}$145.44$} & {\footnotesize{}$\mu_{13}$} & {\footnotesize{}$2.74$} & {\footnotesize{}$3.19$} & {\footnotesize{}$2.19$}\tabularnewline
{\footnotesize{}$\beta_{14}$} & {\footnotesize{}$1.66$} & {\footnotesize{}$1.79$} & {\footnotesize{}$1.60$} & {\footnotesize{}$\alpha_{14}$} & {\footnotesize{}$14.48$} & {\footnotesize{}$70.44$} & {\footnotesize{}$74.94$} & {\footnotesize{}$\tau_{\epsilon,14}^{2}$} & {\footnotesize{}$14.11$} & {\footnotesize{}$146.32$} & {\footnotesize{}$121.04$} & {\footnotesize{}$\mu_{14}$} & {\footnotesize{}$2.01$} & {\footnotesize{}$2.06$} & {\footnotesize{}$1.73$}\tabularnewline
{\footnotesize{}$\beta_{15}$} & {\footnotesize{}$1.62$} & {\footnotesize{}$1.82$} & {\footnotesize{}$1.45$} & {\footnotesize{}$\alpha_{15}$} & {\footnotesize{}$13.08$} & {\footnotesize{}$126.39$} & {\footnotesize{}$291.78$} & {\footnotesize{}$\tau_{\epsilon,15}^{2}$} & {\footnotesize{}$14.80$} & {\footnotesize{}$148.03$} & {\footnotesize{}$382.86$} & {\footnotesize{}$\mu_{15}$} & {\footnotesize{}$2.20$} & {\footnotesize{}$2.66$} & {\footnotesize{}$2.11$}\tabularnewline
{\footnotesize{}$\beta_{16}$} & {\footnotesize{}$1.69$} & {\footnotesize{}$1.76$} & {\footnotesize{}$1.83$} & {\footnotesize{}$\alpha_{16}$} & {\footnotesize{}$11.58$} & {\footnotesize{}$133.17$} & {\footnotesize{}$39.94$} & {\footnotesize{}$\tau_{\epsilon,16}^{2}$} & {\footnotesize{}$11.64$} & {\footnotesize{}$210.38$} & {\footnotesize{}$99.40$} & {\footnotesize{}$\mu_{16}$} & {\footnotesize{}$1.54$} & {\footnotesize{}$1.54$} & {\footnotesize{}$1.55$}\tabularnewline
{\footnotesize{}$\beta_{17}$} & {\footnotesize{}$2.12$} & {\footnotesize{}$2.54$} & {\footnotesize{}$1.95$} & {\footnotesize{}$\alpha_{17}$} & {\footnotesize{}$14.52$} & {\footnotesize{}$39.97$} & {\footnotesize{}$30.94$} & {\footnotesize{}$\tau_{\epsilon,17}^{2}$} & {\footnotesize{}$15.65$} & {\footnotesize{}$94.23$} & {\footnotesize{}$54.03$} & {\footnotesize{}$\mu_{17}$} & {\footnotesize{}$1.30$} & {\footnotesize{}$1.25$} & {\footnotesize{}$1.24$}\tabularnewline
{\footnotesize{}$\beta_{18}$} & {\footnotesize{}$1.94$} & {\footnotesize{}$2.04$} & {\footnotesize{}$1.93$} & {\footnotesize{}$\alpha_{18}$} & {\footnotesize{}$15.24$} & {\footnotesize{}$51.58$} & {\footnotesize{}$40.02$} & {\footnotesize{}$\tau_{\epsilon,18}^{2}$} & {\footnotesize{}$17.46$} & {\footnotesize{}$105.41$} & {\footnotesize{}$70.14$} & {\footnotesize{}$\mu_{18}$} & {\footnotesize{}$1.36$} & {\footnotesize{}$1.51$} & {\footnotesize{}$1.36$}\tabularnewline
{\footnotesize{}$\beta_{19}$} & {\footnotesize{}$1.80$} & {\footnotesize{}$1.92$} & {\footnotesize{}$1.73$} & {\footnotesize{}$\alpha_{19}$} & {\footnotesize{}$15.14$} & {\footnotesize{}$36.14$} & {\footnotesize{}$28.02$} & {\footnotesize{}$\tau_{\epsilon,19}^{2}$} & {\footnotesize{}$13.73$} & {\footnotesize{}$81.59$} & {\footnotesize{}$68.35$} & {\footnotesize{}$\mu_{19}$} & {\footnotesize{}$1.28$} & {\footnotesize{}$1.48$} & {\footnotesize{}$1.37$}\tabularnewline
{\footnotesize{}$\beta_{20}$} & {\footnotesize{}$1.87$} & {\footnotesize{}$1.81$} & {\footnotesize{}$1.73$} & {\footnotesize{}$\alpha_{20}$} & {\footnotesize{}$14.52$} & {\footnotesize{}$33.78$} & {\footnotesize{}$28.57$} & {\footnotesize{}$\tau_{\epsilon,20}^{2}$} & {\footnotesize{}$17.10$} & {\footnotesize{}$72.67$} & {\footnotesize{}$55.28$} & {\footnotesize{}$\mu_{20}$} & {\footnotesize{}$1.27$} & {\footnotesize{}$1.51$} & {\footnotesize{}$1.22$}\tabularnewline
$\phi$ & {\footnotesize{}$8.77$} & {\footnotesize{}$25.64$} & {\footnotesize{}$20.05$} & {\footnotesize{}$\tau_{f,1}^{2}$} & {\footnotesize{}$14.24$} & {\footnotesize{}$55.08$} & {\footnotesize{}$48.92$} &  &  &  &  &  &  &  & \tabularnewline
\hline
\end{tabular}
\end{sidewaystable}

Table~\ref{tab:Inefficiency-factor-of real data} gives the inefficiency factors of $\boldsymbol{\beta}$, $\boldsymbol{\alpha}$,
$\boldsymbol{\mu}$, $\boldsymbol{\tau}_{\epsilon}^{2}$, $\boldsymbol{\phi}$,
and $\boldsymbol{\tau}_{f}^{2}$ with the exact transition density for the Gaussian OU model for the three samplers: Sampler
I: $PMMH\left(\boldsymbol{\alpha},\boldsymbol{\tau}_{\epsilon}^{2},\boldsymbol{\tau}_{f}^{2}\right)+PG\left(\boldsymbol{\beta},\boldsymbol{\mu},\boldsymbol{\phi}\right)$,
Sampler $II$: $PGAT\left(\boldsymbol{\beta},\boldsymbol{\alpha},\boldsymbol{\tau}_{\epsilon}^{2},\boldsymbol{\mu},\boldsymbol{\phi},\boldsymbol{\tau}_{f}^{2}\right)$,
sampler III: $PGBS\left(\boldsymbol{\beta},\boldsymbol{\alpha},\boldsymbol{\tau}_{\epsilon}^{2},\boldsymbol{\mu},\boldsymbol{\phi},\boldsymbol{\tau}_{f}^{2}\right)$
for US stock returns data with $T=1000$,
$S=20$, and $K=1$, and with the number of particles $N=500$.

Table~\ref{tab:Inefficiency-factor-of real data-1} gives the inefficiency factors of $\boldsymbol{\beta}$, $\boldsymbol{\alpha}$,
$\boldsymbol{\mu}$, $\boldsymbol{\tau}_{\epsilon}^{2}$, $\boldsymbol{\phi}$,
and $\boldsymbol{\tau}_{f}^{2}$ with the approximate Euler based transition density for the Gaussian OU model, for the three samplers:
Sampler
I: $PMMH\left(\boldsymbol{\alpha},\boldsymbol{\tau}_{\epsilon}^{2},\boldsymbol{\tau}_{f}^{2}\right)+PG\left(\boldsymbol{\beta},\boldsymbol{\mu},\boldsymbol{\phi}\right)$,
Sampler $II$: $PGAT\left(\boldsymbol{\beta},\boldsymbol{\alpha},\boldsymbol{\tau}_{\epsilon}^{2},\boldsymbol{\mu},\boldsymbol{\phi},\boldsymbol{\tau}_{f}^{2}\right)$, Sampler $III$: $PGBS\left(\boldsymbol{\beta},\boldsymbol{\alpha},\boldsymbol{\tau}_{\epsilon}^{2},\boldsymbol{\mu},\boldsymbol{\phi},\boldsymbol{\tau}_{f}^{2}\right)$
for US stock returns data with $T=1000$,
$S=20$, and $K=1$, and with the number of particles $N=1000$.

\begin{sidewaystable}
\caption{Inefficiency factors of $\boldsymbol{\beta}$, $\boldsymbol{\alpha}$,
$\boldsymbol{\mu}$, $\boldsymbol{\tau}^{2}$, $\boldsymbol{\phi}$,
and $\boldsymbol{\tau}_{f}^{2}$ with an Euler approximation for the state
transition densities for the Gaussian OU model: Sampler I: $PMMH\left(\boldsymbol{\alpha},\boldsymbol{\tau}^{2},\boldsymbol{\mu},\boldsymbol{\tau}_{f}^{2}\right)+PG\left(\boldsymbol{\beta},\boldsymbol{\phi}\right)$,
Sampler $II$: $PGAT\left(\boldsymbol{\beta},\boldsymbol{\alpha},\boldsymbol{\tau}^{2},\boldsymbol{\mu},\boldsymbol{\phi},\boldsymbol{\tau}_{f}^{2}\right)$,
sampler III: $PGBS\left(\boldsymbol{\beta},\boldsymbol{\alpha},\boldsymbol{\tau}^{2},\boldsymbol{\mu},\boldsymbol{\phi},\boldsymbol{\tau}_{f}^{2}\right)$
for US stock returns data with $T=1000$,
$S=20$, and $K=1$, and number of particles $N=1000$.\label{tab:Inefficiency-factor-of real data-1}}

\centering{}%
\begin{tabular}{cccccccccccccccc}
\hline
 & {\footnotesize{}I} & {\footnotesize{}II} & {\footnotesize{}III} &  & {\footnotesize{}I} & {\footnotesize{}II} & {\footnotesize{}III} &  & {\footnotesize{}I} & {\footnotesize{}II} & {\footnotesize{}III} &  & {\footnotesize{}I} & {\footnotesize{}II} & {\footnotesize{}III}\tabularnewline
\hline
{\footnotesize{}$\beta_{1}$} & {\footnotesize{}$2.01$} & {\footnotesize{}$2.18$} & {\footnotesize{}$1.89$} & {\footnotesize{}$\alpha_{1}$} & {\footnotesize{}$15.52$} & {\footnotesize{}$559.41$} & {\footnotesize{}$723.77$} & {\footnotesize{}$\tau_{\epsilon,1}^{2}$} & {\footnotesize{}$15.33$} & {\footnotesize{}$787.45$} & {\footnotesize{}$977.73$} & {\footnotesize{}$\mu_{1}$} & {\footnotesize{}$11.99$} & {\footnotesize{}$23.88$} & {\footnotesize{}$18.17$}\tabularnewline
{\footnotesize{}$\beta_{2}$} & {\footnotesize{}$1.86$} & {\footnotesize{}$1.84$} & {\footnotesize{}$1.83$} & {\footnotesize{}$\alpha_{2}$} & {\footnotesize{}$16.40$} & {\footnotesize{}$554.76$} & {\footnotesize{}$186.87$} & {\footnotesize{}$\tau_{\epsilon,2}^{2}$} & {\footnotesize{}$14.35$} & {\footnotesize{}$914.39$} & {\footnotesize{}$475.00$} & {\footnotesize{}$\mu_{2}$} & {\footnotesize{}$15.46$} & {\footnotesize{}$12.34$} & {\footnotesize{}$11.59$}\tabularnewline
{\footnotesize{}$\beta_{3}$} & {\footnotesize{}$1.80$} & {\footnotesize{}$1.79$} & {\footnotesize{}$1.73$} & {\footnotesize{}$\alpha_{3}$} & {\footnotesize{}$18.50$} & {\footnotesize{}$342.35$} & {\footnotesize{}$210.83$} & {\footnotesize{}$\tau_{\epsilon,3}^{2}$} & {\footnotesize{}$19.40$} & {\footnotesize{}$688.81$} & {\footnotesize{}$546.86$} & {\footnotesize{}$\mu_{3}$} & {\footnotesize{}$13.45$} & {\footnotesize{}$12.56$} & {\footnotesize{}$12.66$}\tabularnewline
{\footnotesize{}$\beta_{4}$} & {\footnotesize{}$1.79$} & {\footnotesize{}$1.83$} & {\footnotesize{}$1.76$} & {\footnotesize{}$\alpha_{4}$} & {\footnotesize{}$15.40$} & {\footnotesize{}$215.12$} & {\footnotesize{}$111.11$} & {\footnotesize{}$\tau_{\epsilon,4}^{2}$} & {\footnotesize{}$16.44$} & {\footnotesize{}$455.87$} & {\footnotesize{}$326.75$} & {\footnotesize{}$\mu_{4}$} & {\footnotesize{}$12.11$} & {\footnotesize{}$12.82$} & {\footnotesize{}$12.83$}\tabularnewline
{\footnotesize{}$\beta_{5}$} & {\footnotesize{}$1.85$} & {\footnotesize{}$1.75$} & {\footnotesize{}$1.68$} & {\footnotesize{}$\alpha_{5}$} & {\footnotesize{}$15.82$} & {\footnotesize{}$308.00$} & {\footnotesize{}$305.18$} & {\footnotesize{}$\tau_{\epsilon,5}^{2}$} & {\footnotesize{}$19.29$} & {\footnotesize{}$576.04$} & {\footnotesize{}$456.33$} & {\footnotesize{}$\mu_{5}$} & {\footnotesize{}$21.39$} & {\footnotesize{}$16.70$} & {\footnotesize{}$20.62$}\tabularnewline
{\footnotesize{}$\beta_{6}$} & {\footnotesize{}$1.73$} & {\footnotesize{}$1.75$} & {\footnotesize{}$1.74$} & {\footnotesize{}$\alpha_{6}$} & {\footnotesize{}$20.72$} & {\footnotesize{}$494.91$} & {\footnotesize{}$374.78$} & {\footnotesize{}$\tau_{\epsilon,6}^{2}$} & {\footnotesize{}$20.06$} & {\footnotesize{}$995.03$} & {\footnotesize{}$797.53$} & {\footnotesize{}$\mu_{6}$} & {\footnotesize{}$19.43$} & {\footnotesize{}$26.62$} & {\footnotesize{}$36.67$}\tabularnewline
{\footnotesize{}$\beta_{7}$} & {\footnotesize{}$1.78$} & {\footnotesize{}$1.75$} & {\footnotesize{}$1.77$} & {\footnotesize{}$\alpha_{7}$} & {\footnotesize{}$16.07$} & {\footnotesize{}$340.91$} & {\footnotesize{}$464.08$} & {\footnotesize{}$\tau_{\epsilon,7}^{2}$} & {\footnotesize{}$14.49$} & {\footnotesize{}$783.92$} & {\footnotesize{}$754.46$} & {\footnotesize{}$\mu_{7}$} & {\footnotesize{}$13.71$} & {\footnotesize{}$13.56$} & {\footnotesize{}$14.80$}\tabularnewline
{\footnotesize{}$\beta_{8}$} & {\footnotesize{}$1.83$} & {\footnotesize{}$1.83$} & {\footnotesize{}$1.74$} & {\footnotesize{}$\alpha_{8}$} & {\footnotesize{}$19.85$} & {\footnotesize{}$400.60$} & {\footnotesize{}$128.31$} & {\footnotesize{}$\tau_{\epsilon,8}^{2}$} & {\footnotesize{}$23.81$} & {\footnotesize{}$928.38$} & {\footnotesize{}$328.45$} & {\footnotesize{}$\mu_{8}$} & {\footnotesize{}$18.30$} & {\footnotesize{}$13.56$} & {\footnotesize{}$12.63$}\tabularnewline
{\footnotesize{}$\beta_{9}$} & {\footnotesize{}$1.76$} & {\footnotesize{}$1.96$} & {\footnotesize{}$1.77$} & {\footnotesize{}$\alpha_{9}$} & {\footnotesize{}$15.19$} & {\footnotesize{}$909.99$} & {\footnotesize{}$546.01$} & {\footnotesize{}$\tau_{\epsilon,9}^{2}$} & {\footnotesize{}$14.84$} & {\footnotesize{}$1215.77$} & {\footnotesize{}$937.28$} & {\footnotesize{}$\mu_{9}$} & {\footnotesize{}$13.06$} & {\footnotesize{}$19.86$} & {\footnotesize{}$19.14$}\tabularnewline
{\footnotesize{}$\beta_{10}$} & {\footnotesize{}$1.74$} & {\footnotesize{}$1.78$} & {\footnotesize{}$1.77$} & {\footnotesize{}$\alpha_{10}$} & {\footnotesize{}$16.96$} & {\footnotesize{}$385.25$} & {\footnotesize{}$236.04$} & {\footnotesize{}$\tau_{\epsilon,10}^{2}$} & {\footnotesize{}$23.59$} & {\footnotesize{}$962.51$} & {\footnotesize{}$716.08$} & {\footnotesize{}$\mu_{10}$} & {\footnotesize{}$11.92$} & {\footnotesize{}$50.67$} & {\footnotesize{}$35.06$}\tabularnewline
{\footnotesize{}$\beta_{11}$} & {\footnotesize{}$1.77$} & {\footnotesize{}$1.78$} & {\footnotesize{}$1.74$} & {\footnotesize{}$\alpha_{11}$} & {\footnotesize{}$18.43$} & {\footnotesize{}$368.53$} & {\footnotesize{}$115.84$} & {\footnotesize{}$\tau_{\epsilon,11}^{2}$} & {\footnotesize{}$23.99$} & {\footnotesize{}$811.02$} & {\footnotesize{}$872.32$} & {\footnotesize{}$\mu_{11}$} & {\footnotesize{}$13.76$} & {\footnotesize{}$15.12$} & {\footnotesize{}$14.85$}\tabularnewline
{\footnotesize{}$\beta_{12}$} & {\footnotesize{}$1.81$} & {\footnotesize{}$1.82$} & {\footnotesize{}$1.77$} & {\footnotesize{}$\alpha_{12}$} & {\footnotesize{}$20.48$} & {\footnotesize{}$521.58$} & {\footnotesize{}$460.67$} & {\footnotesize{}$\tau_{\epsilon,12}^{2}$} & {\footnotesize{}$20.43$} & {\footnotesize{}$771.17$} & {\footnotesize{}$700.72$} & {\footnotesize{}$\mu_{12}$} & {\footnotesize{}$16.91$} & {\footnotesize{}$20.80$} & {\footnotesize{}$19.88$}\tabularnewline
{\footnotesize{}$\beta_{13}$} & {\footnotesize{}$1.81$} & {\footnotesize{}$1.86$} & {\footnotesize{}$1.83$} & {\footnotesize{}$\alpha_{13}$} & {\footnotesize{}$17.79$} & {\footnotesize{}$362.85$} & {\footnotesize{}$548.70$} & {\footnotesize{}$\tau_{\epsilon,13}^{2}$} & {\footnotesize{}$18.43$} & {\footnotesize{}$632.95$} & {\footnotesize{}$707.42$} & {\footnotesize{}$\mu_{13}$} & {\footnotesize{}$15.76$} & {\footnotesize{}$14.73$} & {\footnotesize{}$19.90$}\tabularnewline
{\footnotesize{}$\beta_{14}$} & {\footnotesize{}$1.77$} & {\footnotesize{}$1.79$} & {\footnotesize{}$1.64$} & {\footnotesize{}$\alpha_{14}$} & {\footnotesize{}$15.48$} & {\footnotesize{}$195.27$} & {\footnotesize{}$375.87$} & {\footnotesize{}$\tau_{\epsilon,14}^{2}$} & {\footnotesize{}$17.05$} & {\footnotesize{}$603.04$} & {\footnotesize{}$704.08$} & {\footnotesize{}$\mu_{14}$} & {\footnotesize{}$14.14$} & {\footnotesize{}$14.75$} & {\footnotesize{}$19.37$}\tabularnewline
{\footnotesize{}$\beta_{15}$} & {\footnotesize{}$1.59$} & {\footnotesize{}$1.69$} & {\footnotesize{}$1.57$} & {\footnotesize{}$\alpha_{15}$} & {\footnotesize{}$17.48$} & {\footnotesize{}$485.37$} & {\footnotesize{}$1097.26$} & {\footnotesize{}$\tau_{\epsilon,15}^{2}$} & {\footnotesize{}$15.58$} & {\footnotesize{}$897.29$} & {\footnotesize{}$1228.99$} & {\footnotesize{}$\mu_{15}$} & {\footnotesize{}$15.76$} & {\footnotesize{}$18.84$} & {\footnotesize{}$29.16$}\tabularnewline
{\footnotesize{}$\beta_{16}$} & {\footnotesize{}$1.80$} & {\footnotesize{}$1.70$} & {\footnotesize{}$1.74$} & {\footnotesize{}$\alpha_{16}$} & {\footnotesize{}$15.94$} & {\footnotesize{}$240.28$} & {\footnotesize{}$211.86$} & {\footnotesize{}$\tau_{\epsilon,16}^{2}$} & {\footnotesize{}$14.50$} & {\footnotesize{}$571.97$} & {\footnotesize{}$434.93$} & {\footnotesize{}$\mu_{16}$} & {\footnotesize{}$13.40$} & {\footnotesize{}$13.29$} & {\footnotesize{}$13.19$}\tabularnewline
{\footnotesize{}$\beta_{17}$} & {\footnotesize{}$2.14$} & {\footnotesize{}$2.12$} & {\footnotesize{}$2.02$} & {\footnotesize{}$\alpha_{17}$} & {\footnotesize{}$16.99$} & {\footnotesize{}$143.03$} & {\footnotesize{}$330.84$} & {\footnotesize{}$\tau_{\epsilon,17}^{2}$} & {\footnotesize{}$17.16$} & {\footnotesize{}$496.86$} & {\footnotesize{}$683.20$} & {\footnotesize{}$\mu_{17}$} & {\footnotesize{}$15.79$} & {\footnotesize{}$11.49$} & {\footnotesize{}$10.91$}\tabularnewline
{\footnotesize{}$\beta_{18}$} & {\footnotesize{}$1.88$} & {\footnotesize{}$1.96$} & {\footnotesize{}$1.87$} & {\footnotesize{}$\alpha_{18}$} & {\footnotesize{}$18.10$} & {\footnotesize{}$225.30$} & {\footnotesize{}$184.31$} & {\footnotesize{}$\tau_{\epsilon,18}^{2}$} & {\footnotesize{}$16.15$} & {\footnotesize{}$518.71$} & {\footnotesize{}$683.36$} & {\footnotesize{}$\mu_{18}$} & {\footnotesize{}$18.81$} & {\footnotesize{}$11.63$} & {\footnotesize{}$12.63$}\tabularnewline
{\footnotesize{}$\beta_{19}$} & {\footnotesize{}$1.84$} & {\footnotesize{}$1.88$} & {\footnotesize{}$1.79$} & {\footnotesize{}$\alpha_{19}$} & {\footnotesize{}$16.61$} & {\footnotesize{}$200.54$} & {\footnotesize{}$70.61$} & {\footnotesize{}$\tau_{\epsilon,19}^{2}$} & {\footnotesize{}$16.24$} & {\footnotesize{}$474.26$} & {\footnotesize{}$276.66$} & {\footnotesize{}$\mu_{19}$} & {\footnotesize{}$16.64$} & {\footnotesize{}$11.33$} & {\footnotesize{}$8.73$}\tabularnewline
{\footnotesize{}$\beta_{20}$} & {\footnotesize{}$1.91$} & {\footnotesize{}$1.86$} & {\footnotesize{}$1.77$} & {\footnotesize{}$\alpha_{20}$} & {\footnotesize{}$13.97$} & {\footnotesize{}$94.55$} & {\footnotesize{}$310.15$} & {\footnotesize{}$\tau_{\epsilon,20}^{2}$} & {\footnotesize{}$16.21$} & {\footnotesize{}$306.76$} & {\footnotesize{}$726.35$} & {\footnotesize{}$\mu_{20}$} & {\footnotesize{}$17.31$} & {\footnotesize{}$10.68$} & {\footnotesize{}$11.21$}\tabularnewline
$\phi$ & {\footnotesize{}$8.22$} & {\footnotesize{}$22.73$} & {\footnotesize{}$34.20$} & {\footnotesize{}$\tau_{f,1}^{2}$} & {\footnotesize{}$12.36$} & {\footnotesize{}$52.16$} & {\footnotesize{}$68.27$} &  &  &  &  &  &  &  & \tabularnewline
\hline
\end{tabular}
\end{sidewaystable}

Figures~\ref{fig:The-kernel-density alpha real data}  and \ref{fig:The-kernel-density tau real data}
present the kernel density estimates of marginal posterior densities of four representative $\alpha$ and $\tau_{\epsilon}^2$ respectively
for the Gaussian OU model for the US stock returns data. The density estimates are for PMMH+PG using exact and approximate transition densities
and PG with approximate transition densities using ancestral tracing and  backward simulation.
Both figures show that both PMMH+PG samplers
produce estimates that are close to each other, whereas the PG samplers are much less reliable.
\begin{figure}[H]
\caption{The kernel density estimates of marginal posterior densities of four representative $\alpha$
for the US stock returns data. The density estimates are for PMMH+PG using exact and approximate transition densities
and PG with  approximate transition densities using ancestral tracing and  backward simulation for the Gaussian OU model.\label{fig:The-kernel-density alpha real data}}
\centering{}\includegraphics[width=20cm,height=15cm]{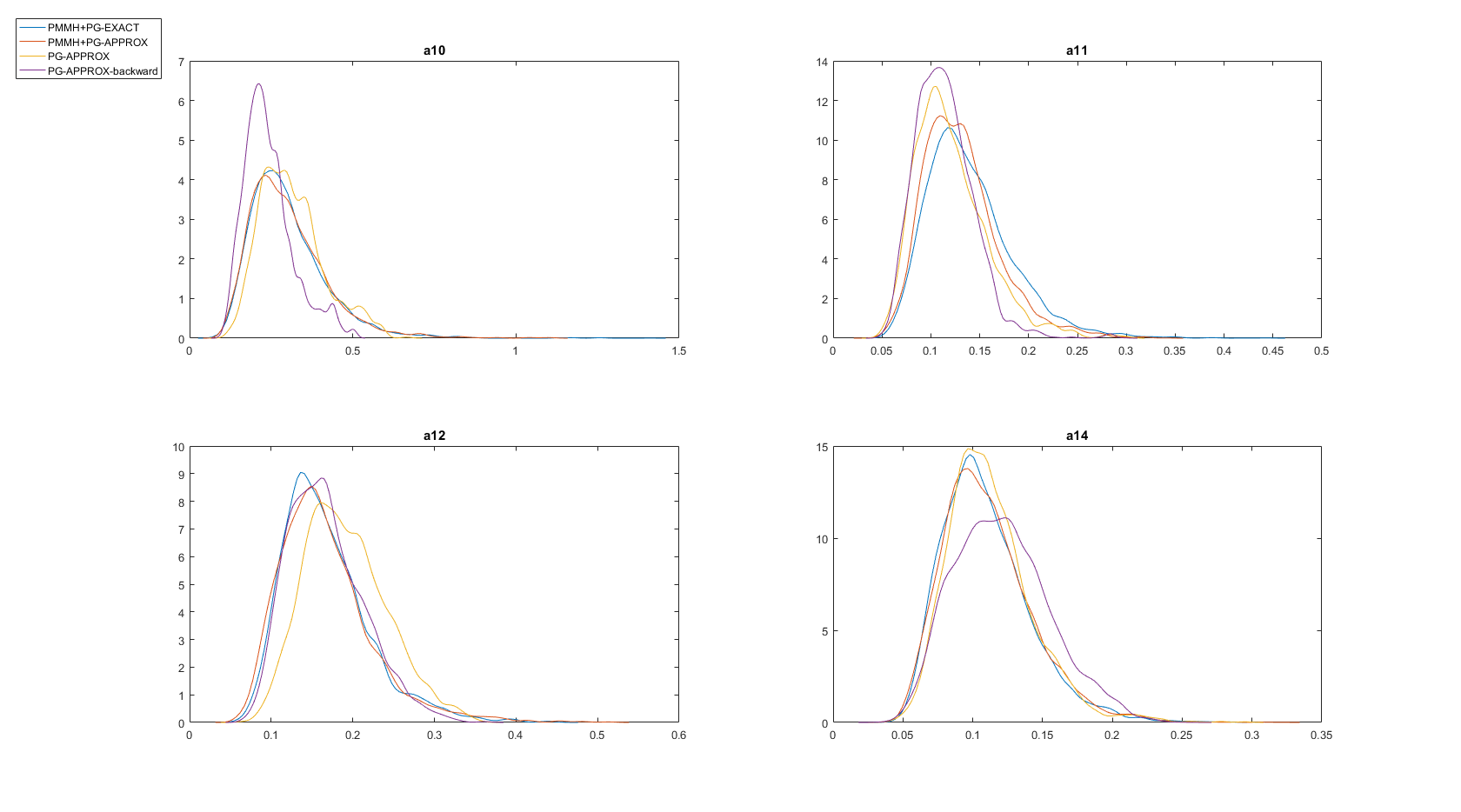}
\end{figure}

\begin{figure}[H]
\caption{The kernel density estimates of marginal posterior densities of $\boldsymbol{\tau}_{\epsilon}^{2}$
for the US stock returns data for four representative $\tau_{\epsilon}^2$. The density estimates are for PMMH+PG using exact and approximate transition densities
and PG  using ancestral tracing and  backward simulation for the Gaussian OU model.\label{fig:The-kernel-density tau real data}}
\centering{}\includegraphics[width=20cm,height=15cm]{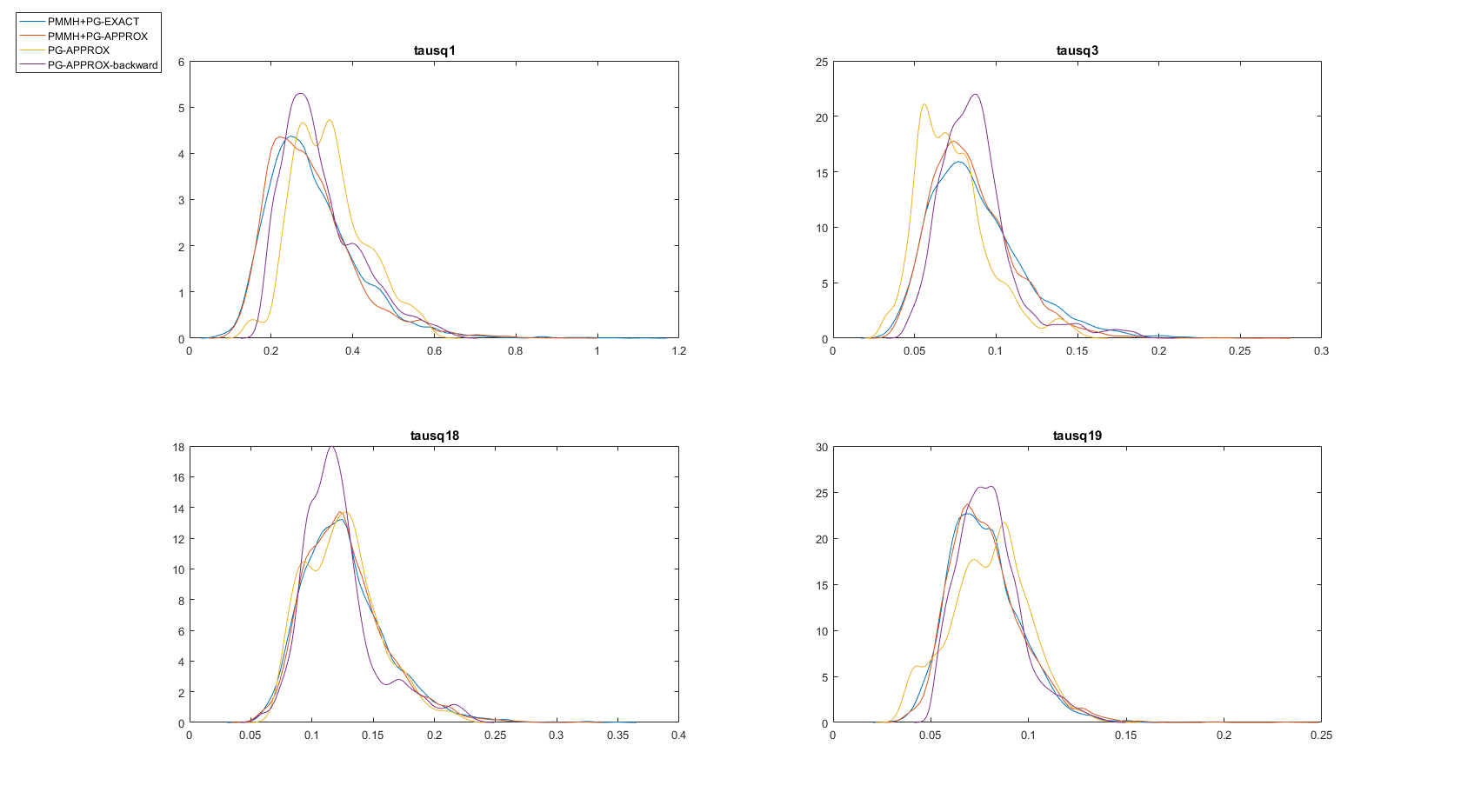}
\end{figure}

Table~\ref{tab:Inefficiency-factor-of real data-GARCH} gives the inefficiency factors of $\boldsymbol{\beta}$, $\boldsymbol{\alpha}$,
$\boldsymbol{\mu}$, $\boldsymbol{\tau}_{\epsilon}^{2}$, $\boldsymbol{\phi}$,
and $\boldsymbol{\tau}_{f}^{2}$ with the approximate Euler based transition density for the GARCH diffusion model, for the three samplers:
Sampler
I: $PMMH\left(\boldsymbol{\alpha},\boldsymbol{\tau}_{\epsilon}^{2},\boldsymbol{\tau}_{f}^{2},\boldsymbol{\mu}\right)+PG\left(\boldsymbol{f}_{1:T},\boldsymbol{\beta},\boldsymbol{\phi}\right)$,
Sampler $II$: $PGAT\left(\boldsymbol{f}_{1:T},\boldsymbol{\beta},\boldsymbol{\alpha},\boldsymbol{\tau}_{\epsilon}^{2},\boldsymbol{\mu},\boldsymbol{\phi},\boldsymbol{\tau}_{f}^{2}\right)$, Sampler $III$: $PGBS\left(\boldsymbol{f}_{1:T},\boldsymbol{\beta},\boldsymbol{\alpha},\boldsymbol{\tau}_{\epsilon}^{2},\boldsymbol{\mu},\boldsymbol{\phi},\boldsymbol{\tau}_{f}^{2}\right)$
for US stock returns data with $T=1000$,
$S=20$, and $K=1$, and with the number of particles $N=1000$.

\begin{sidewaystable}
\caption{Inefficiency factors of $\boldsymbol{\beta}$, $\boldsymbol{\alpha}$,
$\boldsymbol{\mu}$, $\boldsymbol{\tau}^{2}$, $\boldsymbol{\phi}$,
and $\boldsymbol{\tau}_{f}^{2}$ with an Euler approximation for the state
transition densities for the GARCH diffusion model: Sampler I: $PG\left(\boldsymbol{f}_{1:T},\boldsymbol{\beta},\boldsymbol{\phi}\right)+PMMH\left(\boldsymbol{\alpha},\boldsymbol{\tau}^{2},\boldsymbol{\mu},\boldsymbol{\tau}_{f}^{2}\right)$,
Sampler $II$: $PGAT\left(\boldsymbol{f}_{1:T},\boldsymbol{\beta},\boldsymbol{\alpha},\boldsymbol{\tau}^{2},\boldsymbol{\mu},\boldsymbol{\phi},\boldsymbol{\tau}_{f}^{2}\right)$,
sampler III: $PGBS\left(\boldsymbol{f}_{1:T},\boldsymbol{\beta},\boldsymbol{\alpha},\boldsymbol{\tau}^{2},\boldsymbol{\mu},\boldsymbol{\phi},\boldsymbol{\tau}_{f}^{2}\right)$
for US stock returns data with $T=1000$, $P=20$, and $K=1$, and
number of particles $N=1000$.\label{tab:Inefficiency-factor-of real data-GARCH}}

\centering{}%
\begin{tabular}{cccccccccccccccc}
\hline 
 & {\footnotesize{}I} & {\footnotesize{}II} & {\footnotesize{}III} &  & {\footnotesize{}I} & {\footnotesize{}II} & {\footnotesize{}III} &  & {\footnotesize{}I} & {\footnotesize{}II} & {\footnotesize{}III} &  & {\footnotesize{}I} & {\footnotesize{}II} & {\footnotesize{}III}\tabularnewline
\hline 
{\footnotesize{}$\beta_{1}$} & {\footnotesize{}$1.91$} & {\footnotesize{}$2.06$} & {\footnotesize{}$1.88$} & {\footnotesize{}$\alpha_{1}$} & {\footnotesize{}$32.83$} & {\footnotesize{}$197.50$} & {\footnotesize{}$318.60$} & {\footnotesize{}$\tau_{1}^{2}$} & {\footnotesize{}$48.24$} & {\footnotesize{}$1944.73$} & {\footnotesize{}$1079.01$} & {\footnotesize{}$\mu_{1}$} & {\footnotesize{}$111.02$} & {\footnotesize{}$207.54$} & {\footnotesize{}$229.44$}\tabularnewline
{\footnotesize{}$\beta_{2}$} & {\footnotesize{}$1.73$} & {\footnotesize{}$1.77$} & {\footnotesize{}$1.76$} & {\footnotesize{}$\alpha_{2}$} & {\footnotesize{}$13.31$} & {\footnotesize{}$144.70$} & {\footnotesize{}$135.32$} & {\footnotesize{}$\tau_{2}^{2}$} & {\footnotesize{}$12.39$} & {\footnotesize{}$2186.34$} & {\footnotesize{}$2205.98$} & {\footnotesize{}$\mu_{2}$} & {\footnotesize{}$14.57$} & {\footnotesize{}$227.46$} & {\footnotesize{}$130.33$}\tabularnewline
{\footnotesize{}$\beta_{3}$} & {\footnotesize{}$1.69$} & {\footnotesize{}$1.76$} & {\footnotesize{}$1.72$} & {\footnotesize{}$\alpha_{3}$} & {\footnotesize{}$13.53$} & {\footnotesize{}$179.77$} & {\footnotesize{}$186.45$} & {\footnotesize{}$\tau_{3}^{2}$} & {\footnotesize{}$23.59$} & {\footnotesize{}$1794.07$} & {\footnotesize{}$654.43$} & {\footnotesize{}$\mu_{3}$} & {\footnotesize{}$19.46$} & {\footnotesize{}$212.72$} & {\footnotesize{}$143.36$}\tabularnewline
{\footnotesize{}$\beta_{4}$} & {\footnotesize{}$1.70$} & {\footnotesize{}$1.77$} & {\footnotesize{}$1.76$} & {\footnotesize{}$\alpha_{4}$} & {\footnotesize{}$14.94$} & {\footnotesize{}$225.04$} & {\footnotesize{}$157.04$} & {\footnotesize{}$\tau_{4}^{2}$} & {\footnotesize{}$16.84$} & {\footnotesize{}$3098.27$} & {\footnotesize{}$1208.39$} & {\footnotesize{}$\mu_{4}$} & {\footnotesize{}$23.45$} & {\footnotesize{}$163.81$} & {\footnotesize{}$118.05$}\tabularnewline
{\footnotesize{}$\beta_{5}$} & {\footnotesize{}$1.71$} & {\footnotesize{}$1.74$} & {\footnotesize{}$1.71$} & {\footnotesize{}$\alpha_{5}$} & {\footnotesize{}$16.23$} & {\footnotesize{}$420.88$} & {\footnotesize{}$421.26$} & {\footnotesize{}$\tau_{5}^{2}$} & {\footnotesize{}$14.29$} & {\footnotesize{}$558.61$} & {\footnotesize{}$3257.52$} & {\footnotesize{}$\mu_{5}$} & {\footnotesize{}$19.34$} & {\footnotesize{}$322.81$} & {\footnotesize{}$243.11$}\tabularnewline
{\footnotesize{}$\beta_{6}$} & {\footnotesize{}$1.66$} & {\footnotesize{}$1.72$} & {\footnotesize{}$1.68$} & {\footnotesize{}$\alpha_{6}$} & {\footnotesize{}$18.66$} & {\footnotesize{}$875.82$} & {\footnotesize{}$1166.83$} & {\footnotesize{}$\tau_{6}^{2}$} & {\footnotesize{}$21.67$} & {\footnotesize{}$1097.21$} & {\footnotesize{}$2746.64$} & {\footnotesize{}$\mu_{6}$} & {\footnotesize{}$18.93$} & {\footnotesize{}$359.97$} & {\footnotesize{}$638.61$}\tabularnewline
{\footnotesize{}$\beta_{7}$} & {\footnotesize{}$1.64$} & {\footnotesize{}$1.73$} & {\footnotesize{}$1.66$} & {\footnotesize{}$\alpha_{7}$} & {\footnotesize{}$14.81$} & {\footnotesize{}$488.09$} & {\footnotesize{}$447.91$} & {\footnotesize{}$\tau_{7}^{2}$} & {\footnotesize{}$16.79$} & {\footnotesize{}$1932.45$} & {\footnotesize{}$2415.33$} & {\footnotesize{}$\mu_{7}$} & {\footnotesize{}$35.95$} & {\footnotesize{}$247.08$} & {\footnotesize{}$205.50$}\tabularnewline
{\footnotesize{}$\beta_{8}$} & {\footnotesize{}$1.75$} & {\footnotesize{}$1.86$} & {\footnotesize{}$1.72$} & {\footnotesize{}$\alpha_{8}$} & {\footnotesize{}$18.77$} & {\footnotesize{}$180.04$} & {\footnotesize{}$152.92$} & {\footnotesize{}$\tau_{8}^{2}$} & {\footnotesize{}$17.51$} & {\footnotesize{}$681.34$} & {\footnotesize{}$2236.32$} & {\footnotesize{}$\mu_{8}$} & {\footnotesize{}$16.08$} & {\footnotesize{}$140.56$} & {\footnotesize{}$231.74$}\tabularnewline
{\footnotesize{}$\beta_{9}$} & {\footnotesize{}$1.76$} & {\footnotesize{}$1.79$} & {\footnotesize{}$1.85$} & {\footnotesize{}$\alpha_{9}$} & {\footnotesize{}$23.51$} & {\footnotesize{}$655.71$} & {\footnotesize{}$543.04$} & {\footnotesize{}$\tau_{9}^{2}$} & {\footnotesize{}$23.17$} & {\footnotesize{}$2465.44$} & {\footnotesize{}$3065.63$} & {\footnotesize{}$\mu_{9}$} & {\footnotesize{}$147.16$} & {\footnotesize{}$434.49$} & {\footnotesize{}$814.62$}\tabularnewline
{\footnotesize{}$\beta_{10}$} & {\footnotesize{}$1.70$} & {\footnotesize{}$1.77$} & {\footnotesize{}$1.75$} & {\footnotesize{}$\alpha_{10}$} & {\footnotesize{}$13.04$} & {\footnotesize{}$1159.77$} & {\footnotesize{}$969.04$} & {\footnotesize{}$\tau_{10}^{2}$} & {\footnotesize{}$14.04$} & {\footnotesize{}$2013.82$} & {\footnotesize{}$1638.88$} & {\footnotesize{}$\mu_{10}$} & {\footnotesize{}$17.20$} & {\footnotesize{}$902.78$} & {\footnotesize{}$322.78$}\tabularnewline
{\footnotesize{}$\beta_{11}$} & {\footnotesize{}$1.69$} & {\footnotesize{}$1.77$} & {\footnotesize{}$1.74$} & {\footnotesize{}$\alpha_{11}$} & {\footnotesize{}$11.05$} & {\footnotesize{}$298.47$} & {\footnotesize{}$210.95$} & {\footnotesize{}$\tau_{11}^{2}$} & {\footnotesize{}$14.72$} & {\footnotesize{}$1224.84$} & {\footnotesize{}$2551.95$} & {\footnotesize{}$\mu_{11}$} & {\footnotesize{}$17.49$} & {\footnotesize{}$335.21$} & {\footnotesize{}$216.52$}\tabularnewline
{\footnotesize{}$\beta_{12}$} & {\footnotesize{}$1.68$} & {\footnotesize{}$1.86$} & {\footnotesize{}$1.78$} & {\footnotesize{}$\alpha_{12}$} & {\footnotesize{}$19.20$} & {\footnotesize{}$462.64$} & {\footnotesize{}$495.65$} & {\footnotesize{}$\tau_{12}^{2}$} & {\footnotesize{}$22.52$} & {\footnotesize{}$2865.97$} & {\footnotesize{}$1412.81$} & {\footnotesize{}$\mu_{12}$} & {\footnotesize{}$49.95$} & {\footnotesize{}$179.47$} & {\footnotesize{}$351.02$}\tabularnewline
{\footnotesize{}$\beta_{13}$} & {\footnotesize{}$1.78$} & {\footnotesize{}$1.85$} & {\footnotesize{}$1.86$} & {\footnotesize{}$\alpha_{13}$} & {\footnotesize{}$14.12$} & {\footnotesize{}$232.22$} & {\footnotesize{}$270.89$} & {\footnotesize{}$\tau_{13}^{2}$} & {\footnotesize{}$13.88$} & {\footnotesize{}$1646.83$} & {\footnotesize{}$2770.24$} & {\footnotesize{}$\mu_{13}$} & {\footnotesize{}$16.15$} & {\footnotesize{}$230.03$} & {\footnotesize{}$597.87$}\tabularnewline
{\footnotesize{}$\beta_{14}$} & {\footnotesize{}$1.61$} & {\footnotesize{}$1.63$} & {\footnotesize{}$1.63$} & {\footnotesize{}$\alpha_{14}$} & {\footnotesize{}$17.59$} & {\footnotesize{}$159.37$} & {\footnotesize{}$337.67$} & {\footnotesize{}$\tau_{14}^{2}$} & {\footnotesize{}$16.22$} & {\footnotesize{}$2651.10$} & {\footnotesize{}$1083.02$} & {\footnotesize{}$\mu_{14}$} & {\footnotesize{}$15.34$} & {\footnotesize{}$146.47$} & {\footnotesize{}$227.23$}\tabularnewline
{\footnotesize{}$\beta_{15}$} & {\footnotesize{}$1.54$} & {\footnotesize{}$1.54$} & {\footnotesize{}$1.52$} & {\footnotesize{}$\alpha_{15}$} & {\footnotesize{}$13.93$} & {\footnotesize{}$330.76$} & {\footnotesize{}$329.03$} & {\footnotesize{}$\tau_{15}^{2}$} & {\footnotesize{}$16.10$} & {\footnotesize{}$1551.35$} & {\footnotesize{}$1303.25$} & {\footnotesize{}$\mu_{15}$} & {\footnotesize{}$30.23$} & {\footnotesize{}$164.37$} & {\footnotesize{}$182.29$}\tabularnewline
{\footnotesize{}$\beta_{16}$} & {\footnotesize{}$1.67$} & {\footnotesize{}$1.69$} & {\footnotesize{}$1.62$} & {\footnotesize{}$\alpha_{16}$} & {\footnotesize{}$17.17$} & {\footnotesize{}$352.23$} & {\footnotesize{}$275.77$} & {\footnotesize{}$\tau_{16}^{2}$} & {\footnotesize{}$15.05$} & {\footnotesize{}$2166.59$} & {\footnotesize{}$1121.20$} & {\footnotesize{}$\mu_{16}$} & {\footnotesize{}$11.35$} & {\footnotesize{}$141.30$} & {\footnotesize{}$246.14$}\tabularnewline
{\footnotesize{}$\beta_{17}$} & {\footnotesize{}$2.04$} & {\footnotesize{}$2.16$} & {\footnotesize{}$2.03$} & {\footnotesize{}$\alpha_{17}$} & {\footnotesize{}$16.20$} & {\footnotesize{}$202.76$} & {\footnotesize{}$198.07$} & {\footnotesize{}$\tau_{17}^{2}$} & {\footnotesize{}$17.68$} & {\footnotesize{}$2007.36$} & {\footnotesize{}$3053.61$} & {\footnotesize{}$\mu_{17}$} & {\footnotesize{}$36.97$} & {\footnotesize{}$728.55$} & {\footnotesize{}$820.53$}\tabularnewline
{\footnotesize{}$\beta_{18}$} & {\footnotesize{}$1.83$} & {\footnotesize{}$1.86$} & {\footnotesize{}$1.77$} & {\footnotesize{}$\alpha_{18}$} & {\footnotesize{}$13.94$} & {\footnotesize{}$347.07$} & {\footnotesize{}$192.65$} & {\footnotesize{}$\tau_{18}^{2}$} & {\footnotesize{}$17.27$} & {\footnotesize{}$1478.12$} & {\footnotesize{}$2889.07$} & {\footnotesize{}$\mu_{18}$} & {\footnotesize{}$19.63$} & {\footnotesize{}$311.89$} & {\footnotesize{}$603.94$}\tabularnewline
{\footnotesize{}$\beta_{19}$} & {\footnotesize{}$1.74$} & {\footnotesize{}$1.80$} & {\footnotesize{}$1.78$} & {\footnotesize{}$\alpha_{19}$} & {\footnotesize{}$14.17$} & {\footnotesize{}$398.65$} & {\footnotesize{}$157.60$} & {\footnotesize{}$\tau_{19}^{2}$} & {\footnotesize{}$18.14$} & {\footnotesize{}$2896.07$} & {\footnotesize{}$2682.24$} & {\footnotesize{}$\mu_{19}$} & {\footnotesize{}$19.85$} & {\footnotesize{}$1340.20$} & {\footnotesize{}$235.55$}\tabularnewline
{\footnotesize{}$\beta_{20}$} & {\footnotesize{}$1.75$} & {\footnotesize{}$1.81$} & {\footnotesize{}$1.80$} & {\footnotesize{}$\alpha_{20}$} & {\footnotesize{}$17.59$} & {\footnotesize{}$130.58$} & {\footnotesize{}$262.31$} & {\footnotesize{}$\tau_{20}^{2}$} & {\footnotesize{}$15.98$} & {\footnotesize{}$2096.28$} & {\footnotesize{}$1352.18$} & {\footnotesize{}$\mu_{20}$} & {\footnotesize{}$17.63$} & {\footnotesize{}$119.10$} & {\footnotesize{}$148.03$}\tabularnewline
$\phi$ & {\footnotesize{}$8.74$} & {\footnotesize{}$21.52$} & {\footnotesize{}$20.91$} & {\footnotesize{}$\tau_{f_{1}}^{2}$} & {\footnotesize{}$13.65$} & {\footnotesize{}$47.47$} & {\footnotesize{}$46.10$} &  &  &  &  &  &  &  & \tabularnewline
\hline 
\end{tabular}
\end{sidewaystable}

\end{document}